\documentclass[11pt]{article}

\usepackage[english]{babel}

\usepackage{fullpage}
\usepackage{amsfonts} 
\usepackage{amsmath,amssymb}
\usepackage{graphicx}
\usepackage[colorlinks=true, allcolors=blue]{hyperref}
\usepackage{cleveref}
\usepackage{nicefrac}
\usepackage{multirow}

\usepackage{algorithm}
\usepackage{algorithmicx}
\usepackage{algpseudocode}
\algrenewcommand\algorithmicrequire{\textbf{Input:}}
\algrenewcommand\algorithmicensure{\textbf{Output:}}

\newtheorem{theorem}{Theorem}[section]
\makeatletter
\renewcommand*{\@opargbegintheorem}[3]{\trivlist
	\item[\hskip \labelsep{\bfseries #1\ #2}] \textbf{(#3)}\ \itshape}
\makeatother

\newtheorem{claim}[theorem]{Claim}
\newtheorem{remark}[theorem]{Remark}

\newtheorem{assumption}{Assumption}
\newtheorem{corollary}[theorem]{Corollary}

\newtheorem{definition}[theorem]{Definition}

\newtheorem{lemma}[theorem]{Lemma}
\newtheorem{observation}[theorem]{Observation}

\newtheorem{problem}{Problem}

\newcommand{\eps}{\varepsilon}
\renewcommand{\epsilon}{\varepsilon}

\newcommand{\eat}[1]{}

\newcommand{\R}{\mathbb{R}}

\newcommand{\poly}{\operatorname{poly}}

\newcommand{\ddim}{\ensuremath{\mathrm{ddim}}}

\newcommand{\Var}{\ensuremath{\mathrm{Var}}}

\newcommand{\calC}{\mathcal{X}^k}

\newcommand{\calF}{\mathcal{F}}
\newcommand{\calG}{\mathcal{G}}

\newcommand{\calX}{\mathcal{X}}
\newcommand{\calZ}{\mathcal{Z}}

\newcommand{\Exp}{\ensuremath{\mathbb{E}}}

\newcommand{\err}{\ensuremath{\mathsf{err}}}

\renewcommand{\eqref}[1]{(\ref{#1})}
\newcommand{\cost}{\ensuremath{\mathrm{cost}}}

\newcommand{\mN}{\ensuremath{\mathcal{N}^{(m)}}}
\newcommand{\myN}{\ensuremath{\widehat{\mathcal{N}}^{(m)}}}
\newcommand{\oN}{\ensuremath{\mathcal{N}^{(o)}}}
\newcommand{\ZGCA}[1]{\ensuremath{\cost_z(#1,C+A^\star)}}
\newcommand{\ZZGCA}[1]{\ensuremath{\cost_z(#1,C+A^\star)}}
\newcommand{\HGCB}[1]{\ensuremath{R(#1,C,\beta)}}
\newcommand{\HGC}[1]{\ensuremath{R(#1,C,0)}}

\newcommand{\ProblemName}[1]{\textsc{#1}}
\newcommand{\kzC}{\ProblemName{$(k, z)$-Clustering}}
\newcommand{\kMedian}{\ProblemName{$k$-Median}}
\newcommand{\kMeans}{\ProblemName{$k$-Means}}

\title{On Optimal Coreset Construction for Euclidean \kzC}

\author{Lingxiao Huang\thanks{
		Email: \texttt{huanglingxiao1990@126.com}
	}\\
	Nanjing University
	\and
	Jian Li\thanks{
		Email: \texttt{lapordge@gmail.com}
	}\\
	Tsinghua University
	\and
	Xuan Wu\thanks{ 
		Email: \texttt{wu3412790@gmail.com}}\\
	Tsinghua University
}

\date{}

\begin{document}
	\maketitle
	
	\begin{abstract}
		\sloppy
		Constructing small-sized coresets for various clustering problems in Euclidean spaces has attracted significant attention for the past decade.
        A central problem in the coreset literature is to understand what is the best possible coreset size for \kzC\ in Euclidean space. While there has been significant progress in the problem, there is still a gap between the state-of-the-art upper and lower bounds. For instance, the best known upper bound
        for \kMeans\ ($z=2$) is 
        $\min \{O(k^{3/2} \epsilon^{-2}),O(k \eps^{-4})\}$ [Cohen-Addad, Larsen, Saulpic, Schwiegelshohn, Sheikh-Omar. NeurIPS'22~\cite{cohenaddad2022improved}], 
        while the best known lower bound is $\Omega(k\eps^{-2})$ [Cohen-Addad, Larsen, Saulpic, Schwiegelshohn. STOC'22~\cite{cohenaddad2022towards}].
        
        In this paper, we make significant progress on both lower and upper bounds
        on the coreset size for Euclidean \kzC. For a wide range of parameters (i.e., $\eps, k$), we have a complete understanding of the optimal bound. In particular, we obtain the following results:
                
        (1) We present a new coreset lower bound $\Omega(k \eps^{-z-2})$ for Euclidean \kzC\ when $\eps \geq \Omega(k^{-\nicefrac{1}{(z+2)}})$.
        In view of the prior upper bound $\tilde{O}_z(k \eps^{-z-2})$~\cite{cohenaddad2022towards}, the bound is optimal.
        The new lower bound is surprising since $\Omega(k\eps^{-2})$ \cite{cohenaddad2022towards} is ``conjectured" to be the correct bound 
        in some recent works (see e.g., \cite{cohenaddad2022towards,cohenaddad2022improved}).
        %
        Our new lower bound instance is a delicate construction with multiple clusters of points, which differs significantly from the previous construction in~\cite{cohenaddad2022towards} that contains a single cluster of points.
        The new lower bound also implies improved lower bounds for \kzC\ in doubling metrics. 
        %
        
		(2) For the upper bound, we provide efficient coreset construction algorithms for Euclidean \kzC\ with improved coreset sizes.
		In particular, we provide an $\tilde{O}_z(k^{\frac{2z+2}{z+2}} \eps^{-2})$-sized coreset, with
        a unified analysis, for \kzC\ for all $z\geq 1$ in Euclidean space.
		This upper bound improves upon the $\tilde{O}_z(k^2\eps^{-2})$ upper bound by~\cite{cohenaddad2022towards} (when $k\leq \eps^{-1}$), and matches the recent results~\cite{cohenaddad2022improved} for \kMedian\ and \kMeans\ ($z=1,2$) and extends them to all $z\geq 1$.
	\end{abstract}
	
	\newpage
	\tableofcontents
	\newpage

	\section{Introduction}
	\label{sec:intro}
	
	We study the problem of constructing small-sized coresets for the classic \kzC\ problem,
	where $z\geq 1$ is a given constant.
	The \kzC\ problem is defined as follows.
	
	\vspace{-0.2cm}
	\paragraph{\kzC.} 
	In the \kzC\ problem, the input consists of a metric space $(\calX,d)$, where $\calX$ is a ground set (continuous or discrete) and $d: \calX\times \calX\rightarrow \R_{\geq 0}$ is a distance function, and a dataset $P\subseteq \calX$ of $n$ points. 
	The goal is to find a set $C \subseteq \calX$ of $k$ points, called \emph{center set}, 
	that minimizes the objective function
	\begin{equation} \label{eq:DefCost}
	\cost_z(P, C) := \sum_{x \in P}{d^z(x, C)},
	\end{equation}
	where 
	$
	d(x, C):=\min\left\{d(x,c): c\in C\right\}
	$
	is the distance from $x$ to center set $C$ and
	$d^z$ denotes the distance raised to power $z\ge 1$.
	
	This formulation captures several classical clustering problems,
	including the well studied \kMedian\ (when $z = 1$) and \kMeans\ (when $z=2$).
	\textcolor{black}{We mainly focus on the problem under Euclidean metric $\calX=\R^d$ and more generally metrics with bounded doubling dimension.}
	The \kzC\ problem has numerous applications in a variety of domains, including data analysis, approximation algorithms, unsupervised learning and computational geometry (see e.g.,~\cite{lloyd1982least,tan2006cluster,arthur2007k,coates2012learning}).

	\vspace{-0.2cm}
	\paragraph{Coresets.} 
	%
	Motivated by the ever increasing volume of data, a powerful data-reduction technique, called \emph{coresets}, has been developed for harnessing large datasets for various problems~\cite{harpeled2004on,feldman2011unified,feldman2013turning}.
	Roughly speaking, for an optimization problem, a coreset is a small-sized subset of (weighted) data points that can be used to compute an approximation of the optimization objective for every possible solution.
	In the context of \kzC, we use $\calC$ to denote the collection of all solutions, i.e., all
	$k$-center sets in $\calX$.
	A coreset for \kzC\ is formally defined as follows.

	\begin{definition}[\bf Coreset~\cite{langberg2010universal,feldman2011unified}]
		\label{def:coreset}
		Given a metric space $(\calX,d)$ together with a dataset $P\subseteq \calX$ of $n$ points, an $\eps$-coreset for the \kzC\ problem is a weighted subset $S \subseteq P$ with weight $w : S \to \R_{\geq 0}$, such that 
		\begin{equation} \label{eq:DefCoreset}
		\forall C \in \calC,
		\qquad
		\sum_{x \in S}{w(x) \cdot d^z(x, C)}
		\in (1 \pm \eps) \cdot \cost_z(P, C).
		\end{equation}
	\end{definition}
	
	\noindent
	From the above definition, one can see that if we run any existing approximation or exact algorithm on the coreset instead of the full dataset, the resulting solution provides almost the same performance guarantee in terms of the clustering objective, but the running time can be much smaller.
	Hence, constructing small-sized coresets for \kzC\ has been an important research topic and studied extensively	in the literature, for various metric spaces including Euclidean metrics~\cite{feldman2011unified,feldman2013turning,huang2020coresets,cohenaddad2021new,cohenaddad2022towards}, doubling metrics~\cite{huang2018epsilon,cohenaddad2021new,cohenaddad2022towards}, shortest path metrics in graphs~\cite{baker2020coresets,braverman2021coresets,cohenaddad2021new} and general metrics~\cite{feldman2011unified,cohenaddad2021new,cohenaddad2022towards}.
	%
	
	%
     There is a long line of work attempting to obtain tight 
     upper and lower bounds of the optimal coreset size
     (e.g., \cite{harpeled2004on,chen2009on,feldman2013turning,sohler2018strong,huang2020coresets,cohenaddad2021improved,cohenaddad2022towards}).
     Despite the substantial effort, there is still a gap between the current best-known upper 
     and lower bounds for most metric spaces studied in the literature, specifically for the Euclidean space.
     As mentioned in~\cite{cohenaddad2021improved}, finding the tight coreset size bound for Euclidean space is ``the most important open problem in coresets''. 

	Now we briefly mention some of the existing results.
	Please refer to Table~\ref{tab:result} for the best-known upper and lower bounds,
	and  \cite[Table 1]{cohenaddad2022towards} for many earlier results.
    The best known upper bound of the coreset size of \kzC\ in Euclidean space is $\tilde{O}_z(\min\left\{k^2\eps^{-2}, k\eps^{-z-2}\right\})$ 
    for all $z\geq 1$ and $k\leq \eps^{-1}$, by    
    Cohen-Addad et al. \cite{cohenaddad2022towards}. 
    In the same work, they obtained a ``natural-looking" lower bound of $\Omega(k\eps^{-2})$ for Euclidean \kzC,
    and the worst case instance achieving this bound is simple and clean
    (the points form a simplex).
    Very recently, for \kMedian\ ($z=1$) and \kMeans\ $(z=2)$ 
    the upper bounds \textcolor{black}{(term $k^2 \eps^{-2}$)} are further improved to
    $\tilde{O}(k^{4/3} \eps^{-2})$ and $\tilde{O}(k^{3/2} \eps^{-2})$ respectively~\cite{cohenaddad2022improved}. The new improved bounds are 
    arguably ``less natural", which seem to suggest
    that the upper bound side can be further improved.
%
            %
    Hence, Cohen-addad et al.~\cite{cohenaddad2022improved} 
    stated in their paper that there is ``stronger evidence that the $\Theta(k \eps^{-2})$ bound will be the correct answer'', while they believed that ``further ideas will be necessary to reach the (conjectured) optimal bound of $\Theta(k \eps^{-2})$'' \cite{cohenaddad2022towards}.
    The above discussion motivates the following natural problem.

    \begin{problem}
    \label{problem:coreset}
        What are the optimal coreset size bounds for \kzC\ in
        Euclidean metrics? Specifically, is $\Theta(k \eps^{-2})$ the optimal bound?
    \end{problem}

    \noindent

\newcommand{\specialcell}[2][c]{%
  \begin{tabular}[#1]{@{}c@{}}#2\end{tabular}}
  
	\begin{table}[t]
		\centering
		\scriptsize
		
		\begin{tabular}{|cc|c|c|}
			\hline
			\multicolumn{2}{|c|}{Metric space}  & Prior results & Our results \\ \hline
			\multicolumn{1}{|c|}{\multirow{3}{*}{Euclidean metric}} & Upper bound &  \specialcell{$\tilde{O}(\min\left\{ k^2\eps^{-2}, k\eps^{-z-2}\right\})$~\cite{cohenaddad2022towards} \\ $\tilde{O}(k^{4/3}\eps^{-2})$ ($z=1$),  $\tilde{O}(k^{3/2}\eps^{-2})$ ($z=2$)~\cite{cohenaddad2022improved}} &  $\tilde{O}( k^{\frac{2z+2}{z+2}} \eps^{-2})$  
			\\ \cline{2-4}
                \multicolumn{1}{|c|}{} & Lower bound & $\Omega(k\eps^{-2} + k 2^{\frac{z}{20}})$~\cite{cohenaddad2022towards,huang2020coresets} & \specialcell{$\Omega(k \eps^{-z-2})$ for $\eps = \Omega(k^{-\nicefrac{1}{(z+2)}})$ \\ $\Omega(k^{\frac{2z+3}{z+2}} \eps^{-1}) $ for $\eps = \Omega(k^{-\frac{z+1}{z+2}})$}
                \\ \hline
			\multicolumn{1}{|c|}{\multirow{3}{*}{Doubling metric}} & Upper bound & $\tilde{O}(k\cdot \ddim\cdot \min\left\{k\eps^{-2},\eps^{-\max\{z,2\}}\right\} )$~\cite{cohenaddad2021new} &  $/$  
			\\ \cline{2-4}
                \multicolumn{1}{|c|}{} & Lower bound & $\Omega(k\cdot \ddim\cdot \eps^{-2})$~\cite{cohenaddad2022towards} & \specialcell{$\Omega(k \cdot \ddim\cdot \eps^{-\max\{z,2\}}/\log k)$ \\ for $\eps = \Omega(k^{-\nicefrac{1}{(z+2)}})$ \\ $\Omega(k^{\frac{2z+1}{z+2}} \ddim\cdot \eps^{-1})$ for $\eps = \Omega(k^{-\frac{z+1}{z+2}})$} \\ \hline
		\end{tabular}
		\caption{
			Comparison of the state-of-the-art coreset sizes and our results for \kzC.
			We assume that $z\geq 1$ is a constant and ignore $2^{O(z)}$ or $z^{O(z)}$ factors in the coreset size.
			``$\ddim$'' denotes the doubling dimension.
			\textcolor{black}{Note that our lower bound for doubling metrics is 
            a corollary of the lower bound for
            Euclidean case and
            is nearly optimal for $\eps = \Omega(k^{-\nicefrac{1}{(z+2)}})$ (up to poly-log factors).}
		}
		\label{tab:result}
	\end{table}

	%
	%
        %
        
	
		%

	\subsection{Our Contributions}
	\label{sec:contribution}

    In this paper, we attempt to provide an answer to  Problem~\ref{problem:coreset} and make progress on 
    both the lower and upper bounds.
    Our new bounds are summarized in Table~\ref{tab:result}
    and Figure~\ref{fig:curve}.
    
    %
    \vspace{0.2cm}
    \noindent
    {\bf Lower bounds:}
    Our first main contribution is an optimal coreset lower bound for Euclidean \kzC\ for a large range of parameter $\eps$.
    In particular, our lower bound implies that the previously conjectured bound $\Theta(k \eps^{-2})$ is not the right answer.
    Our lower bound for Euclidean space also implies improved lower bound for doubling metrics.
    %
	%
    %
    We first state our lower bounds for Euclidean spaces.

        \begin{figure}
		\centering
		\includegraphics[width=0.85\textwidth]{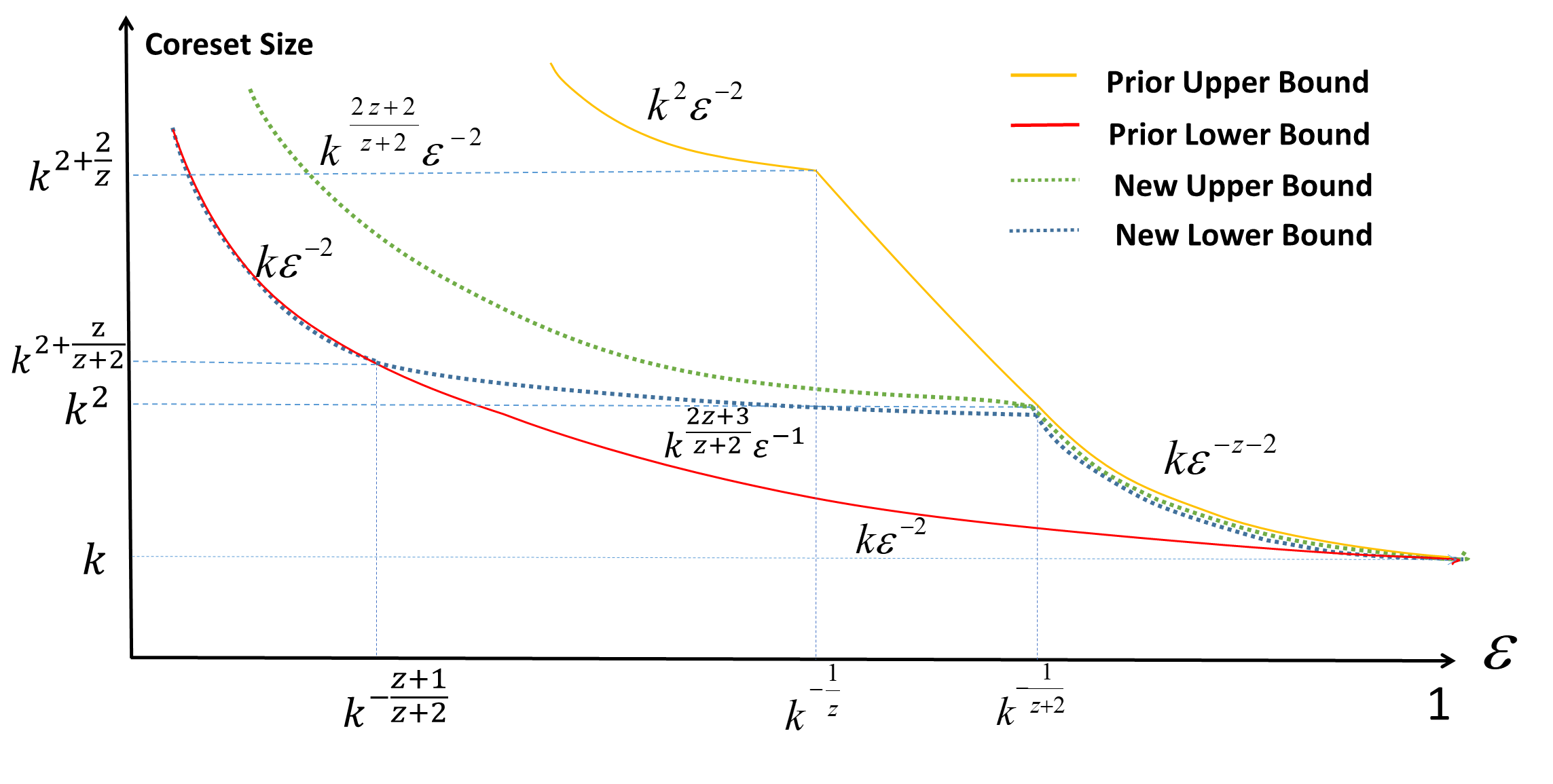}
		\caption{Comparison of prior coreset bounds and the new ones for Euclidean \kzC. Note that this figure is based on the case that $z>2$ in which $k^{-1/z}> k^{-0.5}$.}
		\label{fig:curve}
	\end{figure}

\begin{theorem}[\bf{Coreset lower bound for Euclidean \kzC}]
    \label{thm:lb}
    For any integer $k\geq 1$, constant $z\geq 1$, and error parameter $\eps \geq \Omega(k^{-\nicefrac{1}{(z+2)}})$,
    there exists an instance $P$ in Euclidean metric 
    $\mathbb{R}^d$ (for $d = O(k \eps^z)$)
    such that any $\eps$-coreset of $P$ is of size at least 
    $\Omega(k \eps^{-z-2})$.

    Moreover, for $\eps\geq \Omega(k^{-\frac{z+1}{z+2}})$, there exists an instance $P$ in Euclidean metric 
    $\mathbb{R}^d$ (for $d = O(k^{\frac{2}{z+2}})$) such that any $\eps$-coreset of $P$ is of size at least $\Omega(k^{\frac{2z+3}{z+2}} \eps^{-1})$. 
\end{theorem}

\noindent
The proof of the first lower bound can be found in Section~\ref{subsec:lbproof1}
and the second in Section~\ref{subsec:lbproof2}.
Compared to previous work~\cite{cohenaddad2022towards}, our first coreset lower bound $\Omega(k \eps^{-z-2})$ contains an additional factor $\eps^{-z}$ and matches the known upper bound~\cite{cohenaddad2022towards} when $\eps \geq \Omega(k^{-\nicefrac{1}{(z+2)}})$.
This result is surprising, as the tight dependence on $\eps$ is beyond quadratic, even for Euclidean \kMedian\ or \kMeans.
%
%
%
%
Since the coreset lower bound is monotonically non-increasing as $\eps$ increases, we directly have a lower bound of $\Omega(k^2)$ when $\eps\leq O(k^{-\nicefrac{1}{(z+2)}})$. 
By modifying the construction in Theorem~\ref{thm:lb},
we also obtain a better lower bound of $\Omega(k^{(2z+3)/(z+2)} \eps^{-1})$ 
(better than $\Omega(k^2)$ and the previous $k \eps^{-2}$ bound in \cite{cohenaddad2022towards}) for $O(k^{-\nicefrac{1}{(z+2)}})\leq \eps \leq \Omega(k^{-(z+1)/(z+2)})$.
See Figure~\ref{fig:curve}.

\color{black}



    Theorem~\ref{thm:lb} also implies a tight lower bound for doubling metrics
    for a wide range of parameters.
    In particular, noting the well known fact that the doubling dimension of Euclidean metric $\R^d$ is at most $O(d)$  (e.g., \cite{vergergaugry2005covering}), one can easily derive the following corollary via terminal embedding.

    \begin{corollary}[\bf{Coreset lower bound for \kzC\ in doubling metrics}]
    \label{cor:doubling}
    For any integer $k\geq 1$, constant $z\geq 1$, and error parameter $\eps \geq \Omega(k^{-\nicefrac{1}{(z+2)}})$,
    there exists an instance $P$ in a metric space $(\calX, d)$ with doubling dimension $\ddim$ such that any $\eps$-coreset of $P$ is of size at least 
    $\Omega(k\cdot \ddim \cdot \eps^{-\max\{z,2\}}/ \log k)$.

    Moreover, for $\eps\geq \Omega(k^{-\frac{z+1}{z+2}})$, there exists an instance $P$ in a metric space $(\calX, d)$ with doubling dimension $\ddim$ such that any $\eps$-coreset of $P$ is of size at least 
    $\Omega(k^{\frac{2z+1}{z+2}}\cdot \ddim \cdot \eps^{-1})$.
    \end{corollary}

\noindent
The proof can be found in Section~\ref{sec:doubling}.
We recall that the known upper bound is $\tilde{O}(k\cdot \ddim \cdot \min\left\{\eps^{-\max\{z,2\}}, k \eps^{-2} \right\})$~\cite{cohenaddad2021new}.
The first coreset lower bound $\Omega(k\cdot \ddim \cdot \eps^{-\max\{z,2\}}/ \log k)$ matches this upper bound when $\eps \geq \Omega(k^{-\nicefrac{1}{(z+2)}})$ (up to poly-log factors), which is near optimal.
The second bound $\Omega(k^{\frac{2z+1}{z+2}}\cdot \ddim \cdot \eps^{-1})$ is also better than the previous $\Omega(k\cdot \ddim\cdot \eps^{-2})$~\cite{cohenaddad2022towards} when $\eps = \Omega(k^{-\frac{z+1}{z+2}})$.

\color{black}

\vspace{0.2cm}
\noindent
{\bf Upper bounds: }
For the upper bound, we adopt the existing importance sampling-based coreset framework (Algorithm~\ref{alg:coreset}), proposed in~\cite{cohenaddad2021new}.
    Our contribution is a unified and tighter analysis that leads to optimal or improved coreset size upper bounds for Euclidean \kzC.
 
	\begin{theorem}[\bf{Improved upper bound for Euclidean \kzC; see also Theorem~\ref{thm:Euclidean}}]
		\label{thm:Euclidean_informal}
		Given a finite set $P\subset \R^d$ of $n$ points, there exists a randomized algorithm that constructs an $\eps$-coreset of $P$ of size $2^{O(z)}\cdot \tilde{O}(k^{\frac{2z+2}{z+2}} \eps^{-2})$ for \kzC.
		Provided an $O(1)$-approximate solution (the center set) for $P$,
		the algorithm runs in $O(nk)$ time.
	\end{theorem}
	
	\noindent
	%
	%
	%
        For general $z\geq 1$ and $\eps = O(k^{-\nicefrac{1}{(z+2)}})$, our upper bound improves on the $O(k^2\eps^{-2})$ bound in~\cite{cohenaddad2022towards}.
        Recently, Cohen-Addad et al.~\cite{cohenaddad2022improved}
        obtained a coreset size bound $\tilde{O}(k^{3/2}\eps^{-2})$ for Euclidean \kMeans\ and $\tilde{O}(k^{4/3}\eps^{-2})$ for Euclidean \kMedian. 
        Our results match theirs for $z=1$ and $z=2$ and can be seen as a generalization of their results to all powers $z\geq 1$. 
        Their analysis also tries to tighten the variance bound. 
        But our details are quite different from theirs. In particular, we use a different and unified relative error bound (for defining the covering) that works for all $z\geq 1$, and we bound the covering number differently.
    Note that Cohen-Addad et al.~\cite{cohenaddad2022improved} used
    different error bounds for $k=1$ and $k=2$.
    \footnote{
    In their Arxiv version, they also defined 
    an error bound for all $z\geq 1$. 
    But it is unclear
    to us how to use this error bound to obtain the same results for general $z\geq 1$ as ours.
    }
    For bounding the covering numbers, We discover a new geometric observation that the covering number of a cluster of points is proportional to the square of the inverse distance between the center of this cluster and center sets.
    This cannot be derived from \cite{cohenaddad2022improved} since their covering number does not relate to such distances. This observation is based on a new notion of dimension reduction, called \emph{additive terminal embedding}, which may be of independent interest.
        
        It is also worth noting that 
        for $\eps = \Theta(k^{-\nicefrac{1}{(z+2)}})$, our upper bound 
        is $\tilde{O}(k^{\frac{2z+2}{z+2}} \eps^{-2}) = \tilde{O}(k^2)$
        which matches the lower bound $\Omega(k \eps^{-z-2}) = \Omega(k^2)$ of Theorem~\ref{thm:Euclidean_informal}.
        Although from the current lower bound, we can infer that 
        the new upper bound is optimal only at a single point $\eps = \Theta(k^{-\nicefrac{1}{(z+2)}})$, curiously, it reveals that the curve corresponding to the optimal coreset bound should be piecewise and $\eps = \Theta(k^{-\nicefrac{1}{(z+2)}})$ is a breaking point of the curve. See Figure~\ref{fig:curve}.
        %
        %
	Further closing the gap between our upper bound and the lower  bound in Theorem~\ref{thm:lb} for $\eps = O(k^{-\nicefrac{1}{(z+2)}})$ is an intriguing
    and important open question.

	\subsection{Technical Overview}
	\label{sec:overview}
 
\color{black}
 \subsubsection{Lower bounds}

    For convenience, we present an overview of our lower bound for Euclidean \kMeans\ ($z=2$).
    %
    %

    \vspace{-0.3cm}
    \paragraph{Prior lower bound instances~\cite{cohenaddad2022towards}.}
    We first briefly review the construction for the best-known 
    lower bound $\Omega(k \eps^{-2})$ in~\cite{cohenaddad2022towards}.
    Their instance is simply the set of unit basis vectors $X = \left\{e_i: i\in [k \eps^{-2}]\right\}\subset \R^{k \eps^{-2}}$.
    To prove the lower bound, we assume by contradiction there exists an $\eps$-coreset $S\subseteq X$ of size at most $|X|/2$.
    We arbitrarily partition $S$ into $k$ groups, each containing 
    roughly equal number of points ($\approx |S|/k$).
    W.l.o.g., assume one group is $e_1, \ldots, e_{|S|/k}$ and we construct a center $c = \sqrt{k/|S|}\sum_{i\in [|S|/k]} e_i$
    for this group.
    We can observe that
    $
    d^2(e_1, c) = \cdots = d^2(e_{|S|/k}, c) = 1 - \sqrt{k/|S|},
    $
    but for every $i > |S|/k$, $d^2(e_i, c) = 1 + \sqrt{k/|S|}$.
    We repeat the above construction for each group, and obtain a $k$-center set $C\subset \R^d$, which satisfies that 1) $d^2(p,C) = 1 - \sqrt{k/|S|}$ for every $p\in S$ and 
    2) $d^2(q,C) = 1 + \sqrt{k/|S|}$ for every $q\in X\setminus S$.
    Noting that $d^2(p,C)< (1-\eps)d^2(q,C)$ for any
    $p\in S$ and $q\in X\setminus S$, 
    we can conclude that $\cost_2(S,C) < (1-\eps)\cdot \cost_2(X,C)$, 
    and hence $S$ is not an $\eps$-coreset of $X$.
    \vspace{-0.2cm}
    \paragraph{A failed attempt.}
    %
    A natural idea to improve the lower bound 
    is to construct a dataset $P$ that consists of $m\leq k/2$ ``independent" copies of the above instance $X$ such that they are far away from each other, and attempt to argue that a constant proportion of points in every copy must be selected to form a coreset (see e.g.,~\cite{baker2020coresets}).
    Again, by contradiction we assume there exists an $\eps$-coreset $S\subseteq P$ of size at most $|P|/2$.

    We can assume without loss of generality that $w(S\cap X_i)\approx |X_i|$ holds for every copy $X_i$ of $X$.
    Otherwise, we can construct a center set $C\subset \R^d$ that contains the origin of every copy {\em except $X_i$}, and the remaining centers are far away from all points in $P$.
    In this case, the clustering cost is mainly from $X_i$ and $\frac{\cost_2(S,C)}{\cost_2(X,C)}\approx \frac{w(S\cap X_i)}{|X_i|}\notin [1-\eps, 1+\eps]$, which contradicts the assumption that $S$ is a coreset.
    Next, there must exist a copy $X_i$ of $X$ such that $|S\cap X_i| \leq |X_i|/2$.
    %
    We may construct a $k$-center set $C$ as follows: first, construct 
    $k - m + 1$ ($\geq k/2$) centers for this copy $X_i$ such that the distance gap between $d(p,C)$ for $p\in S\cap X_i$ and 
    $d(q,C)$ for $q\in X_i\setminus S$ is as large as possible; 
    we also place a center at the origin of every remaining copy of $X$ except $X_i$.

    Recall $|X_i| = k \eps^{-2}$.
    By applying the same construction of center set in~\cite{cohenaddad2022towards}, we can construct $k - m + 1$ centers in $C$ for $X_i$ such that
    the induced error on $X_i$ is     \[
    \left|\cost_2(X_i, C) - \cost_2(S\cap X_i, C)\right| = \Omega(\sqrt{k/|S\cap X_i|}\cdot k \eps^{-2}) = \Omega(k \eps^{-1}). 
    \]
    On the other hand, we can see that
    $\cost_2(P,C) = \Theta(m k \eps^{-2}) \gg \cost_2(X_i,C) = k \eps^{-2}$.
    Hence, the induced error $\Omega(k \eps^{-1})$ on $X_i$ is of order 
    much smaller than $\eps\cdot \cost_2(P,C)$,
    and we cannot conclude that $S$ is not an $\epsilon$-coreset.

    Another attempt is to construct $k - m + 1$ centers for $X_i$ such that
    \begin{enumerate}
        \item[P1.] $\cost_2(X_i, C) = \Theta(m k \eps^{-2})$, i.e., ensuring $\cost_2(X_i,C)$ and $\cost_2(P\setminus X_i,C)$ are in the same order;
        \item[P2.] $\left|\cost_2(X_i, C) - \cost_2(S\cap X_i, C)\right| = \Theta(m \sqrt{k/|S\cap X_i|}\cdot k \eps^{-2}) = \Theta(m k \eps^{-1})$, i.e., the induced error on $X_i$ is proportional to that of $\eps\cdot\cost_2(X_i,C)$. 
    \end{enumerate}
    Unfortunately, under P1, all points in $X_i$ are relatively close to each other (with distance $O(1)$), and hence there is no center set $C$ that can make $|\cost_2(X_i, C) - \cost_2(S\cap X_i, C)|$ large enough.
    %

    \vspace{-0.2cm}
    \paragraph{Our lower bound instance.}
    Motivated by the above attempts, we aim to construct an instance $X$ such that both P1 and P2 hold. 
    %
    Our key idea is to construct multiple clusters and ensure that each center can affect $\eps^{-2}$ points {\em from different clusters}
    (instead from one cluster as in \cite{cohenaddad2022towards})!
    For convenience, we discuss the case of $\eps = \Theta(k^{-1/4})$ and
    describe the construction of an instance $P$ for which the coreset size is at least 
    $\Omega(k \eps^{-4}) = \Omega(k^2)$.
    %
    %

    We first construct a point set $X$ and the instance $P$ consists of $m = \Theta(\eps^{-2}) = \Theta(k^{1/2})$ copies of $X$ that are far apart.
    Let $X$ consist of $B = 100 k^{1/2}$ clusters where the $j$-th cluster is centered at a scaled basis vector $a^\star_j = \eps^{-1} e_j\in \R^B$.
    Next, we carefully select $k/2$ unit vectors $v_1,\ldots, v_{k/2}\in \R^B$ that are almost orthogonal to each other.
    Every point in $X$ is of the form $a^\star_j + v_i$ for some $j\in [B]$ and $i\in [k/2]$.
    There are total $m B=k/2$ clusters (for all copies). 
    We ensure that every $v_i$ uniformly appears in $0.5 k^{1/2}$ clusters.
    Now let $S\subseteq X$ be any subset of size at most $|X|/2$.
    We show it is possible to create a center $c_i$ for every $i\in [k/2]$ which affects all points in different clusters with the same offset $v_i$ (\Cref{claim:property}).
    More concretely, for every $p = a^\star_j + v_i\in X$,
    \begin{enumerate}
        \item $c_i$ is the closest center of $p$ among all created centers $\left\{c_1,\ldots,c_{k/2}\right\}$;
        \item For any $p\in S$ and $q\in X\setminus S$, $d^2(p,c_i)-d^2(q,c_i) = \Theta(\eps^{-1})$;
    \end{enumerate}
    Consequently, the total induced error of these centers is
    \[
    \left|\cost_2(X, C) - \cost_2(S, C)\right| = \Theta(\eps^{-1} |X|) = \Theta(k^{7/4}),
    \]
    while $\cost_2(X,C) = \Theta(\eps^{-2} |X|) = \Theta(k^2)$.
    Hence both P1 and P2 hold, which implies that $S$ is not an $\eps$-coreset.

\color{black}
 \subsubsection{Upper Bounds}
 
	Next, we provide a high level overview for the analysis of
 our upper bound in Theorem~\ref{thm:Euclidean_informal}.
	Our algorithm is based on a recent importance sampling framework (Algorithm~\ref{alg:coreset}), developed in~\cite{cohenaddad2021new}.
	%
	%
	%
	In the analysis, it suffices to bound the number of samples 
    $\Gamma_G$ for each group $G$, which consists of at most $k$ rings, each having similar cost.
	%
	%
	$\Gamma_G$ should be large enough to ensure the uniform convergence of $|\sum_{p\in S_G} w(p)\cdot d_z(p,C)-\cost_z(G,C)|$ for all center sets $C$ (see Lemma~\ref{lm:key}). 
    To show the uniform convergence bound,
    \cite{cohenaddad2022towards} adopted a chaining argument 
    for tuples $u^C = (d(p,C))_{p\in G}$ over all center sets $C\in \calX^k$ in Euclidean metrics $\R^d$.
    Roughly speaking, they construct $2^{-h}$-coverings $V_h \subset \R^{|G|}_{\geq 0}$ at different scales $h\geq 0$, and bound the covering number $|V_h|$ and the maximum variance of $\Var_{C,h}$ for each scale.
    Their method leads to an upper bound of the sample number $\Gamma_G$ (ignoring polylog factor) being at most
    $k\eps^{-2}\cdot \min\left\{\eps^{-z},k\right\}$.
	The unexpected term $\min\left\{\eps^{-z},k\right\}$ appears due to the upper bound for the variance 
    $\Var_{C,h}\approx 2^{-2h}\cdot \min\left\{\eps^{-z},k\right\}$.
	
    \vspace{-0.4cm}
	\paragraph{A new chaining argument.}
    Our main improvement is to apply the chaining argument for $y^C = (d^z(p,C) - d^z(a^\star_i,C))_{i,p}$ instead of 
    $u^C = (d(p,C))_{p\in G}$ (Lemma~\ref{lm:main_small_y}) and construct coverings $V_h$s for tuples $y^C$ instead of $u^C$.
	The benefit of introducing $y^C$ is that $|d^z(p,C) - d^z(a^\star_i,C)| \leq O(z)\cdot d(p,a^\star_i)\cdot (d^{z-1}(p,C)+ d^{z-1}(a^\star_i,C))$ by 
	triangle inequality, which can be much smaller than $d^z(p,C)$.
	Moreover, considering the space of $y^C$ enables us to design a {\em tighter and unified relative covering error} $\err(p,C)$ for all powers $z\geq 1$ (see Definition~\ref{def:covering_main}) than the previous one defined in~\cite{cohenaddad2022towards} (which is $d^z(p,C) + d^z(p,A^\star)$), and the new relative covering error is key to achieving a smaller variance upper bound (Lemma~\ref{lm:variance}) for the chaining argument.
	Specifically, our variance $\Var_{C,h}$ is exactly proportional to $2^{2h}$
	and avoids the additional term $\min\left\{\eps^{-z},k\right\}$.
	%
	\vspace{-0.4cm}
	\paragraph{Bounding the covering number.}
    The remaining issue is to bound the covering number $|V_h|$s with respect to the new relative covering error $\err(p,C)$.
	To bound the size of $2^{-h}$-coverings $V_h$, 
    we introduce a new dimension reduction notion, called 
	{\em additive terminal embedding} (see Definition~\ref{def:additive_terminal_embedding}), 
	which embeds a ring of radius $r>0$ to a low dimensional space, and the required distance distortion should be an additive error $2^{-h} r$ (proportional to the radius), instead of 
	a multiplicative error as in the ordinary terminal embedding.
	Using the approach developed in \cite{mahabadi2018nonlinear,Narayanan2019OptimalTD}, 
	we can prove that the target dimension of additive terminal embedding can be bounded by $\tilde{O}(2^{-2h})$ (Theorem~\ref{thm:additive_terminal}).
    The new dimension reduction technique allows us to further reduce the target dimensions, which leads to better covering number $\exp(\tilde{O}(k^{\frac{2z+2}{z+2}} 2^{2h}))$.
	Consequently, the required sample number $\Gamma_G$ is at most $\tilde{O}(k^{\frac{2z+2}{z+2}} \eps^{-2})$ (see Theorem~\ref{thm:Euclidean_informal}).
	The notion of additive terminal embedding might be useful in 
    other contexts.

	\subsection{Other Related Work}
	\label{sec:related}
	
	\paragraph{Coresets for variants of clustering.} Coresets for several variants of clustering problems have also been studied.
	Cohen-Addad and Li~\cite{cohen2019on} first constructed a coreset of size $\tilde{O}(k^2\eps^{-3} \log^2 n )$ for capacitated \kMedian\ in $\R^d$, and the coreset size was improved to $\tilde{O}(k^3\eps^{-6})$ by~\cite{braverman2022power}.
	A generalization of capacitated clustering, called fair clustering, has also been shown to admit coresets of small size~\cite{schmidt2019fair,huang2019coresets,braverman2022power}.
	Another important variant of clustering is robust clustering, in which we can exclude at most $m$ points as outliers from the clustering objective.
	Recently, Huang et al.~\cite{huang2022near} provided a coreset construction for robust \kzC\ in Euclidean spaces of size $m+\poly(k,\eps^{-1})$.
	Other variants of clustering that admit small-sized coresets include fuzzy clustering~\cite{blomer2018coresets}, ordered weighted clustering~\cite{braverman2019coresets}, and time-series clustering~\cite{huang2021coresets}.

	\paragraph{Coresets for other problems.} Coresets have also been applied to a wide range of optimization and machine learning problems, including regression~\cite{drineas2006sampling,li2013iterative,boutsidis2013near,cohen2015uniform,jubran2019fast,chhaya2020coresets}, low rank approximation~\cite{cohen2017input}, projective clustering~\cite{feldman2011unified,tukan2022new}, and mixture model~\cite{lucic2017training,huang2020coresetsFR}.

\section{Optimal Coreset Lower Bounds}
\label{sec:lower}

\subsection{Proof of \Cref{thm:lb} (first part)}
\label{subsec:lbproof1}
    %
    

%

In this subsection, we prove the first part of \Cref{thm:lb}:
for any $k\geq 1$, $z\geq 1$, and $\eps \geq \Omega(k^{-\nicefrac{1}{(z+2)}})$,
there is an instance $P$ in $\mathbb{R}^d$ (for $d \geq \Omega(k \eps^{-z})$) 
such that any $\eps$-coreset of $P$ is of size at least $\Omega(k \eps^{-z-2})$.

\begin{proof}
%
For the convenience of analysis, we first consider the special case of \kMeans\ ($z=2$) and $\eps = 0.01 k^{-1/4}$.
Later we show how to extend the proof to general $z\geq 1$ and $\eps \geq \Omega(k^{-\nicefrac{1}{(z+2)}})$. 
We first show the construction of the lower bound instance $P$.

    \paragraph{Construction of $P$.}
    W.l.o.g., assume $k^{1/4}$ is an even integer $\geq 10000$.
    Let $B = 100 k^{1/2} $ and $t = k^{1/4}$.
    Let $A^\star = \left\{a^\star_j = t e_j: j\in [B]\right\}$ be a collection of scaled basis vectors in $\R^B$.
    We construct a family of $k/2$ subsets $B_1,\ldots, B_{k/2}\subseteq B$ satisfying that
    \begin{enumerate}
        \item for every $i\in [k/2]$, $|B_i| = k^{1/2}$;
        \item for every $i\neq i'\in [k/2]$, $|B_i \cap B_{i'}|\leq 0.1 k^{1/2}$.
    \end{enumerate}
    Indeed, the existence of such a family is well known (e.g.,
    similar design is used the celebrated Nisan-Wigderson generator~\cite{nisan1994hardness}), and it can be easily
    constructed with nonzero probability by i.i.d. randomly sampling $B_i$'s from $B$. Then we know that for every $i\neq i'\in [k/2]$, the expected size of intersection $\Exp\left[|B_i \cap B_{i'}|\right] = 0.01 k^{1/2}$, which implies that
    \[
    \Pr\left[|B_i \cap B_{i'}|> 0.1 k^{1/2}\right] \leq e^{-10^{-4} k^{1/2}}
    \]
    by Chernoff bound.
    Consequently, we have that
    \[
    \Pr\left[|B_i \cap B_{i'}|\leq 0.1 k^{1/2}, \forall i\neq i'\in [k/2]\right]\geq 1- 10^4 k\cdot e^{-10^{-4} k^{1/2}} > 0.5
    \]
    by the union bound and the assumption that $k^{1/4}\geq 10000$.
    Hence, we can construct such a family of $k/2$ subsets $B_1,\ldots, B_{k/2}\subseteq B$ as desired.

    Given a subset $A\subseteq B$, let $e_A \in \R^B$ denote the vector with $e_{A,i} = 1$ if $i\in A$ and $e_{A,i} = 0$ if $i\notin A$.
    For every $i\in [k]$, arbitrarily divide $B_i$ into $B_i = B_i^+ \cup B_i^-$ where $|B_i^+| = |B_i^-| = 0.5 k^{1/2}$.
    We construct a dataset $X\subset \R^B$ to be
    \[
    X := \left\{p = a_j^\star + \frac{1}{t} (e_{B_i^+} - e_{B_i^-}): i\in [k/2], j\in B_i^+ \right\},
    \]
    i.e., we create a point for every $B_i$ and every $j\in B_i^+$.
    By construction, we know that $|X| = 0.25 k^{3/2}$.
    Also, note that every $p= a_j^\star + \frac{1}{t} (e_{B_i^+} - e_{B_i^-})\in X$ satisfies 
    \[
    \|p\|_2^2 = (t+ \frac{1}{t})^2 + (k^{1/2} - 1)\cdot \frac{1}{t^2} = t^2 + 3 \text{ and } \|p - a_j^\star\|^2 = 1.
    \]
    Our final construction of $P$ is to make $\frac{k}{2B}$ copies of $X$ in $\R^B$, say $X_1, \ldots, X_{\frac{k}{2B}}$ such that every two copies $X_i, X_j$ is sufficiently far from each other.
    Let $o_l\in \R^B$ be the ``origin'' of each $X_l$.
    Note that $|P| = \Theta(k^2)$ and $d(o_i,o_j) \rightarrow +\infty$ for $i\neq j\in [\frac{k}{2B}]$.

    \paragraph{Coreset lower bound of $P$.}
    Next, we show that any $\eps$-coreset of $P$ is of size at least $\Omega(k \eps^{-z-2}) = \Omega(k^2)$.
    By contradiction assume that $(S,w)$ is an $\eps$-coreset of $P$ with $|S| = o(k^2)$.
    For every copy $X_l$ ($l\in [\frac{k}{2B}]$), let $S_l = S\cap X_l$.
    We first have the following claim.

    \begin{claim}[\bf{Weights of $S_i$}]
        \label{claim:coreset_weight}
        For every $l\in [\frac{k}{2B}]$, we have $w(S_l)\in (1\pm 2\eps)\cdot |X_l|$.
    \end{claim}

    \begin{proof}
        The proof is easy. 
        By contradiction assume w.l.o.g. $w(S_1)\notin (1\pm 2\eps)\cdot |X_1|$.
        We construct a $k$-center set $C\subset \R^d$ as follows:
        \begin{enumerate}
            \item for every $l\in [\frac{k}{2B}]\setminus \{1\}$ and $i\in [B]$, construct a center $c$ at the copy of $a^\star_j$ in $X_l$;
            \item for $X_1$, construct a center $c^\star$ such that $d(o_1, c^\star) = 100 kt \eps^{-1}$.
        \end{enumerate}
        Note that the above procedure only constructs $\frac{k}{2} - B + 1$ centers.
        We let all remaining centers of $C$ be located at $c^\star$.
        Observe that $\cost_2(P,C)\approx |X_1|\cdot d^2(o_1, c^\star)$ and $\cost_2(S,C)\approx w(S_1)\cdot d^2(o_1, c^\star)$, which implies the claim. 
    \end{proof}

    \noindent
    By the pigeonhole principle, there must exist some $l\in [\frac{k}{2B}]$ such that $|S_l| = o(|X_l|) = o(k^{3/2})$.
    W.l.o.g., we assume $l = 1$ and $o_1 = 0$.

    Now we construct a $k$-center set $C\subset \R^d$ such that $\cost_2(S,C)\notin (1\pm \eps)\cdot \cost_2(P,C)$.
    We first construct $k/2 - B$ centers by the following way: for every $l\in [\frac{k}{2B}]\setminus \{1\}$ and $i\in [B]$, construct a center $c_{li}$ at the copy of $a^\star_i$ in $X_l$.
    By construction, the points in the cluster around $a^{\star}_i$ 
    in $X_l$ should be assigned to center $c_{li}$. 
    %
    %
    It is not difficult to see that
    \begin{align}
    \label{eq1:lower}
    \cost_2(X\setminus X_1, C)\leq k^2, \text{ and } \cost_2(S\setminus S_1, C)  \leq k^2,
    \end{align}
    since there are at most $k^2/4$ points in $X\setminus X_1$ and each point $p\in X_l$ ($l\in [\frac{k}{2B}]\setminus \{1\}$) contributes $d^2(p,c_{l i}) = 1$.
    In the following, we only consider the case that $\cost_2(X\setminus X_1, C) \leq \cost_2(S\setminus S_1, C)$.
    The argument for the reverse case is symmetric.

    For every $i\in [k/2]$, we define $X_{1,i} := \left\{p = a_j^\star + \frac{1}{t} (e_{B_i^+} - e_{B_i^-}): j\in B_i^+ \right\}$ and $S_{1,i} := S_1\cap X_{1,i}$.
    Note that $X_{1,i}$'s (resp. $S_{1,i}$'s) form a partition of $X_1$ (resp. $S_1$).
    Also define $T_i := \left\{ j\in B_i^+: p = a_j^\star + \frac{1}{t} (e_{B_i^+} - e_{B_i^-})\in S_{1,i}\right\}$ to be the collection of entries of $a_j^\star$ for $p\in S_{1,i}$.

    We will construct $k/2$ centers $c_1,\ldots, c_{k/2}$ for the first copy $X_1$, such that 
    the difference of $\cost_2(S_1,C)$ and $\cost_2(X_1,C)$ is sufficiently large.
    Recall that we have constructed $k/2-B$ centers for other copies.
    We place the remaining $B$ centers far away from all other points
    so that the clustering costs are not affected by them.
    Now, we describe the construction of $c_1,\ldots, c_{k/2}$.
    For every $i\in [k/2]$, $c_i\in \R^d$ is defined according to the following cases:
    \begin{enumerate}
        \item for every $j\notin B_i$, let $c_{i,j} = 0$;
        \item for every $j\in T_i$, let $c_{i,j} = 0.8$;
        \item for every $j\in B_i^+\setminus T_i$, let $c_{i,j} = 1$;
        \item arbitrarily select $|B_i^+\setminus T_i|$ distinct $j\in B_i^-$, let $c_{i,j} = -0.8$ for each such $j$;
        \item for the remaining $|T_i|$ entries $j\in B_i^-$, let $c_{i,j} = -1$ for each.
    \end{enumerate}
    By construction, for every $i\in [k/2]$, 
    \[
    \|c_i\|^2 = \frac{k^{1/2}}{2} \cdot (1 + 0.8^2) = 0.82 t^2.
    \]
    We also have the following claim.

    \begin{claim}[\bf{Properties of $C$}]
        \label{claim:property}
        Let $p= a_j^\star + \frac{1}{t} (e_{B_i^+} - e_{B_i^-})\in X_1$ for some $i\in [k/2]$ and $j\in B_i^+$.
        We have $d^2(p, C) = d^2(p,c_i)$ and moreover,
        \begin{enumerate}
            \item if $j\in T_i$, $d^2(p,C) = 1.82 t^2 - 3.4 t + 3$;
            \item if $j\in B_i^+\setminus T_i$, $d^2(p,C) = 1.82 t^2 - 3.8 t + 3$.
        \end{enumerate}
    \end{claim}

    \begin{proof}
        The calculation of $d^2(p,c_i)$ is not difficult.
        Since we already know $\|p\|^2 = t^2+3$ and $\|c_i\|^2 = 0.82t^2$, it remains to compute $-2\langle p,c_i \rangle$, which equals $-3.4 t$ when $j\in T_i$ and equals $- 3.8 t$ when $j\in B_i^+\setminus T_i$ by the construction of $c_i$.

        For every $i'\neq i$, we have that
        \[
        \langle p,c_i \rangle \leq t + 0.1k^{1/2}\cdot \frac{1}{t} = 1.1 t,
        \]
        which implies that $d^2(p, C) = d^2(p,c_i)$.
        Thus, we prove \Cref{claim:property}.
    \end{proof}

    \noindent
    Now we are ready to prove that $S$ is not an $\eps$-coreset of $P$.
    By \Cref{claim:property}, we know that
    \begin{align*}
    \cost_2(P,C) = & \quad \cost_2(X_1,C) + \cost_2(P\setminus X_1,C) & \\
    \leq &\quad |X_1|\cdot (1.82t^2 - 3.4t + 3) + k^2 &(\text{Claim~\ref{claim:property} and Ineq.~\eqref{eq1:lower}})\\
    \leq & \quad 2k^2, &
    \end{align*}
    and
    \begin{align*}
    & \quad \cost_2(S,C) - \cost_2(X,C) \\
    \geq &\quad \cost_2(S_1,C) - \cost_2(X_1,C) \quad (\text{since } \cost_2(X\setminus X_1, C) \leq \cost_2(S\setminus S_1, C))\\
    = &\quad (w(S_1) - |S_1|)\cdot (1.82t^2 - 3.4t + 3) - (|X_1| - |S_1|)\cdot  (1.82t^2 - 3.8t + 3)  \quad (\text{Claim~\ref{claim:property}})\\ 
    \geq &\quad (|X_1| - |S_1|)\cdot 0.4t - (|X_1| - w(S_1))\cdot 1.82 t^2 \\
    > &\quad 0.02 k^{7/4}. \quad \quad (|X_1| - |S_1|> 0.2k^{3/2} \text{ and Claim~\ref{claim:coreset_weight}})
    \end{align*}
    Hence, we have that $\cost_2(S,C) - \cost_2(P,C) > \eps\cdot \cost_2(X,C)$, which completes the proof.

\paragraph{Extension to $\eps \geq 0.01 k^{-1/4}$.}
    For each copy of $X$, let $t = \eps^{-1}$ and $B = 100 k/t^2 = 100 k \eps^{2}$.
    Let every $B_i$ ($i\in [k/2]$) be of cardinality $t^2$.
    As $\eps$ increases, $B$ increases and $|B_i|$ decreases, and hence, we still can guarantee that $|B_i\cap B_{i'}|\leq 0.1t^2$ holds for all $i\neq i'$.
    The construction of $P$ is to make $\frac{k}{2B}$ copies of $X$ in $\R^B$, say $X_1, \ldots, X_{\frac{k}{2B}}$.
    Consequently, there are $\Theta(k t^2) = \Theta(k \eps^{-2})$ points in each copy $X_i$, which implies that $|P| = \Theta(k \eps^{-4})$.
    By contradiction assume $S$ is an $\eps$-coreset of size $o(k \eps^{-4})$.
    Again, there must exist a copy $X_i$ such that $S_i = S\cap X_i = o(k \eps^{-2})$.
    Focus on $X_i$ and apply the same construction of $C$ which ensures $\cost_2(P,C) = \Theta(k \eps^{-4})$.
    Similar to the case of $\eps = 0.01 k^{-1/4}$, the error induced by the coreset $S$ for every point is of order $t$.
    Then the total error is of order $\Theta(k t^3) = \Theta(k \eps^{-3})$, which is larger than $\eps\cdot \cost_2(P,C)$.

\paragraph{Extension to general $z\geq 1$.}
    The construction is almost identical to that for \kMeans.
    For $\eps = \Omega(k^{-\nicefrac{1}{(z+2)}})$, the only difference is that we set $t = \eps^{-1}$ and $B = 100k/t^z = \Theta(k \eps^z)$ in each copy of $X$.
    Let every $B_i$ ($i\in [k/2]$) be of cardinality $t^2 \leq B/100$.
    The construction of $P$ is still to make $\frac{k}{2B}$ copies of $X$ in $\R^B$, say $X_1, \ldots, X_{\frac{k}{2B}}$.
    Consequently, there are $\Theta(k t^2) = \Theta(k \eps^{-2})$ points in each copy $X_i$, which implies that $|P| = \Theta(k \eps^{-z-2})$.
    Also note that the dimension of $P$ is $B = \Theta(k \eps^{-z})$.

    By construction, Claim~\ref{claim:property} still holds and hence,  we have for $p\in S$ and $q\in X\setminus S$,
    \[
    d^z(p,C) = (1.82t^2 - 3.4t +3)^{z/2}, \text{ and }  d^z(q,C) = (1.82t^2 - 3.8t +3)^{z/2}
    \]
    Since we consider $d^z$, the difference $d^z(p,C) - d^z(q,C)$ is of order $\Omega(t^{z-1})$.
    %
    %
    Hence, the overall error is of order $k B t^{z-1} = \Omega(k \eps^{-z-1})$ and the clustering cost is of order $O_z(k \eps^{-z-2})$, which implies the first coreset lower bound $\Omega(k \eps^{-z-2})$ when $\eps = \Omega(k^{-\nicefrac{1}{(z+2)}})$.
\end{proof}

\subsection{Proof of \Cref{thm:lb} (second part)}
\label{subsec:lbproof2}

\color{black}
In this subsection, we extend our previous construction to prove the second lower bound in
\Cref{thm:lb}: 
for $\eps\geq \Omega(k^{-\frac{z+1}{z+2}})$, there exists an instance $P$ in 
$\mathbb{R}^d$ (for $d = O(k^{\frac{2}{z+2}})$) such that any $\eps$-coreset of $P$ is of size at least $\Omega(k^{\frac{2z+3}{z+2}} \eps^{-1})$. 

\begin{proof}
    Since $k^{\frac{2z+3}{z+2}} \eps^{-1} \leq k \eps^{-z-2}$ when $\eps \geq k^{-\nicefrac{1}{(z+2)}}$, we only need to prove the bound for any fixed $\eps\in [k^{-\frac{z+1}{z+2}} , k^{-\nicefrac{1}{(z+2)}}]$.
    We let $t = k^{\nicefrac{1}{(z+2)}}$ and $B = 100 k^{\nicefrac{2}{(z+2)}}$.
    Similar to the construction in Section~\ref{subsec:lbproof1}, the construction of $P$ is to make $\frac{k}{2B}$ copies of point set $X$ in $\R^B$, say $X_1, \ldots, X_{\frac{k}{2B}}$, where the construction of $X$ will be discussed shortly. 
    To prove the lemma, we want $|X| = k^{\frac{z+3}{z+2}} \eps^{-1}$ such that $|P| = \Theta(k^{\frac{2z+3}{z+2}} \eps^{-1})$.
    Using the same construction as in Section~\ref{subsec:lbproof1}, we construct $k^{\frac{z+1}{z+2}} \eps^{-1}$  different $B_i$'s from $B$ such that
    %
    \begin{enumerate}
        \item for every $i\in [k/2]$, $|B_i| = B/100 = k^{\nicefrac{2}{(z+2)}}$;
        \item for every $i\neq i'\in [k/2]$, $|B_i \cap B_{i'}|\leq 0.1 k^{\nicefrac{2}{(z+2)}}$.
    \end{enumerate}
    Recall that the existence of such a family is well known and it can be easily
    constructed with nonzero probability by i.i.d. randomly sampling $B_i$'s from $B$.
    Then we know that for every $i\neq i'\in [k^{\frac{z+1}{z+2}} \eps^{-1}]$, the expected size of intersection $\Exp\left[|B_i \cap B_{i'}|\right] = 0.01 k^{\nicefrac{2}{(z+2)}}$, which implies that
     \[
    \Pr\left[|B_i \cap B_{i'}|> 0.1 k^{2/(z+2)}\right] \leq e^{- \Omega(k^{2/(z+2)})}
    \]
    by Chernoff bound.
    Consequently, 
    by the union bound we have that
    \[
    \Pr\left[|B_i \cap B_{i'}|\leq 0.1 k^{2/(z+2)}, \forall i\neq i'\in [k^{\frac{z+1}{z+2}} \eps^{-1}]\right]\geq 1- O(k^{\frac{4z+4}{z+2}})\cdot e^{-\Omega( k^{2/(z+2)})} > 0.5,
    \]
    where the last inequality holds when $k\geq \Omega_z(1)$ is large enough.
    Similar to that in Section~\ref{subsec:lbproof1}, we construct $X\subset \R^B$ via these $B_i$'s in the following way: 
    For every $i\in [k^{\frac{z+1}{z+2}} \eps^{-1}]$, arbitrarily divide $B_i$ into $B_i = B_i^+ \cup B_i^-$ where $|B_i^+| = |B_i^-| = 0.5 k^{2/(z+2)}$.
    Define
    \[
    X := \left\{p = a_j^\star + \frac{1}{t} (e_{B_i^+} - e_{B_i^-}): i\in [k^{\frac{z+1}{z+2}} \eps^{-1}], j\in B_i^+ \right\},
    \]
    i.e., we create a point for every $B_i$ and every $j\in B_i^+$.
    By construction, we know that $|X| = \Theta(k^{\frac{z+3}{z+2}} \eps^{-1})$ as desired.

    Next, we prove  the coreset lower bound for
    such dataset $P$.
    By contradiction assume that $(S,w)$ is an $\eps$-coreset for \kzC\ of $P$ with $|S| = o(|P|)$.
    There must exist a copy $X_1$ such that $|S\cap X_1| = o(|X_1|)$ and we focus on this copy together with $S_1 = S\cap X_1$.
    For every $i\in [k^{\frac{z+1}{z+2}} \eps^{-1}]$, we define $X_{1,i} := \left\{p = t e_j + \frac{1}{t} (e_{B_i^+} - e_{B_i^-}): j\in B_i^+ \right\}$ and $S_{1,i} := S_1\cap X_{1,i}$.
    Then, we construct $k/2-B$ centers in $C$ for other copies $X_2, \ldots, X_{\frac{k}{2B}}$ 
    in the same way as in Section~\ref{subsec:lbproof1}:
    for every $l\in [\frac{k}{2B}]\setminus \{1\}$ and $i\in [B]$, construct a center $c_{li}$ at the copy of $a^\star_i = t e_i$ in $X_l$. 
    This construction ensures that $\cost_z(P\setminus X_1, C) = O(|P|) = O(k^{\frac{2z+3}{z+2}} \eps^{-1})$.
    We also construct a specific center $c^\star$ for $X_1$ as
    \[
    c^\star = \alpha (1,1,\ldots,1) 
    \]
    for some $\alpha > 0$ such that for every point $p\in X_1$, $d^2(p,c^\star) = 1.82t^2 - 3.4t + 3$.
    Such $\alpha$ must exist since for any $p,q\in X_1$, $\langle p, (1,1,\ldots,1) \rangle = \langle q, (1,1,\ldots,1) \rangle$ holds.
    By the same argument, Claim~\ref{claim:coreset_weight} still holds, i.e., $w(S_1) \in (1\pm 2\eps)\cdot |X_1|$.
    W.l.o.g., we assume $\cost_2(P\setminus X_1, C) \leq \cost_2(S\setminus S_1, C)$ and the inverse case is symmetric.

    We claim that there must exist a collection $A$ of $k/2$ distinct $i$'s such that 
    \[
    \Bigl|\bigcup_{i\in A} S_{1,i}\Bigr| = o\left(\Bigl|\bigcup_{i\in A} X_{1,i}\Bigr|\right) = o(k^{\frac{z+4}{z+2}}),
    \]
    and $\sum_{i\in A} w(S_{1,i}) \geq (1-3\eps) \cdot |\bigcup_{i\in A} X_{1,i}| = (1-3\eps) \cdot k^{\frac{z+4}{z+2}}/2$.
    This claim is not hard to verify by uniformly randomly selecting $A$ from $[k^{\frac{z+1}{z+2}} \eps^{-1}]$.
    Then we have 
    \[
    \Exp_A\left[\Bigl|\bigcup_{i\in A} S_{1,i}\Bigr| \right] = \frac{|A|}{k^{\frac{z+1}{z+2} }\eps^{-1}} \cdot |S_1| = \frac{|A|}{k^{\frac{z+1}{z+2} } \eps^{-1}} \cdot o(|X_1|) = \frac{k/2}{k^{\frac{z+1}{z+2} } \eps^{-1}} \cdot o(k^{\frac{z+3}{z+2} } \eps^{-1}) = o(k^{\frac{z+4}{z+2}}),
    \]
    and 
    \[
    \Exp_A\left[\sum_{i\in A} w(S_{1,i})\right] = \frac{|A|}{k^{\frac{z+1}{z+2} }\eps^{-1}} \cdot w(S_1) \geq (1-2\eps) \cdot \frac{|A|}{k^{\frac{z+1}{z+2}}\eps^{-1}}\cdot |X_1| \geq (1-2\eps) \cdot k^{\frac{z+4}{z+2}}/2.
    \] 
    These properties imply that the claim holds with a random $A$ with non-zero probability.

    Let $S' = \bigcup_{i\in A} S_{1,i}$ and $X' = \bigcup_{i\in A} X_{1,i}$.
    Now we apply the previous construction for $A$ such that for every $p\in S'$, $d^2(p,C) = 1.82t^2 - 3.4t +3$ and for every $p\in X'\setminus S'$, $d^2(p,C) = 1.82t^2 - 3.6t +3$.
    By a similar calculation as in Section~\ref{subsec:lbproof1}, this construction implies that
    \[
    \cost_2(S_1,C) - \cost_2(X_1,C) \geq \Omega(\sum_{i\in A} w(S_{1,i}) \cdot t^{z-1}) - (|X_1|- w(S_1))\cdot 1.82t^z = \Omega(k^{\frac{2z+3}{z+2}}),
    \]
    while $\cost_z(P, C) = \cost_z(P\setminus X_1, C) + \cost_z(X_1, C) = O(k^{\frac{2z+3}{z+2}} \eps^{-1})$.
    Hence, $S$ is not an $\eps$-coreset for \kzC\ of $P$.
    This completes the proof of the second half of Theorem~\ref{thm:lb}.
\end{proof}

\color{black}

 \subsection{New Lower Bounds for Doubling Spaces: Proof of Corollary~\ref{cor:doubling}}
 \label{sec:doubling}

In this subsection, we extend the lower bounds in Theorem~\ref{thm:lb} to doubling metrics.

\begin{proof}[Proof of Corollary~\ref{cor:doubling}]
    We first show the coreset lower bound $\Omega(k\cdot \ddim \cdot \eps^{-\max\{z,2\}}/ \log k)$ when $\eps \geq \Omega(k^{-\nicefrac{1}{(z+2)}})$.
    For $1\leq z\leq 2$, it is known that $\Omega(k\cdot \ddim\cdot \eps^{-2})$ is a lower bound by~\cite{cohenaddad2022towards}.
    Thus, we only need to consider the case that $z>2$.

    Let $P\subset \R^d$ be the instance constructed in the proof of Theorem~\ref{thm:lb} that consists of $\Theta(k \eps^{-z-2})$ points.
    We perform an $O(z^{-1}\eps)$-terminal embedding $f$ on $P$ (see Section~\ref{sec:notation_Euclidean}) and obtain an embedded dataset $f(P)\subset \R^m$ satisfying that
    \begin{enumerate}
    \item $m = O(\eps^{-2} \log (k \eps^{-1})) = O(\eps^{-2} \log k)$ since $\eps \geq \Omega(k^{-\nicefrac{1}{(z+2)}})$;
    \item for every point $p\in P$ and $q\in \R^d$, $d^z(p,q)\in (1\pm \eps)\cdot d^z(f(p),f(q))$.
    \end{enumerate}
    It is well known that the doubling dimension of $\R^m$ is at most $\ddim \leq O(m)$ (see e.g.,\cite{vergergaugry2005covering}).
    We assume by contradiction that there exists an $\eps$-coreset $f(S)\subseteq f(P)$ of size at most
    $$
    o(k \cdot \ddim \cdot \eps^{-z}/\log k) \leq 
    o(k \cdot m \cdot \eps^{-z}/\log k) \leq  
    o(k\cdot\eps^{-z-2})
    $$
    for the embedded instance $f(P)$ for \kzC.
    Since $f$ is an $O(z^{-1}\eps)$-terminal embedding, we know that $S$ is an $O(\eps)$-coreset of the original instance $P$ for \kzC.
    However, $|S| =  o(k\cdot \eps^{-z-2})$, which contradicts Theorem~\ref{thm:lb}.
    Hence, 
    any $\eps$-coreset of $f(P)$ should be of size at least 
    $\Omega(k\cdot \ddim \cdot \eps^{-\max\{z,2\}}/ \log k)$.

    The argument for the second lower bound $\Omega(k^{\frac{2z+1}{z+2}} \cdot \ddim \cdot \eps^{-1})$ is almost identical.
    We can derive this lower bound from the proof of Theorem~\ref{thm:lb} for the bound $\Omega(k^{\frac{2z+3}{z+1}} \eps^{-1})$.
    The only difference is that the dimension $B = k^{\nicefrac{2}{(z+2)}}$ now (noting that terminal
    embedding does not help here).
    Hence, we have $\ddim \leq O(k^{\nicefrac{2}{(z+2)}})$ which implies that any $\eps$-coreset of $P$ should be of size at least $\Omega(k^{\frac{2z+1}{z+2}} \cdot \ddim \cdot \eps^{-1})$.
    
    To sum up, we complete the proof of Corollary~\ref{cor:doubling}.
\end{proof}

 \color{black}


\section{Improved Coreset Upper Bounds for \kzC}
\label{sec:upper_short}

%
In this section, we introduce the main ideas for our analysis of the upper bound and discuss some key technical differences between ours and prior work~\cite{cohenaddad2022towards,cohenaddad2022improved}.
For completeness, we provide all 
technical details in the appendix.

\sloppy
We first briefly review the framework of prior work~\cite{cohenaddad2022towards,cohenaddad2022improved}.
Let $A^\star = \left\{a_1^\star, \ldots, a_k^\star\right\}\subset \R^d$ be an $O(1)$-approximation of $P$ for \kzC.
For every $i\in [k]$, let $P_i:= \left\{p\in P: a_i^\star = \arg\min_{c\in A^\star} d(p,c)\right\}$ denote the induced cluster of $P$ whose closest center is $a_i^\star$. 
Cohen-addad et al.~\cite{cohenaddad2021new,cohenaddad2022towards,cohenaddad2022improved} showed that we can make the following simplifying assumption without loss of generality.
\begin{enumerate}
	\item The number of distinct points $\|P\|_0$ is at most $2^{O(z)}\cdot\poly(k\eps^{-1})$;
    \item For every $i\in [k]$ and every $p,q\in P_i$, $d^z(p, A^\star) \leq 2 d^z(q,A^\star)$;
    \item For every $i,j\in [k]$, $d^z(P_i, A^\star) \leq 2d^z(P_j, A^\star)$.
\end{enumerate}

\noindent
The first assumption is via the iterative size reduction approach \cite{braverman2021coresets}.
The second assumption ensures that the distances of all points in $P_i$ to $a_i^\star$ are close.
This structure is called a \emph{ring} in \cite{cohenaddad2021new,cohenaddad2022towards}.
The third assumption ensures that the clustering cost of every partition $P_i$ to $A^\star$ is close.
Combining these assumptions, dataset $P$ consists of ``balanced'' rings, which is called \emph{a main group} in~\cite{cohenaddad2021new,cohenaddad2022towards}. 
Cohen-addad et al. \cite{cohenaddad2021new} showed that a general dataset can be decomposed into $\tilde{O}_z(1)$ main groups, where $\tilde{O}_z(1) = \poly(z, \log k, \log \eps^{-1})$.
Consequently, if there exists a coreset for a main group of size $t$, the coreset size without these data assumptions is at most $\tilde{O}_z(t)$.
Then the most essential part is to analyze the coreset size for a main group (under the aforementioned assumptions).

The key of our analysis is to show that a collection $S$ of $\Gamma = 2^{O(z)}\cdot \tilde{O}(k^{\frac{2z+2}{z+2}} \eps^{-2})$ samples from $P$ is an $\eps$-coreset, where each sample $p\in P$ is selected with probability $\frac{d^z(p, A^\star)}{\cost_z(P, A^\star)}$ and weighted by $w(p) = \frac{\cost_z(P,A^\star)}{\Gamma\cdot d^z(p, A^\star)}$.
It suffices to bound the (expected) estimation error of the sample $S$:
\begin{align}
\label{eq:sketch_1}
\Exp_{S}\sup_{C\in \calC}\left[|\cost_2(P,C) - \cost_2(S,C)|\right]\leq \frac{\eps\cdot \cost_z(P, C+A^\star)}{\tilde{O}_z(1)},
\end{align}
where $\cost_z(P,C+A^\star):= \cost_z(P, C) + \cost_z(P, A^\star)$.
Theorem~\ref{thm:Euclidean_informal} is a direct corollary of the above lemma by the Markov inequality.
To prove the above inequality, we also adopt a chaining argument, which is different from the one in \cite{cohenaddad2022towards,cohenaddad2022improved}.

Define $r_i:= \min_{p\in P_i} d^z(p, a_i^\star)$ for every $i\in [k]$.\footnote{$r_i$ is referred to as $2^j \Delta_i$ in Appendix \ref{sec:notation}.}
Given a subset $B\subseteq [k]$ and an integer $\beta\geq 0$, let $\calC_{B,\beta}$ denote the collection of center sets $C\in \calC$ satisfying $2^\beta\cdot r_i\leq d^z(a_i^\star, C) < 2^{\beta+1}\cdot r_i$ for every $i\in B$.
We say a set $V\subset \R^{|S|}$ of \emph{cost vectors} is an \emph{$\alpha$-covering} of $S$ if for each $C\in \calC_{B,\beta}$, there exists a cost vector $v\in V$ such that for any $i\in B$ and $p\in P_i\cap S$, the following inequality holds:
\begin{align*}
		\left|d^z(p,C)-d^z(a^\star_i,C)-v_p\right|\leq \alpha\cdot \err(p,C),
\end{align*}
where $\err(p,C)$ is the {\em relative covering error} of $p$ to $C$ defined as follows
\begin{align*}
        \err(p,C):=\sqrt{d^z(p,C)\cdot d^z(p,A^\star)}+d^z(p,A^\star))\cdot \sqrt{\frac{\ZZGCA{G}}{\cost_z(G,A^\star)}}.
\end{align*}
The key is to bound the covering number $|V_\alpha|$.

We remark that our notion of covering is different from that in prior work \cite{cohenaddad2022towards,cohenaddad2022improved}. 
For instance, our definition has the following differences from \cite[Definition 3]{cohenaddad2022towards}, which are crucial for improving the coreset sizes.
One difference is that our covering is constructed on the distance difference $d^z(p,C)-d^z(a^\star_i,C)$ instead of $d^z(p,C)$, which is upper bounded by $z\cdot d(p,a^\star_i)\cdot (d^{z-1}(p,C) + d^{z-1}(p, a^\star_i))$ by the triangle inequality. 
The other one is that our relative covering error $\err(p,C)$ is proportional to $\sqrt{d^z(p,C)\cdot d^z(p,A^\star)}$, which is typically larger (especially when $d(p,C)\gg d(p,A^\star)$ ).
For comparison, the covering error defined in \cite[Definition 3]{cohenaddad2022towards} is $d^z(p,C)+d^z(p,A^\star)$. 
Compared to \cite{cohenaddad2022improved}, our notion contains an additional factor $\sqrt{\frac{\ZZGCA{G}}{\cost_z(G,A^\star)}}$ that is at least 1.
Our smaller error bound is essential for the improvement of coreset sizes, since its scale affects the variance of our sampling scheme (Lemma~\ref{lm:variance}) and the variance bound directly appears in the coreset size (Lemma~\ref{lm:Gaussian_sup}).
Specifically, our covering notion enables us to bound the variance of our sampling scheme by $\Var_\alpha = O(\frac{\alpha^2}{\Gamma})$ (Lemma~\ref{lm:variance}), saving a factor of $\min\left\{\eps^{-z}, k\right\}$ compared to that in \cite[Lemma 23]{cohenaddad2022towards}.
We remark that \cite{cohenaddad2022improved} also introduces coverings of $S$ for \kMedian\ and \kMeans, but uses different covering error bounds.
However, it is unclear to us how to use their error bounds to obtain the same results for general power $z\geq 1$ as ours.

By a standard chaining argument, it suffices to prove the existence of an $\alpha$-covering $V_\alpha$ of $S$ with
\begin{align}
\label{eq:sketch_2}
\log |V_\alpha| \leq \tilde{O}_z(1) \cdot  k \alpha^{-2}\cdot \min\left\{2^{\beta z}, 1+k 2^{-2\beta}\right\}.
\end{align}
By simple calculation, we have $\min\left\{2^{\beta z}, 1+k 2^{-2\beta}\right\} \leq k^{\frac{z}{z+2}}$, where the equality can be achieved when $2^\beta = k^{1/(z+2)}$.
Then we can select the sample number $\Gamma$ to be
\[
\eps^{-2}\cdot \Var_{\alpha} \cdot \log (|V_\alpha|) \approx k^{\frac{2z+2}{z+2}}\eps^{-2}
\]
such that Inequality~\eqref{eq:sketch_1} holds.
It remains to prove Inequality~\eqref{eq:sketch_2}.
Note that the first bound $\log |V_\alpha| \leq \tilde{O}_z(1) \cdot  k \alpha^{-2}\cdot 2^{\beta z}$ has been proved in the literature by applying terminal embedding (see e.g., \cite[Lemma 4]{cohenaddad2022improved}).
Thus, we only need to prove the second bound $\log |V_\alpha| \leq \tilde{O}_z(1) \cdot  k \alpha^{-2}\cdot (1+k 2^{-2\beta})$.
To this end, we introduce a new notion for dimension reduction, called 
{\em additive terminal embedding}.

\begin{definition}[\bf{Additive terminal embedding}]
\label{def:additive_terminal_embedding}
Let $\alpha\in (0,1)$, $r>0$, and $X\subset \R^d$ be a collection of $n$ points within a ball $B(0, r)$ with $0^d\in X$.
A mapping $f:\R^d\rightarrow \R^m$ is called an $\alpha$-additive terminal embedding of $X$ if for any $p\in X$ and $q\in \R^d$, $|d(p,q) - d(f(p),f(q))|\leq \alpha\cdot r$.
\end{definition}
	
\noindent
The main difference of the above definition from the classic terminal embedding~\cite{Narayanan2019OptimalTD,Cherapanamjeri2022TerminalEI} is that we consider an additive error $\alpha\cdot r$ instead of a multiplicative error $\alpha\cdot d(p,q)$. 
Our error is smaller for remote points $q\in \R^d\setminus B(0,2r)$. 
Since $p\in B(0,r)$, we know that $d(p,q)\geq r$ by the triangle inequality.
Hence, $\alpha\cdot r$ is a smaller error compared to $\alpha\cdot d(p,q)$. 
We have the following theorem showing that the target dimension of an additive terminal embedding is the same as the multiplicative version, using the same mapping as in~\cite{mahabadi2018nonlinear,Narayanan2019OptimalTD}.
The proof can be found in Appendix~\ref{sec:notation_Euclidean}.

\begin{theorem}[\bf{Additive terminal embedding}]
\label{thm:additive_terminal_main}
Let $\alpha\in (0,1)$, $r>0$, and $X\subset \R^d$ be a collection of $n$ points within a ball $B(0, r)$ with $0^d\in X$.
There exists an $\alpha$-additive terminal embedding with a target dimension $O(\alpha^{-2}\log n)$.
\end{theorem}

\noindent
Now we prove $\log |V_\alpha| \leq \tilde{O}_z(1) \cdot  k \alpha^{-2}\cdot (1+k 2^{-2\beta})$.
Fix $C\in \calC_{B,\beta}$ and we have $\frac{\cost_1(P, C)}{\cost_1(P,A^\star)} \geq \frac{|B| \cdot 2^{\beta}}{4k}$ by definition.
For each $i\in B$, let $f_i$ be an $\sqrt{|B|k^{-1}} 2^{\beta - 5z}\alpha$-additive terminal embedding of $P_i$ into $m=O(|B|^{-1} k 2^{10z-2\beta}\alpha^{-2}\log \|P\|_0)$ dimensions given by Theorem~\ref{thm:additive_terminal_main}.
By definition and some standard calculation, we obtain that
\begin{align*}
& \quad  \left|(d^z(f_i(p),f_i(C)) - d^z(f_i(a^\star_i),f_i(C))) - (d^z(p,C) - d^z(a^\star_i, C)) \right| \\
\leq & \quad  O(\alpha) \cdot \sqrt{d^z(p,C)\cdot d^z(p,A^\star)}\cdot \sqrt{\frac{\cost_z(P,C)}{\cost_z(P,A^\star)}} = \alpha\cdot \err(p, C),
\end{align*}
which implies that the covering number for a single cluster $P_i$ is at most $\exp(\tilde{O}_z(k m))$.
Hence, the covering number $|V_\alpha|$ over all $|B|$ clusters satisfy that
\[
\log |V_\alpha| \leq |B|\cdot \tilde{O}_z(k m) \leq \tilde{O}_z(1) \cdot  k \alpha^{-2}\cdot (1+k 2^{-2\beta}).
\]

We remark that the target dimension $m$ is proportional to $2^{-2\beta}$, indicating that the covering number of a cluster $P_i$ decreases as the distance $d(a_i^\star,C)$ to $C$ increases.
This new geometric observation is essential for us to reduce the covering number, which may be of independent interest together with the notion of additive terminal embedding.
%

 \section{Final Remarks}
    In this paper, we make significant progress on understanding the 
    optimal coreset sizes for Euclidean \kzC, in terms of both upper and lower bounds. In particular, for Euclidean space, when 
    $\eps \geq \Omega(k^{-\nicefrac{1}{(z+2)}})$, we prove a new coreset lower bound $\Omega(k \eps^{-z-2})$, which matches the upper bound in ~\cite{cohenaddad2022towards} (upto poly-log factors).
    For the upper bound, we provide efficient coreset construction algorithms with improved or optimal coreset sizes, matching the previous bounds for \kMedian\ and \kMeans\ \cite{cohenaddad2022improved} and extending them to all powers $z\geq 1$.
    In fact, before discovering our lower bound, we have tried several other methods to improve the covering number bound, including various dimension notions 
    (such as Natarajan dimension \cite{natarajan1989learning,guermeur2007vc})
    and covering number bounds in
    multiclass prediction (e.g., \cite{zhang2002covering}). 
    \footnote{
        The multi-class classification problem (e.g., in \cite{guermeur2007vc}) considers the setting
        where each function $f$ in the hypothesis class $\calF$ maps a point $x\in X$ to a $k$-dimensional vector
        $\{f_1(x),\ldots, f_k(x)\}$ and assign $x$ to class $y\in [k]$ if $f_y(x)> \max_{i\ne y}f_i(x)$.
        To show learnability and uniform convergence, an important component in this line of work is to bound the covering number of $\calF$ via suitably defined notion of dimensions (which generalize VC-dimension or fat-shattering dimension for binary classification).
        The setting and the goal are very similar to our problem where we also aim to bound certain 
        covering number for all $k$-center sets $C\in \calC$. Each $k$-center set corresponds to a function which maps a point $p$ to the $k$-dimensional distance vector $\{d(p,c_1),\ldots, d(p, c_k\}$
        and $p$ is assigned to the $i$-th center if $d(p,i)\leq \min_{j\ne i} d(p,j)$.
    }
    However, all techniques
    we have tried led to more or less the same upper bound (we omit the details), which led us to suspect 
    that $O(k\eps^{-2})$ may not be the optimal coreset size bound. 
    In our lower bound instance, we allow a center to influence points from various clusters, which 
    may be useful in other contexts, e.g., understanding the covering number bound and sample complexity in multiclass prediction.

Our work reveals several avenues for future investigation. Firstly, addressing the outstanding gap for $\epsilon < \Omega(k^{-\frac{1}{(z+2)}})$ in the context of \kzC\ in Euclidean space remains a significant unsolved problem. 
Secondly, it has been established that the coreset size for a set of base vectors is $\Theta(k\epsilon^{-2})$, which is notably smaller than our derived general coreset lower bound. Exploring instance-optimal coreset size limits constitutes an intriguing pursuit.
Another promising direction involves examining whether storing a coreset represents the most efficient compression method in terms of space complexity for clustering tasks. A recent finding \cite{Zhu2024SpaceCO} demonstrated that the space complexity for Euclidean \kzC\ is $\Omega(kd\epsilon^{-2})$, suggesting that coreset is indeed the optimal compression approach when $k$ is constant.
For the general case of varying $k$, it is worthwhile to explore whether our coreset lower bound construction can be extended to establish a space lower bound of $\Omega(kd \epsilon^{-z-2})$. This extension would imply that, under the condition $\epsilon < \Omega(k^{-\frac{1}{(z+2)}})$, the coreset is also the optimal compression scheme.
%

	
\bibliographystyle{plain}
\bibliography{references}
	
\newpage
\appendix

\section{The Coreset Construction Algorithm for \kzC}
	\label{sec:prior}
	
	In this section, 
	we present a unified coreset construction algorithm (Algorithm~\ref{alg:coreset}) for \kzC\ via importance sampling for all metric spaces considered in this paper.
	Our algorithm largely follows the one proposed in~\cite{cohenaddad2021new,cohenaddad2022towards}, with some minor variations.
    For preparation, we provide the following relaxed triangle inequality.

	\begin{lemma}[\bf{Relaxed triangle inequality}]
		\label{lm:relaxed}
		Let $a,b,c\in \calX$ and $z\geq 1$.
		For every $t\in (0,1]$, the following inequalities hold:
		\[
		d^z(a,b) \leq (1+t)^{z-1} d^z(a,c) + (1+\frac{1}{t})^{z-1} d^z(b,c).
		\]
	\end{lemma}
\noindent
We first introduce some useful notations.

	\subsection{Partition into Rings and Groups}
	\label{sec:notation}
	
	\sloppy
	Let $A^\star = \left\{a^\star_1,\ldots,a^\star_k\in \calX\right\}\in \calC$ be a constant approximation for the \kzC\ problem.
	Let $P_i=\left\{x\in P: \arg\min_{j\in [k]}d(x,a_j) = i\right\}$ be the $i$-th cluster induced by $A^\star$ (breaking ties arbitrarily).
	For each $i\in [k]$, denote $\Delta_i:= \frac{1}{|P_i|}\cost_z(P_i, A^\star)$ to be the average cost of $P_i$ to $A^\star$.

	\paragraph{Ring structure and group structure.}
	Following \cite{cohenaddad2021new,cohenaddad2022towards},
	we first partition the points into a collection of rings, and then group
	some rings into the group structure, based on $P_i$s and $A^\star$.
	We first partition each clusters $P_i$ into rings according to the ratio $\frac{d^z(p,A^\star)}{\Delta_i}$ for $p\in P_i$.

	\begin{definition}[\bf{Ring structure~\cite{cohenaddad2021new,cohenaddad2022towards}}]
		\label{def:ring}
		\sloppy
		For each $i\in [k]$ and $j\in \calZ$, define \emph{ring} $R_{ij} := \left\{p\in P_i: 2^j\Delta_i\leq d^z(p,A^\star) < 2^{j+1}\Delta_i\right\}$ to be the set of points in $P_i$ whose $z$-th power distances to $A^\star$ is within $[2^j\Delta_i, 2^{j+1}\Delta_i)$.
		For each $\in \calZ$, define $R(j):=\bigcup_{i\in [k]} R_{ij}$ to be the union of all $j$-th level rings $R_{ij}$.
		For each $P_i$, we divide rings in the following way:
		\begin{itemize}
			\item If $j\leq z\log (\eps/z)$, we call $R_{ij}$ an \emph{inner ring}, where each point $p\in R_{ij}$ satisfies $d(p, A^\star)\leq \eps \Delta_i$.
			\item Let $R^{(o)}_i = \bigcup_{j\geq 2z\log(z\eps^{-1})} R_{ij}$ denote an \emph{outer ring}, where each point $p\in R_{ij}$ satisfies $d(p, A^\star)\geq \eps^{-2} \Delta_i$.
			Let $R^{(o)} = \bigcup_{i\in [k]} R^{(o)}_i$ denote the collection of all points in outer rings.
			\item If $z\log (\eps/z)<j<2z\log(z\eps^{-1})$, we call $R_{ij}$ a \emph{main ring}.
		\end{itemize} 
	\end{definition}

	\noindent
	%
	By merging rings $R$ at the same level from different clusters with similar $\cost_z(R,A^\star)$, we propose the following group structure.

	\begin{definition}[\bf{Group structure}]
		\label{def:group}
		For each integer $z\log (\eps/z)<j<2z\log(z\eps^{-1})$ and $b\leq 0$, we 
		denote the interval 
		$I(j,b)=[2^b\cdot \cost_z(R(j), A^\star) ,2^{b+1}\cdot \cost_z(R(j), A^\star)$
		and define a \emph{main group}
		\[
		G^{(m)}(j,b) := \bigcup_{i\in [k]: \cost_z(R_{ij},A^\star) \in I(j,b)} R_{ij}
		\]
		as the collection of main rings $R_{ij}$ with $2^b\cdot \cost_z(R(j), A^\star) \leq \cost_z(R_{ij},A^\star)< 2^{b+1}\cdot \cost_z(R(j), A^\star)$.
		%
		%
		Let $\calG^{(m)}(j):=\left\{G^{(m)}(j,b): z\log(\eps/4z) - \log k <b\leq 0\right\}$ be the collection of $j$-th level main groups.
		Let $\calG^{(m)}:=\left\{ G\in \calG^{(m)}(j): z\log (\eps/z)<j<2z\log(z\eps^{-1})\right\}$ be the collection of all main groups.

		Similarly, for each integer $b\leq 0$, 
		we define
		$I(o,b)=[2^b\cdot \cost_z(R^{(o)}(j), A^\star) ,2^{b+1}\cdot 
		\cost_z(R^{(o)}(j), A^\star)$
		and an \emph{outer group}
		\[
		G^{(o)}(b) := \bigcup_{i\in [k]:  \cost_z(R^{(o)}_i,A^\star)
			\in I(o,b)} R^{(o)}_i
		\]
		to be the collection of outer rings $R_i^{(o)}$ with $2^b\cdot \cost_z(R^{(o)}, A^\star) \leq \cost_z(R^{(o)}_i,A^\star)< 2^{b+1}\cdot \cost_z(R^{(o)}, A^\star)$.
		%
		%
		Let $\calG^{(o)}:=\left\{ G^{(o)}(b):z\log(\eps/4z) - \log k <b\leq 0\right\} $ be the collection of outer groups. 
		%
		
		Let $\calG:= \calG^{(m)} \cup \calG^{(o)}$ be the collection of all groups.
	\end{definition}
	
	\noindent
	By definition, we know that all groups in $\calG$, including main groups $G^{(m)}(j,b)$s and outer groups $G^{(o)}(b)$s, are pairwise disjoint.
	For main groups, we provide an illustration in Figure~\ref{fig:group}.
	We also have the following observation that lower bounds $\frac{d^z(p,A^\star)}{\cost_z(G,A^\star)}$ for main groups $G\in \calG^{(m)}$. 
	\begin{observation}[\bf{Main group cost~\cite{cohenaddad2022towards}}]
		\label{ob:group}
		Let $G\in \calG^{(m)}$ be a main group.
		Let $i\in [k]$ be an integer satisfying that $P_i\cap G\neq \emptyset$.
		For any $p\in P_i\cap G$, we have
		\[
		\cost_z(G,A^\star) \leq 2k\cdot \cost_z(P_i\cap G,A^\star)\leq 4k\cdot |P_i\cap G|\cdot d(p,A^\star), \text{ and } |P_i\cap G|\cdot d(p,A^\star) \leq 2\cost_z(P_i\cap G).
		\]
	\end{observation}
	
	\noindent
	Note that $\calG$ may not contain all points in $P$ -- actually, we discard some ``light'' rings in Definition~\ref{def:group}.
	We will see that we only need to take samples from groups in $\calG$ for coreset construction and the remaining points can be ``represented by'' points in 
	$A^\star$ with a small estimation error.
	Our group structure is slightly different from that in~\cite{cohenaddad2021new,cohenaddad2022towards}: they gather groups $G^{(m)}(j,b)$/$G^{(o)}(b)$ with $ -\log k\leq b \leq 0$ as an entirety and sample from them together.
	%
	%
	
	\begin{figure}
		\centering
		\includegraphics[width=0.9\textwidth]{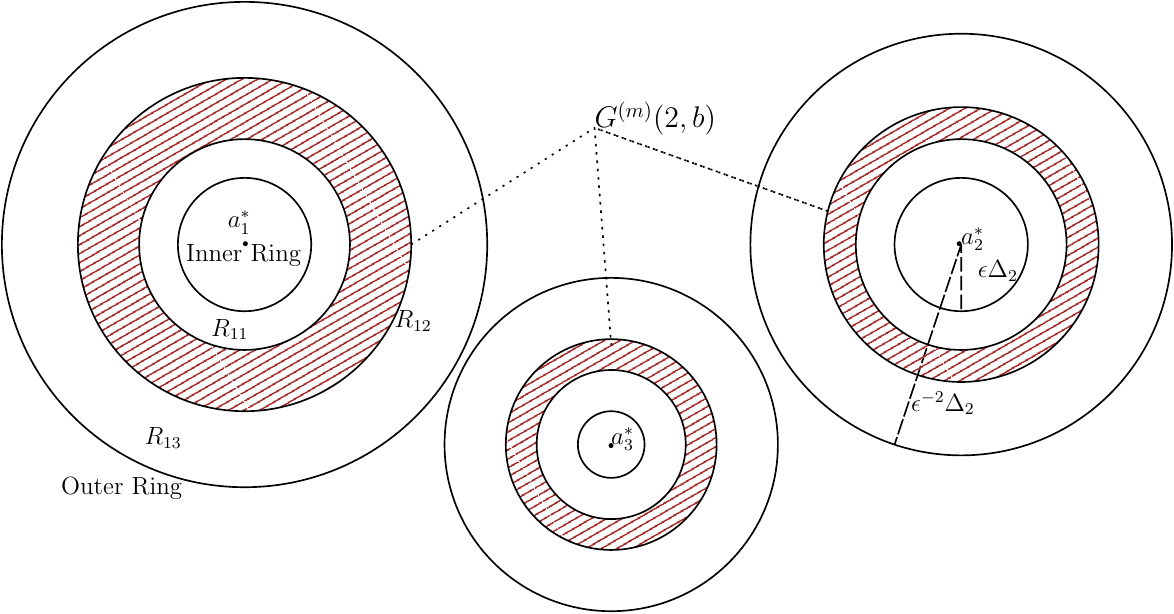}
		\caption{An example of Definition \ref{def:group}}
		\label{fig:group}
	\end{figure}

	\begin{lemma}[\bf{Group number}]
		\label{lm:group_number}
		There exist at most $O(z^2\log(k\eps^{-1})\log (z\eps^{-1}))$ groups in $\calG$.
	\end{lemma}
	
	\noindent
	We also have the following lemma for the construction time of $\calG$.

	\begin{lemma}[\bf{Construction time of $\calG$}]
		\label{lm:construction_time}
		Given a constant approximation $A^\star\in \calC$, it takes $O(nk)$ time to construct $\calG$.
	\end{lemma}
	
	\begin{proof}
		Firstly, it takes $O(nk)$ time to compute all distances $d(p,A^\star)$ for $p\in P$ and construct clusters $P_i$s.
		Then it takes $O(n)$ time to compute all $\Delta_i$s and construct all rings and their corresponding cost $\cost_z(R_{ij}, A^\star)$ and $\cost_z(R^{(o)}_i, A^\star)$.
		Finally, it takes at most $O(n)$ time to construct all groups $G^{(m)}(j,b)$ and $G^{(o)}(b)$.
		Overall, we complete the proof.
	\end{proof}

	\subsection{A Coreset Construction Algorithm for \kzC~\cite{cohenaddad2021new,cohenaddad2022towards}}
	\label{sec:algorithm}

	Now we are ready to present our coreset construction algorithm
	for \kzC\ (see Algorithm~\ref{alg:coreset}).
	We use the same importance sampling procedure proposed in~\cite{cohenaddad2021new,cohenaddad2022towards}.
	The key is to perform an importance sampling procedure for each group $G\in \calG$ with a carefully selected number of samples $\Gamma_G$ (Line 1), where $\Gamma_G$ is defined in Theorem~\ref{thm:coreset}.
	We also weigh each center $a^\star_i\in A^\star$ by the number of remaining points in $P_i\setminus \calG$ (We slightly abuse the notation $P_i\setminus \cal G$ to denote 
	$P_i\setminus (\cup_{G\in\calG} G)$. See Line 2).

Next, we provide a new analysis for Algorithm~\ref{alg:coreset} (Theorem~\ref{thm:coreset}).

	\begin{algorithm}[h]
		\caption{Coreset Construction Algorithm for \kzC~\cite{cohenaddad2021new,cohenaddad2022towards}}
		\label{alg:coreset}
		\begin{algorithmic}[1]
			\Require a metric space $(\calX, d)$, a dataset $P\subseteq \calX$, $z\geq 1$, integer $k\geq 1$;
			A constant approximation $A^\star\in \calC$ for $P$, and the partition $P_1,\ldots, P_k$ of $P$ according to $A^\star$; 
			The collection $\calG$ of groups as in Definition~\ref{def:group} together with
			the number of samples $\Gamma_G$ for each $G\in \calG$.
			\Ensure a weighted subset $(S,w)$  
			\State For each group $G\in \calG$, sample a collection $S_G$ of $\Gamma_G$ points from $G$, where each sample $p\in G$ is selected with probability $\frac{d^z(p,A^\star)}{\cost_z(G,A^\star)}$ and weighted by $w(p)=\frac{ \cost_z(G,A^\star)}{\Gamma_G\cdot d^z(p,A^\star)}$.
			\State For each $i\in [k]$, set the weight of $a^\star_i\in A^\star$ as $w(a^\star_i) = |P_i\setminus \calG|$.
			\State \Return $S:=A^\star\cup (\bigcup_{G\in \calG} S_G)$ together with the weights $w$.
		\end{algorithmic}
	\end{algorithm}

	\subsection{$\alpha$-Covering and Covering Number of Groups}
	\label{sec:covering}
	
	We first define $\alpha$-covering of groups in $\calG$.
	The new definition of  $\alpha$-covering for main groups is crucial for improving the coreset size.

	\paragraph{Coverings of main groups.}
	Recall that $A^\star = \left\{a^\star_1,\ldots,a^\star_k\in \calX\right\}\in \calC$ is a constant approximation for the \kzC\ problem.
	We define the following \emph{huge subset} of main groups $G\in \calG^{(m)}(j)$ w.r.t. a $k$-center set $C\in \calC$:
	\[
	H(G,C):= \left\{p\in R_{ij}\cap G: i\in [k], d^z(a^\star_i,C)\geq 8z \eps^{-1} 2^{j}\Delta_i\right\},
	\]
	i.e., the collection of rings $R_{ij}\in G$ with $d^z(a^\star_i,C)\geq 8z \eps^{-1} 2^{j} \Delta_i$.
        Let $H(S,C) := H(G,C)\cap S$ for every (weighted) subset $S\subseteq G$.
	By construction, for all points $p\in H(G,C)$, the distances $d(p,C)$ are ``close'' to each other, which is an important property. We summarize the property in the following observation.
	\begin{observation}
		\label{ob:main_huge}
		For a $k$-center set $C\in \calC$, $i\in [k]$ and $p\in P_i\cap H(G,C)$, we have 
		$d^z(p,C)\in (1\pm \frac{\eps}{4})\cdot d^z(a^\star_i,C)$.
	\end{observation}

    \noindent
	By this observation, we can use a single distance $d^z(a^\star,C)$ to approximate $d^z(p,C)$ with a small error for all points $p\in H(G,C)$.
	For points in $G\setminus H(G,C)$, we define the following notions of coverings and covering numbers of main groups.

	\begin{definition}[\bf{$\alpha$-Coverings of main groups}]
		\label{def:covering_main}
		Let $G\in \calG^{(m)}$ be a main group. 
		Let $S\subseteq G$ be a subset and $\alpha > 0$. 
		%
		We say a set $V\subset \R^{|S|}$ of \emph{cost vectors} is an \emph{$\alpha$-covering} of $S$ if for each $C\in \calC$, there exists a cost vector $v\in V$ such that for any $i\in [k]$ and $p\in P_i\cap S\setminus H(G,C)$, the following holds:
		\begin{align}
		\label{eq:covering}
		\left|d^z(p,C)-d^z(a^\star_i,C)-v_p\right|\leq \alpha\cdot \err(p,C),
		\end{align}
		where $\err(p,C)$ is called the \emph{relative covering error} of $p$ to $C$ defined as follows
		\[
		\err(p,C) := \left(\sqrt{d^z(p,C)\cdot d^z(p,A^\star)}+d^z(p,A^\star)\right) \cdot \sqrt{\frac{\ZGCA{G}}{\cost_z(G,A^\star)}},
		\]
		where we use $\ZGCA{G}$ as a shorthand notation of $\cost_z(G,C) +\cost_z(G,A^\star)$ throughout.
	\end{definition}
	
	\begin{definition}[\bf{Covering numbers of main groups}]
		\label{def:coveringnumber_main}
		Define $\mN(S,\alpha)$ to be the minimum cardinality $|V|$ of any $\alpha$-covering $V$ of $S$.
		Given an integer $\Gamma \geq 1$, define the \emph{$(\Gamma,\alpha)$-covering number} of $G$ to be
		\[
		\mN_G(\Gamma,\alpha) := \max_{S\subseteq G:|S|\leq \Gamma} \mN(S,\alpha),
		\]
		i.e., the maximum cardinality $\mN(S,\alpha)$ over all possible subsets $S\subseteq G$ of size at most $\Gamma$.
	\end{definition}
	
	\noindent
	Intuitively, the covering number is a complexity measure of the set of the all possible distance difference vectors $\left\{d^z(p,C)-d^z(a^\star_i,C)\right\}_{p\in G}$.
	As $\mN_G(\Gamma,\alpha)$ becomes larger, all possible center sets in $\calC$ can induce more types of such vectors (up to the relative covering error).

    By the triangle inequality, we directly have for any $C\in \calC$, $i\in [k]$ and $p\in P_i\cap G\setminus H(G,C)$
        \begin{align}
            \label{eq:range_general}
            \begin{aligned}
            & \quad d^z(a^\star_i,C)- d^z(p, C) \\
            \leq & \quad (d^z(a^\star_i,C) - d(p,C))\cdot z(d^{z-1}(p,C) + d^{z-1}(a^\star_i,C)) &\\
            \leq & \quad 3z d(p,A^\star)\cdot (8z\eps^{-1}\cdot 2^j\Delta_i)^{z-1} & (p\notin H(G,C)) \\
            \leq & \quad (8z)^z \eps^{-z+1} d^z(p,A^\star),\quad & (p\in R_{ij}).
            \end{aligned}
        \end{align}
		Hence, for any $\alpha\geq (8z)^z \eps^{-z+1}$, $\mN_G(\Gamma,\alpha) = 1$ always holds by Definition~\ref{def:covering_main} and we can trivially use $0^{|S_G|}$ as an $\alpha$-covering of $S_G$.
		Consequently, we only need to consider the range $0< \alpha\leq (8z)^z \eps^{-z+1}$.

	\paragraph{Coverings of outer groups.} 
	Similarly, we define coverings of an outer group $G\in \calG^{(o)}$.
	The same as~\cite{cohenaddad2022towards}, we define the following \emph{far subset} of $G$:
	\[
	F(G,C):= \left\{p\in R_i^{(o)}\cap G: i\in [k], \exists q\in R_i^{(o)}\cap G \text{ with } d(q,C)\geq 4\cdot d(q,A^\star)\right\},
	\]
	i.e., the collection of rings $R_i^{(o)}\in G$ that contain at least one point $q$ with $d(q,C)\geq 4\cdot d(q,A^\star)$.
	Note that for any $p\in G\setminus F(G,C)$, we have $\frac{d(p,C)}{d(p,A^\star)}\leq 4$.
	%
	%
	Again, we formally define coverings and covering numbers of outer groups, which is very similar to those for main groups.

	\begin{definition}[\bf{Coverings and covering numbers of outer groups}]
		\label{def:covering_outer}
		Let $G\in \calG^{(o)}$ be an outer group. 
		%
		%
		Let $S\subseteq G$ be a subset and $\alpha > 0$. 
		We say a set $V\subset \R^{|S|}$ of \emph{cost vectors} is an \emph{$\alpha$-covering} of $S$ if for each $C\in \calC$, there exists a cost vector $v\in V$ such that for any $p\in S\setminus F(G,C)$, the following inequality holds:
		\[
		|d^z(p,C)-v_p|\leq \alpha\cdot (d^z(p,C)+d^z(p,A^\star)).
		\]
		Define $\oN(S,\alpha)$ to be the minimum cardinality $|V|$ of an arbitrary $\alpha$-covering $V$ of $S$.
		Given an integer $\Gamma \geq 1$, define the \emph{$(\Gamma,\alpha)$-covering number} of $G$ to be
		\[
		\oN_G(\Gamma,\alpha) := \max_{S\subseteq G:|S|\leq \Gamma} \oN(S,\alpha),
		\]
		i.e., the maximum cardinality $\oN(S,\alpha)$ over all possible subsets $S\subseteq G$ of size at most $\Gamma$.
	\end{definition}
	
	\noindent
	%
	Our definition for outer groups is almost the same as \cite[Definition 3]{cohenaddad2022towards}, except that we consider coverings for subsets $S\subseteq G$ instead of $G$, which enables us to design a one-stage sampling algorithm.

	\subsection{The Main Theorem for Algorithm~\ref{alg:coreset}}
	\label{sec:main_theorem}
	
	Now we are ready to state the main theorem for Algorithm~\ref{alg:coreset}.

	\begin{theorem}[\bf{Coreset for \kzC}]
		\label{thm:coreset}
		\sloppy
		Let $(\calX, d)$ be a metric space and $P\subseteq \calX$ be a set of 
		$n$ points.
		Let integer $k\geq 1$ and $\eps\in (0,1)$ be the precision parameter.
		We define the number $\Gamma_G$ of samples for group $G$ 
		(in Algorithm~\ref{alg:coreset}) as follows:
		\begin{enumerate}
			\item For each main group $G\in \calG^{(m)}$, let $\Gamma_G$ be the smallest integer satisfying that
			\begin{align}
			\label{eq:main_sample}
			\Gamma_G\geq O\left(\eps^{-2} \left(\int_{0}^{+\infty} \sqrt{\log \mN_G(\Gamma_G, \alpha)} d\alpha \right)^2 + k\eps^{-2}\log (k\eps^{-1})\right);
			\end{align}
			\item For each outer group $G\in \calG^{(o)}$, let $\Gamma_G$ be the smallest integer satisfying that
			\begin{align}
			\label{eq:outer_sample}
			\Gamma_G\geq O\left(\eps^{-2} \left(\int_{0}^{+\infty} \sqrt{\log \oN_G(\Gamma_G, \alpha)} d\alpha \right)^2 + k\eps^{-2}\log (k\eps^{-1})\right);
			\end{align}
		\end{enumerate}
		With probability at least 0.9, Algorithm~\ref{alg:coreset} outputs an $O(\eps)$-coreset $(S,w)$ of size 
		$
		k+ 2^{O(z)}\sum_{G\in \calG} \Gamma_G,
		$
		and the running time is $O(nk)$.
	\end{theorem}
	
	\noindent
	$\Gamma_G$ appears in the familiar form of the entropy integral
	(or Dudley integral) commonly used in the chaining argument (see e.g., ~\cite[Corollary 5.25]{Handel2014ProbabilityIH}).
	Hence, the remaining task is to upper bound the entropy integrals. 
	%
	%
	Moreover, combining with Lemma~\ref{lm:construction_time}, we know that the coreset construction time is at most $O(nk)$ time (for building the rings and groups), given 
	a constant factor approximation $A^\star$.

	\noindent
	%
	Compared to Theorem~\ref{thm:coreset}, the coreset size for general $z\geq 1$ has an additional multiplicative factor $2^{O(z)}$. Typically, we treat $z$ as a constant
	and thus this additional factor is also a constant.

	\subsection{Proof of Theorem~\ref{thm:coreset}: Performance Analysis of Algorithm~\ref{alg:coreset}}	
	\label{sec:proof_coreset}
	
		The coreset size is directly from Line 1 of Algorithm~\ref{alg:coreset}.
	Thus, it remains to verify that $S$ is an $O(\eps)$-coreset.
	%
	%
	
	The key is the following lemma that summarizes the estimation error induced by each group $G\in \calG$ and the remaining points $P\setminus \calG$.
	For a group $G\in \calG$, we define $P^{G}:= \left\{p\in P: \exists i\in [k], p\in P_i \text{ and } P_i\cap G\neq \emptyset\right\}$ to be the union of all clusters $P_i$ that intersects with $G$. Recall that $\ZGCA{G}= \cost_z(G,C)+\cost_z(G,A^\star)$.

	\begin{lemma}[\bf{Error analysis of groups and remaining points}]
		\label{lm:key}
		We have the following:
		\begin{enumerate}
			\item For each main group $G\in \calG^{(m)}$, we have
			\[
			\Exp_{S_G} \sup_{C\in \calC} \left[\frac{1}{\ZGCA{G}}\cdot 
			\Bigl|\sum_{p\in S_G} w(p)\cdot d^z(p,C) - \cost_z(G,C) \Bigr| \right] \leq 3\eps.
			\]
			\item For each outer group $G\in \calG^{(o)}$, we have
			\[
			\Exp_{S_G} \sup_{C\in \calC} 
			\Bigl[\frac{1}{\ZGCA{P^G}}\cdot \left|\sum_{p\in S_G} w(p)\cdot d^z(p,C) - \cost_z(G,C) \Bigr| \right] \leq 2\eps.
			\]
			\item For any center set $C\in \calC$, we have
			\[
			\Bigl|\cost_z(P\setminus \calG, C) - \sum_{i\in [k]} w(a^\star_i) \cdot d^z(a^\star_i, C) \Bigr|\leq \eps\cdot \ZGCA{P}.
			\]
		\end{enumerate}
		Here, the expectations above are taken over the randomness of the sample set $S_G$. 
	\end{lemma}
	
	%
	%
	%
	
	\noindent
	Note that the normalization term of a main group $G\in \calG^{(m)}$ is $\frac{1}{\ZGCA{G}}$ that only relates to $G$, but the normalization term of an outer group $G\in \calG^{(o)}$ is $\frac{1}{\ZGCA{P^G}}$ that depends on a larger subset $P^G$.
	This is because ``remote'' points in outer groups $G\in \calG^{(o)}$ introduce an even larger empirical error than $\ZGCA{G}$, but the total number of these remote points is small compared to $|P^G|$, which enables us to upper bound the empirical error by $\eps\cdot \ZGCA{P^G}$.

	The proof of Lemma~\ref{lm:key} requires the chaining argument and the improved bound of variance
	and is deferred to Appendix~\ref{sec:proof_key}.
	Theorem~\ref{thm:coreset} is a direct corollary of Lemma~\ref{lm:key}. The argument is almost the same as the proof of Theorem 4 in~\cite{cohenaddad2022towards}.
	Intuitively, we can verify that for any center set $C\in \calC$, the expected estimation error of $(S,w)$ is small, say 
	\[
	\Exp_{S} \sup_{C\in \calC} \left[\frac{1}{\ZGCA{P}}\cdot \Bigl|\sum_{p\in S} w(p)\cdot d^z(p,C) - \cost_z(P,C) \Bigr| \right]\leq O(\eps).
	\]
	Then we only need to apply the Markov inequality and use the fact that $A^\star$ is an $O(1)$-approximate solution.
	For completeness, we provide the proof in the following.
	
	\begin{proof}[of Theorem~\ref{thm:coreset}]
		We have
		\begin{align*}
		& \quad \Exp_{S} \sup_{C\in \calC} \left[\frac{1}{\ZGCA{P}}\cdot \left|\sum_{p\in S} w(p)\cdot d^z(p,C) - \cost_z(P,C) \right| \right]  \\
		\leq & \quad \Exp_{S} \sup_{C\in \calC}  \sum_{G\in \calG}  \left[\frac{1}{\ZGCA{P}}\cdot \left|\sum_{p\in S_G} w(p)\cdot d^z(p,C) - \cost_z(G,C) \right| \right]  \\
		& \quad + \sup_{C\in \calC} \left[\frac{1}{\ZGCA{P}}\cdot \left|\sum_{i\in [k]} w(a^\star_i)\cdot d^z(a^\star_i,C) - \cost_z(P\setminus G,C) \right| \right]  \\
		\leq & \quad \Exp_{S}  \sup_{C\in \calC}\sum_{G\in \calG}  \left[\frac{\ZGCA{G}}{\ZGCA{P}}\cdot \Exp_{S_G}\sup_{C\in \calC}\left|\frac{\sum_{p\in S_G} w(p)\cdot d^z(p,C) - \cost_z(G,C) }{\ZGCA{G}} \right| \right]  \\
		& \quad + \sup_{C\in \calC} \left[\frac{1}{\ZGCA{P}}\cdot \left|\sum_{i\in [k]} w(a^\star_i)\cdot d^z(a^\star_i,C) - \cost_z(P\setminus G,C) \right| \right] & \\
		\leq & \quad \Exp_{S} \sup_{C\in \calC} \sum_{G\in \calG} \left[\frac{\ZGCA{G}}{\ZGCA{P}}\cdot 3\eps \right]  + \eps \quad \quad (\text{Lemma~\ref{lm:key}})  \\
		\leq & \quad 4\eps. &
		\end{align*}
		By the Markov inequality, with probability at least 0.9, the following inequality holds
		\[
		\sup_{C\in \calC} \left[\frac{1}{\ZGCA{P}}\cdot \left|\sum_{p\in S} w(p)\cdot d^z(p,C) - \cost_z(P,C) \right| \right] \leq 40\eps.
		\]
		Then for any $C\in \calC$, we have
		\[
		\left|\sum_{p\in S} w(p)\cdot d^z(p,C) - \cost_z(P,C) \right| \leq 40\eps \cdot \ZGCA{P} \leq O(\eps)\cdot \cost_z(P,C),
		\]
		since $A^\star$ is an $O(1)$-approximate solution.
		This completes the proof.
	\end{proof}

	\subsection{Proof of Lemma~\ref{lm:key}: Error Analysis of Groups}
	\label{sec:proof_key}
	
	\paragraph{Item 1.}
	Item 1 is about main groups.
	Fix a main group $G\in \calG^{(m)}(j)$.
        Let $\Gamma'_G = 2^{O(z)}\Gamma_G\geq 2^{3z}\Gamma_G$ be the sample number of $G$.
	We reduce the left hand side to a Gaussian process and apply the chaining argument,
	as done in~\cite{cohenaddad2022towards}, but many details are very different, especially the contruction
	of the $\alpha$-covering.
	Define event $\xi_G$ to be for any $i\in [k]$,
	\begin{align}
	\label{eq:good_event}
		\sum_{p\in P_i\cap S_G} w(p) = \sum_{p\in P_i\cap S_G} \frac{ \cost_z(G,A^\star)}{\Gamma'_G\cdot d^z(p,A^\star)}\in (1\pm \eps)\cdot |P_i\cap G|,
	\end{align}
	Event $\xi_G$ implies that for each cluster $P_i$, the total sample weight of $P_i\cap S_G$ is close to the cardinality $|P_i\cap G|$.
	It is not difficult to show that $\xi_G$ happens with high probability via standard concentration inequality.
	
	\begin{lemma}[\bf{$\xi_G$ happens with high probability~\cite[Lemma 19]{cohenaddad2022towards}}]
		\label{lm:good_event}
		With probability at least $1-k\cdot \exp(-\eps^2\Gamma_G/9k)$, event $\xi_G$ happens.
	\end{lemma}
	
\color{black}
	
 \noindent
     We define the following \emph{tiny subset} of 
    main groups $G\in \calG^{(m)}(j)$ w.r.t. a $k$-center set $C\in \calC$:
	\[
	T(G,C):= \left\{p\in R_{ij}\cap G: i\in [k], d^z(a^\star_i,C)\leq 2^{j+1}\Delta_i\right\},
	\]
	i.e., the collection of rings $R_{ij}\in G$ with $d^z(a^\star_i,C)\leq 2^{j+1} \Delta_i$.
	Next, we define the following cost vectors for each $C\in \calC$.
	\begin{definition}[\bf{Cost vectors of main groups}]
		\label{def:split_main}
		Given a $k$-center set $C\in \calC$, we define the following cost vectors:
		\begin{itemize}
			\item Let $u^C\in \R_{\geq 0}^{|G|}$ be a cost vector satisfying that 1) for any $p\in H(G,C)$, $u^C_p := d^z(p,C)$; 2) for any $p\in G\setminus H(G,C)$, $u^C_p := 0$.
			\item Let $x^C\in \R_{\geq 0}^{|G|}$ be a cost vector satisfying that 1) for any $i\in [k]$ and $p\in P_i\cap G\setminus (H(G,C)\cup T(G,C))$, $x^C_p := d^z(a^\star_i,C) - (8z)^z \eps^{-z+1} 2^{j} \Delta_i$; 2) for any $p\in  H(G,C)\cup T(G,C)$, $x^C_p := 0$.
			\item Let $y^{C}\in \R_{\geq 0}^{|G|}$ be a cost vector satisfying that 1) for any $i\in [k]$ and $p\in P_i\cap G\setminus (H(G,C)\cup T(G,C))$, $y^{C}_p := d^z(p,C)- d^z(a^\star_i, C) + (8z)^z \eps^{-z+1} 2^{j} \Delta_i$; 2) for any $p\in T(G,C)$, $y^{C}_p := d^z(p,C)$; 3) for any $p\in H(G,C)$, $y^{C}_p := 0$.
		\end{itemize}
	\end{definition}
	
	\noindent
	The main difference from~\cite{cohenaddad2022towards} is that we create two cost vectors $x^C$ and $y^{C}$ for $G\setminus H(G,C)$ for variance reduction.
	By the definition of $T(G,C)$, we can check that $u^C,x^C,y^C\in \R_{\geq 0}^{|G|}$, and consequently, $\cost_z(G,C) = \|u^C\|_1 + \|x^C\|_1 + \|y^{C}\|_1$ always holds.
 \color{black}
	Hence, to prove Item 1, it suffices to upper bound the estimation error of $u^C$, $x^C$ and $y^{C}$ by $\eps$ respectively and we give the corresponding lemmas (Lemmas~\ref{lm:main_huge} to~\ref{lm:main_small_y}) in the following.

	For $u^C$, we have the following lemma from~\cite{cohenaddad2022towards} since $\Gamma_G\geq O(k\eps^{-2} \log k)$.

	\begin{lemma}[\bf{Estimation error of $u^C$~\cite[Lemma 15]{cohenaddad2022towards}}]
		\label{lm:main_huge}
		The following inequality holds:
		\[
		\Exp_{S_G} \sup_{C\in \calC} \left[\frac{1}{\ZGCA{G}} \left|\sum_{p\in S_G}w(p)\cdot u^C_p - \|u^C\|_1 \right|\right] \leq \eps.
		\]
	\end{lemma}
	
	\noindent
	For $x^C$, we have the following lemma whose proof idea is from that in~\cite[Lemma 15]{cohenaddad2022towards}.

	\begin{lemma}[\bf{Estimation error of $x^C$}]
		\label{lm:main_small_x}
		The following inequality holds:
		\[
		\Exp_{S_G} \sup_{C\in \calC} \left[\frac{1}{\ZGCA{G}} \left|\sum_{p\in S_G} w(p)\cdot x^C_p - \|x^C\|_1 \right|\right] \leq \eps.
		\]
	\end{lemma}

	\begin{proof}
		We have
		\begin{eqnarray}
		\label{ineq1:proof_main_small_x}
		\begin{aligned}
		& \Exp_{S_G} \sup_{C\in \calC} \left[\frac{1}{\ZGCA{G}} \left|\sum_{p\in S_G} w(p)\cdot x^C_p - \|x^C\|_1 \right|\right] \\
		= &  \Exp_{S_G} \sup_{C\in \calC} \left[\frac{1}{\ZGCA{G}} \left|\sum_{p\in S_G} w(p)\cdot x^C_p - \|x^C\|_1 \mid \xi_G\right|\right]\cdot \Pr\left[\xi_G\right] \\
		&  + \Exp_{S_G} \sup_{C\in \calC} \left[\frac{1}{\ZGCA{G}} \left|\sum_{p\in S_G} w(p)\cdot x^C_p - \|x^C\|_1 \mid \overline{\xi_G}\right|\right]\cdot \Pr\left[\overline{\xi_G}\right].
		\end{aligned}
		\end{eqnarray}

\color{black}
		For the first term on the right side, a trivial upper bound for $\Pr\left[\xi_G\right]$ is 1.
		For each $C\in \calC$, let $M_C\subseteq [k]$ be the collection of $i$ with $P_i\cap (H(G,C)\cup T(G,C)) = \emptyset$. 
		Assuming $\xi_G$ holds, we have that
		\begin{eqnarray}
		\label{ineq2:proof_main_small_x}
		\begin{aligned}
		\sum_{p\in S_G} w(p)\cdot x^C_p & =  \sum_{i\in M_C} \sum_{p\in P_i\cap S_G}  w(p)\cdot x^C_p & \\
		& =  \sum_{i\in M_C} (d^z(a^\star_i,C)-2^{j+1}\Delta_i)\cdot \sum_{p\in P_i\cap S_G}  w(p) & (\text{Definition~\ref{def:split_main}}) \\
		& \in  (1\pm \eps)\sum_{i\in M_C} (d^z(a^\star_i,C)-2^{j+1}\Delta_i)\cdot |P_i\cap G| & (\text{Defn. of $\xi_G$}) \\
		& \in  (1\pm \eps) \cdot \|x^C\|_1, & (\text{Definition~\ref{def:split_main}})
		\end{aligned}
		\end{eqnarray}
\color{black}
		which implies that
		\begin{eqnarray}
		\label{ineq3:proof_main_small_x}
		\begin{aligned}
		&  \quad\Exp_{S_G} \sup_{C\in \calC} \left[\frac{1}{\ZGCA{G}} \left|\sum_{p\in S_G} w(p)\cdot x^C_p - \|x^C\|_1 \right|\mid \xi_G\right]\cdot \Pr\left[\xi_G\right] & \\
		\leq & \quad \eps \cdot \Exp_{S_G} \sup_{C\in \calC} \left[\frac{1}{\ZGCA{G}} \|x^C\|_1 \right] & (\text{Ineq.~\eqref{ineq2:proof_main_small_x}}) \\
		\leq & \quad \eps. & (\text{triangle ineq.})
		\end{aligned}
		\end{eqnarray}
		On the other hand, assuming $\xi_G$ does not hold, we have $\Pr\left[\overline{\xi_G}\right] \leq k\cdot \exp(-\eps^2\Gamma_G/9k)$ by Lemma~\ref{lm:good_event}.
		Since $\Gamma_G\geq O(k\eps^{-2} \log k)$, we have $\Pr\left[\overline{\xi_G}\right] \leq \eps/4k$.
		Moreover,
		\begin{eqnarray}
		\label{ineq4:proof_main_small_x}
		\begin{aligned}
		\sum_{p\in S_G} w(p)\cdot x^C_p  = & \quad  \sum_{i\in M_C} \sum_{p\in P_i\cap S_G}  w(p)\cdot x^C_p & \\
		= & \quad  \sum_{i\in M_C} d^z(a^\star_i,C) \cdot \sum_{p\in P_i\cap S_G} \frac{\cost_z(G,A^\star)}{\Gamma_G\cdot d^z(p,A^\star)} & (\text{Defn. of $w(p)$}) \\
		\leq & \quad  \sum_{i\in M_C} d^z(a^\star_i,C) \cdot \sum_{p\in P_i\cap S_G} \frac{4k |P_i\cap G|}{\Gamma_G} & (\text{Observation~\ref{ob:group}}) \\
		\leq & \quad  4k \sum_{i\in M_C} d^z(a^\star_i,C)\cdot |P_i\cap G| & (|P_i\cap S_G|\leq \Gamma_G) \\
		= & \quad  4k \|x\|_1, &
		\end{aligned}
		\end{eqnarray}
		which implies that
		\begin{eqnarray}
		\label{ineq5:proof_main_small_x}
		\begin{aligned}
		& \quad  \Exp_{S_G} \sup_{C\in \calC} \left[\frac{1}{\ZGCA{G}} \left|\sum_{p\in S_G} w(p)\cdot x^C_p - \|x^C\|_1 \mid \overline{\xi_G}\right|\right]\cdot \Pr\left[\overline{\xi_G}\right] & \\
		\leq & \quad  4k \cdot \Exp_{S_G} \sup_{C\in \calC} \left[\frac{1}{\ZGCA{G}} \|x^C\|_1 \right]\cdot \Pr\left[\overline{\xi_G}\right] & (\text{Ineq.~\eqref{ineq4:proof_main_small_x}}) \\
		\leq & \quad  4k\cdot \Pr\left[\overline{\xi_G}\right] &\\
		\leq & \quad  4k\cdot \frac{\eps}{4k} &  \\
		= & \quad  \eps. 
		\end{aligned}
		\end{eqnarray}
		The lemma is a direct corollary of Inequalities~\eqref{ineq1:proof_main_small_x}, \eqref{ineq3:proof_main_small_x} and \eqref{ineq5:proof_main_small_x}.
	\end{proof}
	
	\noindent
	The main difficulty is to upper bound the estimation error for $y^{C}$ by $\eps$
	as in the following lemma. 

	\begin{lemma}[\bf{Estimation error of $y^{C}$}]
		\label{lm:main_small_y}
		The following inequality holds:
		\[
		\Exp_{S_G} \sup_{C\in \calC} \left[\frac{1}{\ZGCA{G}} \left|\sum_{p\in S_G} w(p)\cdot y^{C}_p - \|y^{C}\|_1 \right|\right] \leq \eps.
		\]
	\end{lemma}
	
	\noindent
	Adding the inequalities in the above lemmas, we directly conclude Item 1 since for any $C\in \calC$,
	\begin{align*}
	&  \left|\sum_{p\in S_G} w(p)\cdot d^z(p,C) - \cost_z(G,C) \right| \\
	\leq & \left|\sum_{p\in S_G}w(p)\cdot u^C_p - \|u^C\|_1 \right| + \left|\sum_{p\in S_G} w(p)\cdot x^C_p - \|x^C\|_1 \right|+ \left|\sum_{p\in S_G} w(p)\cdot y^{C}_p - \|y^{C}\|_1 \right|,
	\end{align*}
	which implies the first item of Lemma~\ref{lm:key}:
	\[
	\Exp_{S_G} \sup_{C\in \calC} \left[\frac{1}{\ZGCA{G}}\cdot \left|\sum_{p\in S_G} w(p)\cdot d^z(p,C) - \cost_z(G,C) \right| \right] \leq 3\eps.
	\]
	Hence, it remains to prove Lemma~\ref{lm:main_small_y}.
\color{black}
     For preparation, we define another notion of coverings and covering numbers of main groups based on Definition~\ref{def:split_main}. 

	\begin{definition}[\bf{$\alpha$-Layered-Coverings of main groups}]
		\label{def:covering_layered}
		Let $G\in \calG^{(m)}(j)$ be a main group. 
		Let $S\subseteq G$ be a subset and $\alpha > 0$. 
		We say a set $V\subset \R^{|S|}$ of \emph{cost vectors} is an \emph{$\alpha$-layered covering} of $S$ if for each $C\in \calC$, there exists a cost vector $v\in V$ such that 
        \begin{itemize}
        \item for any $p\in T(G,C)\cap S\setminus H(G,C)$, the following holds:
		\[
		\left|d^z(p,C)-v_p\right|\leq \sqrt{3}\alpha\cdot \err(p,C);
		\]
        \item for any $i\in [k]$ and $p\in P_i\cap S\setminus (H(G,C)\cup T(G,C))$, the following holds:
		\[
		\left|d^z(p,C)-d^z(a^\star_i,C) + (8z)^z \eps^{-z+1} 2^{j} \Delta_i - v_p\right|\leq \alpha\cdot \err(p,C);
		\]
        \end{itemize}
	where $\err(p,C)$ is the relative covering error defined in Definition~\ref{def:covering_main}.
	\end{definition}
	
	\begin{definition}[\bf{Layered covering numbers of main groups}]
		\label{def:coveringnumber_layered}
		Define $\myN(S,\alpha)$ to be the minimum cardinality $|V|$ of any $\alpha$-layered covering $V$ of $S$.
		Given an integer $\Gamma \geq 1$, define the \emph{$(\Gamma,\alpha)$-covering number} of $G$ to be
		\[
		\myN_G(\Gamma,\alpha) := \max_{S\subseteq G:|S|\leq \Gamma} \myN(S,\alpha),
		\]
		i.e., the maximum cardinality $\myN(S,\alpha)$ over all possible subsets $S\subseteq G$ of size at most $\Gamma$.
	\end{definition}
	
	\noindent
	Different from Definition~\ref{def:covering_main}, we focus on estimations for two parts of points in $S$: $T(G,C)\cap S$ and $S\setminus (H(G,C)\cup T(G,C))$.
        The layered covering number is a complexity measure of the set of the combinations of all possible distance vectors $\left\{d(p,C)\right\}_{p\in T(G,C)\cap S}$ for the first part and all possible distance difference vectors $\left\{d^z(p,C)-d^z(a^\star_i,C): i\in [k], p\in S\setminus (H(G,C)\cup T(G,C))\right\}$ for the second part.
        Also, note that for any $i\in [k]$ and $p\in P_i\cap T(G,C)\cap S\setminus H(G,C)$, we have
        \[
        d^z(p,C)\leq (d(p, a^\star_i) + d^z(a^\star_i,C))^z \leq (16z)^z \eps^{-z+1} d^z(p, A^\star).
        \]
        Since $\err(p, C)\geq \sqrt{d^z(p,C)\cdot d^z(p,A^\star)}$, we again only need to focus on the range $0\leq \alpha\leq (8z)^z \eps^{-z+1}$ for layered coverings.
	We have the following relations between two covering numbers.

 \begin{lemma}[\bf{Relations between two covering numbers}]
    \label{lm:covering_relation}
    Let $G\in \calG^{(m)}(j)$ be a main group. 
    Let $S\subseteq G$ be a subset and $\alpha > 0$. 
    We have
    \[
    \log \myN(S,\alpha) = O\left(\log \mN(S, \alpha) + zk \log (\alpha^{-1}\eps^{-1})\right).
    \]
    Moreover, for any integer $\Gamma\geq 1$, the following holds:
    \[
    \log \myN(\Gamma,\alpha) = O\left(\log \mN(\Gamma, \alpha) + zk \log (\alpha^{-1}\eps^{-1})\right).
    \]
 \end{lemma}

 \begin{proof}
    Let $V\subset \R^{|S|}$ be an $\alpha$-covering of $S$.
    We construct $V'\subset \R^{|S|}$ of $S$ as follows: for every $v\in V$, every subset $M\subseteq [k]$ and every integral vector $\lambda \in [0, z^{O(z)} \eps^{-z+1} \alpha^{-1}]^M$, we add the following vector $v^{(M,\lambda)}$ to $V'$:
    \begin{itemize}
        \item for $i\in M$ and $p\in P_i\cap S$, let $v^{(M,\lambda)}_p := v_p + \lambda_i \alpha 2^{j-1} \Delta_i$;
        \item for $i\in [k]\setminus M$ and $p\in P_i\cap S$, let $v^{(M,\lambda)}_p := v_p + 2^{j+1}\Delta_i$.
    \end{itemize}
    By construction, we know that
    \[
    \log |V'| \leq \log |V| + O(zk \log (\alpha^{-1} \eps^{-1})).
    \]
    Then it suffices to prove that $V'$ is an $\alpha$-layered covering of $S$.

    For any $k$-center set $C\subset \R^d$, let $v\in V$ be a vector satisfying Inequality~\eqref{eq:covering}.
    Let $M\subseteq [k]$ be the collection of indices $i$ with $R_{ij}\in T(G,C)$, which implies that $d^z(a^\star_i, C)\leq 2^{j+1} \Delta_i$.
    Then we can let $\lambda\in [0, z^{O(z)} \eps^{-z+1} \alpha^{-1}]^M$ be an integral vector such that for any $i\in M$, 
    \[
    |d^z(a^\star_i, C) - \lambda_i \alpha 2^{j-1} \Delta_i|\leq \alpha 2^{j-1} \Delta_i \leq 0.5\alpha\cdot d^z(p, A^\star).
    \]
    We can check the following properties for $v^{(M,\lambda)}\in V'$:
    \begin{itemize}
        \item for any $p\in T(G,C)\cap S$, we have
		\begin{align*}
		& \quad \left|d^z(p,C)-v^{(M,\lambda)}_p\right| &\\
            \leq & \quad |d(p,C) - d^z(a^\star_i,C) - v_p| + |d^z(a^\star_i,C) - (v^{(M,\lambda)}_p - v_p)| &  \\
            \leq & \quad \err(p,C) + |d^z(a^\star_i,C) - \lambda_i \alpha 2^{j-1} \Delta_i| & (\text{Defns. of $v$ and $v^{(M,\lambda)}$}) \\ 
            \leq & \quad \err(p,C) + 0.5 \alpha\cdot d^z(p, A^\star) & (\text{Defn. of $\lambda$})\\
            \leq & \quad \sqrt{3}\alpha\cdot \err(p,C); & (\err(p,C)\geq d(p, A^\star))
		\end{align*}
        \item for any $i\in [k]$ and $p\in P_i\cap S\setminus (H(G,C)\cup T(G,C))$, we have
		\[
		\left|d^z(p,C)-d^z(a^\star_i,C) + 2^{j+1}\Delta_i - v^{(M,\lambda)}_p\right| = \left|d^z(p,C)-d^z(a^\star_i,C) - v_p\right|\leq \alpha\cdot \err(p,C).
		\]
    \end{itemize}
    Thus, we conclude that $V'$ is an $\alpha$-layered covering of $S$, which completes the proof of Lemma~\ref{lm:covering_relation}.
 \end{proof}

\color{black}

\noindent
Now we are ready to prove Lemma~\ref{lm:main_small_y}.

	\begin{proof}[of Lemma~\ref{lm:main_small_y}]
		
		We first use the symmetrization trick to reduce the left-hand side to a Gaussian process, 
		and then apply a chaining argument.

		\paragraph{Reduction to a Gaussian process.}
		Let $\xi=\left\{\xi_p\sim N(0,1): p\in S_G\right\}$ be a collection of $\Gamma_G$ independent standard Gaussian random variables.
		Note that $S_G$ is drawn from the unbiased importance sampling algorithm (Line 1 of Algorithm~\ref{alg:coreset}), which implies that $\Exp_{S_G}\left[\sum_{p\in S_G} w(p)\cdot y_p^C\right] = \|y^C\|_1$ holds for any $C\in \calC$.
		By~\cite[Lemma 7.4]{Handel2014ProbabilityIH}, we have
		\begin{align*}
		& \Exp_{S_G} \sup_{C\in \calC} \left[\frac{1}{\ZGCA{G}} \left|\sum_{p\in S_G} w(p)\cdot y^{C}_p - \|y^{C}\|_1 \right|\right] \\
		\leq & \sqrt{2\pi}\cdot \Exp_{S_G, \xi} \sup_{C\in \calC} \left[\frac{1}{\ZGCA{G}} \left|\sum_{p\in S_G}\xi_p\cdot w(p)\cdot y^{C}_p \right|\right],
		\end{align*}
		where $\xi_p$s are i.i.d. standard Gaussian variables.
		Thus, it suffices to prove that for any (multi-)set $S_G\subseteq G$ of size $\Gamma_G$,
		\begin{align}
		\label{ineq:gaussian_process}
		\Exp_{\xi} \sup_{C\in \calC} \left[\frac{1}{\ZGCA{G}} \left|\sum_{p\in S_G}\xi_p\cdot w(p)\cdot y^{C}_p \right|\right]\leq \frac{\eps}{6}.
		\end{align}
		Note that we fix $S_G$ and the only randomness is from $\xi$ (hence we get a Gaussian process).

		For integer $h\geq 0$, let $V_h\subset \R^{\Gamma_G}_{\geq 0}$ be a $2^{-h}$-layered covering of $S_G$.
		For any $C\in \calC$, define $v^{C,h}\in V_h$ to be a vector satisfying that for each $p\in S_G$, 
		\[
		|y^C_p-v^{C,h}_p|\leq \sqrt{3}\cdot 2^{-h}\cdot \err(p,C) = 
		\sqrt{3}\cdot 2^{-h}\cdot \left(\sqrt{d^z(p,C)\cdot d^z(p,A^\star)}+d^z(p,A^\star)\right)\cdot \sqrt{\frac{\ZGCA{G}}{\cost_z(G,A^\star)}}.
		\]
		Now we consider the following estimator
		\[
		X_{C,h} := \frac{1}{\ZGCA{G}} \sum_{p\in H(S_G,C)} \xi_p\cdot w(p)\cdot (v^{C,h+1}_p-v^{C,h}_p),
		\]
		and let $X_C:= \sum_{h\geq 0} X_{C,h}$.
		Note that for any $p\in S_G\setminus H(S_G,C)$, we have $y^{C}_p = 0$ by definition.
		Thus, Inequality~\eqref{ineq:gaussian_process} is equivalent to the following inequality:
		\[
		\Exp_{\xi} \sup_{C\in \calC} |X_C|\leq \frac{\eps}{6},
		\]
		since $X_C$ is a separable process~\cite[Definition 5.22]{Handel2014ProbabilityIH} with $\lim_{h\rightarrow +\infty} v^{C,h}_p = y^{C}_p$.
		Since $X_C:= \sum_{h\geq 0} X_{C,h}$, it suffices to prove
		\begin{align}
		\label{ineq:chaining}
		\sum_{h\geq 0} \Exp_{\xi} \sup_{C\in \calC} |X_{C,h}|\leq \frac{\eps}{6}
		\end{align}
		In the following, we focus on proving Inequality~\eqref{ineq:chaining}.

		\paragraph{A chaining argument for Inequality~\eqref{ineq:chaining}.}
		Fix an integer $h\geq 0$ and we first show how to upper bound $\Exp_{\xi} \sup_{C\in \calC} |X_{C,h}|$ by a chaining argument.
		Note that each $X_{C,h}$ is a Gaussian variable. 
		The main idea is to upper bound the variance of each $X_{C,h}$, which leads to an upper bound for $\Exp_{\xi} \sup_{C\in \calC} |X_{C,h}|$ by the following lemma.
		
		\begin{lemma}[\bf{\cite[Lemma 2.3]{Massart2007ConcentrationIA}}]
			\label{lm:Gaussian_sup}
			Let $g_i\sim N(0,\sigma_i^2)$ for each $i\in [n]$ be Gaussian random variables (not need to be independent) and suppose $\sigma = \max_{i\in [n]} \sigma_i$.
			Then
			\[
			\Exp\left[\max_{i\in [n]} |g_i|\right]\leq 2\sigma\cdot \sqrt{2\ln n}.
			\]
		\end{lemma}
		Note that each $X_{C,h}$ is a Gaussian variable with mean 0.
		Hence, the key is upper bounding the variance of $X_{C,h}$ as in the following lemma.

		\begin{lemma}[\bf{Variance of $X_{C,h}$}]
			\label{lm:variance}
			Fix a $k$-center set $C\in \calC$ and an integer $h\geq 0$. 
			The variance of $X_{C,h}$ is always at most
			\[
			\Var\left[X_{C,h}\right]=\sum_{p\in S_G\setminus H(S_G,C)} \left(\frac{w(p)\cdot (v^{C,h+1}_p - v^{C,h}_p)}{\ZGCA{G}}\right)^2 \leq \frac{2^{-2h+4} k}{\Gamma_G}.
			\]
			Moreover, conditioned on event $\xi_G$ (Inequality~\eqref{eq:good_event}), the variance of $X_{C,h}$ is always at most
			\[
			\Var\left[X_{C,h}\mid \xi_G\right]=\sum_{p\in S_G\setminus H(S_G,C)} \left(\frac{w(p)\cdot (v^{C,h+1}_p - v^{C,h}_p)}{\ZGCA{G}}\right)^2 \leq \frac{2^{- 2h+6}}{\Gamma_G}.
			\]
		\end{lemma}
		
		\noindent
		The proof and analysis of the above lemma can be found in Appendix~\ref{sec:proof_variance}.
		Compared with~\cite[Lemma 23]{cohenaddad2022towards}, the above lemma provides a tighter variance bound,  saving a $\min\left\{\eps^{-z},k\right\}$ factor.
		%
		%

		Now we are ready to upper bound $\Exp_{\xi} \sup_{C\in \calC} |X_{C,h}|$.
		%
		By law of total expectation, we have
		\begin{align}
		\label{ineq1:proof_main_small_y}
		\Exp_{\xi} \sup_{C\in \calC} |X_{C,h}| = \Exp_{\xi} \left[\sup_{C\in \calC} |X_{C,h}| \mid \xi_G \right]\cdot \Pr\left[\xi_G\right] + \Exp_{\xi} \left[\sup_{C\in \calC} |X_{C,h}| \mid \overline{\xi_G} \right]\cdot \Pr\left[\overline{\xi_G}\right].
		\end{align}
		For each $C\in \calC$, we define $\psi^{(C,h)} = v^{C,h+1} - v^{C,h}$ as a certificate of $C$. 
		Note that there are at most $|V_{h+1}|\cdot |V_h|$ different certificates $\psi^{(C,h)}$s and let $\Psi_h = \left\{\psi^{(C,h)}: C\in \calC\right\}$ be the collection of all possible certificates $\psi^{(C,h)}$s.
		For each certificate $\psi\in \Psi$, we select a $k$-center set $C_\psi\in \calC$ with 
		\[
		C_\psi := \arg\min_{C\in \calC: \psi^{(C,h)} = \psi} \cost_z(G,C),
		\]
		to be the representative center set of $\psi$.
		By definition, we note that for any $C\in \calC$ with $\psi^{(C,h)} = \psi$,
		\begin{align}
		\label{eq:property_Cpsi}
		\begin{aligned}
		&  \left|\frac{1}{\ZGCA{G}} \sum_{p\in S_G\setminus H(S_G,C)} \xi_p\cdot w(p)\cdot (v^{C,h+1}_p-v^{C,h}_p) \right| \\
		\leq & \,\, \left|\frac{1}{\cost_z(G,C_\psi+A^\star)} \sum_{p\in S_G\setminus H(S_G,C_\psi)} \xi_p\cdot w(p)\cdot (v^{C_\psi,h+1}_p-v^{C_\psi,h}_p) \right|.
		\end{aligned}
		\end{align}
		Let $\mathcal{C}_h = \left\{C_\psi: \psi\in \Psi\right\}$ be the collection of all such $C_\psi$s.
		By construction, we have $|\mathcal{C}_h|\leq |V_{h+1}|\cdot |V_h|$.
		Now we have
		\begin{eqnarray}
		\label{ineq2:proof_main_small_y}
		\begin{aligned}
		&  \Exp_{\xi} \left[\sup_{C\in \calC} |X_{C,h}|\mid \xi_G\right]\cdot \Pr\left[\xi_G\right] \\
		\leq &\,\,\,  \Exp_{\xi} \left[\sup_{C\in \calC} \left|\frac{1}{\ZGCA{G}} \sum_{p\in S_G\setminus H(S_G,C)} \xi_p\cdot w(p)\cdot \left(v^{C,h+1}_p-v^{C,h}_p\right) \right| \mid \xi_G\right] \\
		\\
		= &\,\,\, \Exp_{\xi} \left[\sup_{C\in \mathcal{C}_h} \left|\frac{1}{\ZGCA{G}} \sum_{p\in S_G\setminus H(S_G,C)} \xi_p\cdot w(p)\cdot \left(v^{C,h+1}_p-v^{C,h}_p\right) \right| \mid \xi_G \right] \\
		\leq & \,\,\, 2\sqrt{\sup_{C\in \mathcal{C}_h} \Var\left[X_{C,h}\mid \xi_G\right]}\cdot \sqrt{2\ln(|V_{h+1}|\cdot |V_h|)} \\
		\leq & \,\,\, \frac{2^{6-h}}{\sqrt{\Gamma_G}} \cdot \sqrt{\ln |V_{h+1}|}, \\
		\end{aligned}
		\end{eqnarray}
		where the first inequality follows from the definition of $X_{C,h}$, the first equality holds since the supreme must achieve at some center set $C\in \mathcal{C}_h$ due to Inequality~\eqref{eq:property_Cpsi}, the second inequality follows from Lemma~\ref{lm:Gaussian_sup}, and the last inequality holds due to Lemma~\ref{lm:variance}.
		Furthermore, we have
		\begin{eqnarray}
		\label{ineq3:proof_main_small_y}
		\begin{aligned}
		&  \Exp_{\xi} \left[\sup_{C\in \calC} |X_{C,h}|\mid \overline{\xi_G}\right]\cdot \Pr\left[\overline{\xi_G}\right] \\
		\leq &  \,\,\, 2\sqrt{\sup_{C\in \mathcal{C}_h} \sum_{p\in S_G\setminus H(S_G,C)} \left(\frac{w(p)\cdot (v^{C,h+1}_p - v^{C,h}_p)}{\ZGCA{G}}\right)^2}\cdot \sqrt{2\ln(|V_{h+1}|\cdot |V_h|)} \cdot \Pr\left[\overline{\xi_G}\right]  \\
		\leq & \,\,\, \frac{2^{6-h} k}{\sqrt{\Gamma_G}} \cdot \sqrt{\ln |V_{h+1}|}\cdot \Pr\left[\overline{\xi_G}\right]  \\
		\leq & \,\,\,  \frac{2^{6-h} \eps}{\sqrt{\Gamma_G}} \cdot \sqrt{\ln |V_{h+1}|}, & \\
		\end{aligned}
		\end{eqnarray}
		where the first inequality follows from Inequality~\eqref{ineq1:proof_main_small_y}, the second inequality from Lemma~\ref{lm:variance}, and the last inequality from Lemma~\ref{lm:good_event}.

\color{black}
        Finally, we observe that
        \begin{align}
            \label{ineq4:proof_main_small_y}
            \begin{aligned}
            \Gamma_G \geq & \quad \Omega\left(\eps^{-2} \left(\int_{0}^{(8z)^z \eps^{-z+1}} \sqrt{\log \mN_G(\Gamma_G, \alpha)} d\alpha \right)^2 + k\eps^{-2}\log (k\eps^{-1})\right) & (\text{Defn. of $\Gamma_G$})\\
            \geq &\quad \Omega\left(\eps^{-2} \left(\int_{0}^{(8z)^z \eps^{-z+1}} \sqrt{\log \mN_G(\Gamma_G, \alpha) + k \log \alpha^{-1}} d\alpha \right)^2\right) & \\
            \geq & \quad \Omega(\int_{0}^{(8z)^z \eps^{-z+1}} \sqrt{\log \myN_G(\Gamma_G, \alpha)} d\alpha). & (\text{Lemma~\ref{lm:covering_relation}})\\
            \end{aligned}
        \end{align}
\color{black}
  
		Using Eq.~\eqref{ineq1:proof_main_small_y} and Ineqs.~\eqref{ineq2:proof_main_small_y}-\eqref{ineq4:proof_main_small_y}, we can conclude that
		\begin{align*}
		\sum_{h\geq 0} \Exp_{\xi} \sup_{C\in \calC} |X_{C,h}| 
		\leq & \,\,\, \sum_{h\geq 0} \frac{2^{7-h}}{\sqrt{\Gamma_G}} \cdot \sqrt{\ln |V_{h+1}|} & (\text{Eq.~\eqref{ineq1:proof_main_small_y}, Ineqs.~\eqref{ineq2:proof_main_small_y} and \eqref{ineq3:proof_main_small_y}})\\
		\leq & \,\,\, \frac{\eps}{16}\cdot\frac{\sum_{h\geq 0} 2^{-h}\cdot \sqrt{\ln |V_{h+1}|}}{\int_{0}^{(8z)^z \eps^{-z+1}} \sqrt{\log \myN_G(\Gamma_G, \alpha)} d\alpha} & (\text{Ineq.~\eqref{ineq4:proof_main_small_y}}) \\
		\leq & \,\,\, \frac{\eps}{6}, & (\text{Defn. of $V_h$})
		\end{align*}
		i.e., Inequality~\eqref{ineq:chaining} holds.
		Hence, we have completed the proof of Lemma~\ref{lm:main_small_y}.
	\end{proof}

	\paragraph{Item 2.}
	For Item 2 of Lemma~\ref{lm:key}, the proof is almost the same as that of \cite[Lemma 13]{cohenaddad2022towards}.
	The only difference is that we consider the coverings on $S_G$ instead of $G$.
	%
	%
	
	\paragraph{Item 3.}
	Item 3 of Lemma~\ref{lm:key} has been proved in~\cite{cohenaddad2021new,cohenaddad2022towards}; see e.g.,~\cite[Lemma 4]{cohenaddad2021new}.

	Overall, we complete the proof of Lemma~\ref{lm:key}.

	\subsection{Proof of Lemma~\ref{lm:variance}: A Tighter Variance Bound}
	\label{sec:proof_variance}
	
	We first present the proof of Lemma~\ref{lm:variance}, and then show that our relative covering error $\err(p,C)$ leads to almost tight variance
	upper bound.

	\begin{proof}[of Lemma~\ref{lm:variance}]
		Let $M_C\subseteq [k]$ denote the collection of $i\in [k]$ with $P_i\cap S_G\setminus H(S_G,C)\neq \emptyset$.
		By the definition of $H(S_G,C)$, we note that for any $i\in M_C$, $P_i\cap  S_G\setminus H(S_G,C) = P_i\cap S_G$.
		For any $i\in [k]$ and point $p\in P_i\cap G$, let $q = \arg\min_{p'\in P_i\cap G} d(p', C)$.
		Suppose $p,q\in R_{ij}$ for some $j$.
		We note that
		\begin{eqnarray}
		\label{ineq1:proof_variance}
		\begin{aligned}
		d^z(p,C) & \leq  (d(q,C)+d(p,A^\star) +d(q,A^\star))^z & (\text{triangle ineq.}) \\
		& \leq  (d(q,C) + 3 d(p,A^\star))^z & (\text{Defn. of $R_{ij}$}) \\
		& \leq  2^z d(q,C) + 6^z d^z(p,A^\star) & (\text{Lemma~\ref{lm:relaxed}}) \\
		& \leq  \frac{2^z\cost_z(P_i\cap G, C)}{|P_i\cap G|} + 6^z d^z(p,A^\star) & (\text{Defn. of $q$}).
		\end{aligned}
		\end{eqnarray}
		Then we can bound the variance as follows:
		\begin{eqnarray}
		\label{ineq2:proof_variance}
		\begin{aligned}
		&  \sum_{p\in S_G\setminus H(S_G,C)} \left(\frac{w(p)\cdot (v^{C,h+1}_p - v^{C,h}_p)}{\ZGCA{G}}\right)^2  \\
		= &  \,\,\, \sum_{p\in S_G\setminus H(S_G,C)} \left(\frac{w(p)\cdot (v^{C,h+1}_p - y^{C}_p + y^{C}_p - v^{C,h}_p)}{\ZGCA{G}}\right)^2  &\\
		\leq & \,\,\,  \sum_{p\in S_G\setminus H(S_G,C)} \frac{\left(w(p)\cdot (\sqrt{3}+1)\cdot 2^{-h} \cdot  (\sqrt{d^z(p,C)\cdot d^z(p,A^\star)}+d^z(p,A^\star))\right)^2}{(\ZGCA{G})\cdot \cost_z(G,A^\star)}   \\
		\leq & \,\,\, \frac{2^{-2h+3} }{\Gamma'_G}\cdot \sum_{p\in S_G\setminus H(S_G,C)} \frac{w(p)\cdot (d^z(p,C)+d^z(p,A^\star))}{\ZGCA{G}}
		& \\
		\leq & \,\,\, \frac{2^{-2h+3} }{\Gamma'_G}\cdot 
		\sum_{i\in M_C}\sum_{p\in P_i\cap S_G} 
		\frac{w(p)}{{\ZGCA{G}}}\cdot 
		\left(\frac{2^z \cost_z(P_i\cap G, C)}{|P_i\cap G|} 
		+6^z d(p,A^\star)\right)
		& \\
		\leq & \,\,\, \frac{2^{-2h+5} }{\Gamma_G}\cdot 
		\sum_{i\in M_C}\frac{1}{\ZGCA{G}} \cdot 
		\left(\frac{\cost_z(P_i\cap G, C) + \cost_z(P_i\cap G, A^\star) }{|P_i\cap G|}\right)\cdot
		\sum_{p\in P_i\cap S_G} w(p) & 
		\end{aligned}
		\end{eqnarray}
		where the first inequality holds due to the Defn. of $v^{C,h+1}_p$ and $v^{C,h}_p$,
		the second follows from the Definition of $w(p)$,
		the third from Inequality~\eqref{ineq1:proof_variance} and the definition of $M_C$,
		and the last from Observation~\ref{ob:group}.
		Recall that event $\xi_G$ is defined as 
		\begin{align}
		\label{eq:xi_G}
		\sum_{p\in P_i\cap S_G} w(p) = \sum_{p\in P_i\cap S_G} \frac{ \cost_z(G,A^\star)}{\Gamma'_G\cdot d^z(p,A^\star)}\in (1\pm \eps)\cdot |P_i\cap G|.
		\end{align}
		
		\noindent 
		Hence, conditioning on $\xi_G$ and continue the derivation, 
		we can see that 
		\begin{align*}
		&  \sum_{p\in S_G\setminus H(S_G,C)} \left(\frac{w(p)\cdot (v^{C,h+1}_p - v^{C,h}_p)}{\ZGCA{G}}\right)^2  \\
		\leq & \,\,\,  \frac{2^{-2h+5} }{\Gamma_G}\cdot \sum_{i\in M_C}
		\frac{(1+\eps) |P_i\cap G|}{\ZGCA{G}} \left(\frac{\cost_z(P_i\cap G, C) + \cost_z(P_i\cap G, A^\star) }{|P_i\cap G|}\right) & (\text{Eq.~\eqref{eq:xi_G}})\\
		\leq & \,\,\, \frac{2^{-2h+6} }{\Gamma_G}\cdot \sum_{i\in M_C}\frac{(\cost_z(P_i\cap G,C) + \cost_z(P_i\cap G,A^\star))}{\cost_z(G, C)+ \cost_z(G,A^\star)} & (\eps\in (0,1))\\
		\leq & \,\,\, \frac{2^{- 2h+6}}{\Gamma_G} \cdot \frac{ \cost_z(G, C)+ \cost_z(G,A^\star)}{\ZGCA{G}}  &  \\
		\leq & \,\,\, \frac{2^{- 2h+6}}{\Gamma_G}. &
		\end{align*}

		\noindent	    
		In the general case without conditioning on $\xi_G$, 
		we also have
		\begin{align*}
		&  \sum_{p\in S_G\setminus H(S_G,C)} \left(\frac{w(p)\cdot (v^{C,h+1}_p - v^{C,h}_p)}{\ZGCA{G}}\right)^2 \\ 
		\leq & \,\,\, \frac{2^{-2h+2} }{\Gamma_G}\cdot 
		\sum_{i\in M_C}\sum_{p\in P_i\cap S_G} 
		\frac{w(p)}{\ZGCA{G}} \cdot \left(\frac{\cost_z(P_i\cap G, C) }{|P_i\cap G|}+4 d^z(p,A^\star)\right) \\
		= & \,\,\, \frac{2^{-2h+2} }{\Gamma_G}\cdot \sum_{i\in M_C}\sum_{p\in P_i\cap S_G} \frac{\cost_z(G,A^\star)} {\Gamma_G\cdot d^z(p,A^\star) \cdot \ZGCA{G}} \cdot \left(\frac{\cost_z(P_i\cap G, C) }{|P_i\cap G|}+4 d^z(p,A^\star)\right) \\
		\leq & \,\,\, \frac{2^{-2h+2} }{\Gamma_G}\cdot \sum_{i\in M_C}\sum_{p\in P_i\cap S_G} \frac{4k\cdot \cost_z(P_i\cap G, C) + 4\cost_z(G,A^\star) }{\Gamma_G\cdot \ZGCA{G}} \\
		\leq & \,\,\, \frac{2^{-2h+2} }{\Gamma_G}\cdot \sum_{i\in M_C}
		\frac{ 4k\cdot\cost_z(P_i\cap G,C) + 4\cost_z(G,A^\star)}{\ZGCA{G}} \\
		\leq & \,\,\, \frac{2^{-2h+4} k}{\Gamma_G}, &
		\end{align*}
		where the first inequality follows from Ineq.~\eqref{ineq2:proof_variance}, the first equality from the Definition of $w(p)$, the second inequality from Observation~\ref{ob:group}, and the third inequality is due to the fact that $|P_i\cap S_G|\leq \Gamma_G$.
		This completes the proof.
	\end{proof}
	
	\begin{remark}
		\label{remark:variance}
		Lemma~\ref{lm:variance} provides an upper bound 
		of $O(\frac{2^{- 2h}}{\Gamma_G})$ for  $\Var\left[X_{C,h}\mid \xi_G\right]$.
		One can verify that this bound is almost tight for a large collection of $C\in \calC$.
		Suppose $H(G,C) = \emptyset$.
		For any $i\in [k]$ and point $p\in P_i\cap G$, suppose $d(p,C)\geq 4 d(p,A^\star)$, we have $d(p,C) = \Omega\left(\frac{\cost_z(P_i\cap G, C)}{|P_i\cap G|}\right)$ by a reverse argument of Inequality~\eqref{ineq1:proof_variance}.
		By reversing the argument of the proof of Lemma~\ref{lm:variance}, we can verify that conditioning on $\xi_G$,
		\[
		\sum_{p\in S_G\setminus H(S_G,C)} \left(\frac{w(p)\cdot (v^{C,h+1}_p - v^{C,h}_p)}{\ZGCA{G}}\right)^2
		= \Omega\left(\frac{2^{- 2h}}{\Gamma_G}\right).
		\]
		The improvement here is mainly due to the new relative covering error $\err(p,C)$, which is smaller than that of~\cite{cohenaddad2021new,cohenaddad2022towards}.
		Choosing a smaller error may increase the covering number, however,
		the covering number may not increase proportionally. 
		We choose the best tradeoff between variance and covering number.
		%
	\end{remark}

	\section{Improved Coreset Size for Euclidean \kzC\ }
	\label{sec:Euclidean}
	
	In this section, we consider the coreset construction for Euclidean \kzC.
	Now, we state our 
	main theorem for Euclidean coreset as follows.

	\begin{theorem}[\bf{Coreset for Euclidean \kzC}]
		\label{thm:Euclidean}
		Let $\calX=\R^d$, $P\subset \R^d$, $\eps\in (0,1)$ and constant $z\geq 1$.
		\sloppy
		For each $G\in \calG$, let 
		$
		\Gamma_G = O\left(k^{\frac{2z+2}{z+2}}\eps^{-2} \log (k\eps^{-1})\log^4\eps^{-1}\right).
		$
		With probability at least 0.9, Algorithm~\ref{alg:coreset} outputs an $\eps$-coreset of $P$ for Euclidean \kzC\ of size 
		$$
		2^{O(z)} k^{\frac{2z+2}{z+2}}\eps^{-2} \log^2(k\eps^{-1}) \log^5\eps^{-1}
		=
		\tilde{O}_z \left(k^{\frac{2z+2}{z+2}}\eps^{-2}\right).
		$$
	\end{theorem}

	\subsection{Proof of Theorem~\ref{thm:Euclidean}: Coreset for Euclidean \kzC}
	\label{sec:Euclidean_kmedian}
Following the same reasoning as in~\cite{cohenaddad2022improved}, we can make the following assumption without loss of generality.

\begin{assumption}[\bf{Assumptions on Euclidean datasets}]
    \label{assumption:Euclidean_dataset}
Let $\eps\in (0,1)$. 
We can assume that the given dataset $P$ satisfies
\begin{itemize}
	\item The number of distinct points $\|P\|_0$ is at most $2^{O(z)}\cdot\poly(k\eps^{-1})$;
	\item The dimension $d=O(z^2\eps^{-2} \log \|P\|_0)$;
	\item $P$ is unweighted.
        \end{itemize}
\end{assumption}

\noindent
In light of Theorem~\ref{thm:coreset}, it remains to upper bound the covering number $\mN_G(\Gamma_G, \alpha)$ and $\oN_G(\Gamma_G, \alpha)$, as in the following lemma.

\begin{lemma}[\bf{Covering number in Euclidean metrics}]
		\label{lm:Euclidean_covering_number}
		For each $G\in \calG^{(m)}$, $0<\alpha \leq 1$, and integer $0\leq \beta < \log (z\eps^{-1})$,
		\[
		\log \mN_G(\Gamma_G, \alpha) = 2^{O(z)} k\cdot\min\left\{k^{\frac{z}{z+2}}\alpha^{-2}\log \|P\|_0 \log \eps^{-1}, d\right\}\cdot \log (\eps^{-1}\alpha^{-1}).
		\]
		For each $G\in \calG^{(o)}$ and any $0<\alpha \leq 1$,
		\[
		\log \oN_G(\Gamma_G, \alpha) = O(z^2 k\cdot\min\left\{\alpha^{-2}\log \|P\|_0, d\right\}\cdot \log \alpha^{-1}).
		\]
	\end{lemma}
	
	\noindent
	The proof can be found in Appendix~\ref{sec:proof_Euclidean_covering_number}.
	Theorem~\ref{thm:Euclidean} is a direct corollary of the above lemma.

	\begin{proof}[of Theorem~\ref{thm:Euclidean}]
		Fix a main group $G\in \calG^{(m)}(j)$ for some integer $z\log (\eps/z)<j<2z\log(z\eps^{-1})$.
        Recall that we only need to consider $0< \alpha\leq (8z)^z \eps^{-z+1}$.
		To apply Theorem~\ref{thm:coreset}, the key is to upper bound the entropy integral $\int_{0}^{(8z)^z \eps^{-z+1}} \sqrt{\log \mN_G(\Gamma_G, \alpha)} d\alpha$.
		We first have
		\begin{eqnarray}
		\label{ineq1:Euclidean_covering_number}
		\int_{0}^{(8z)^z \eps^{-z+1}} \sqrt{\log \mN_G(\Gamma_G, \alpha)} d\alpha = \int_{0}^{\eps} \sqrt{\log \mN_G(\Gamma_G, \alpha)} d\alpha + \int_{\eps}^{(8z)^z \eps^{-z+1}} \sqrt{\log \mN_G(\Gamma_G, \alpha)} d\alpha.
		\end{eqnarray}		
		Next, we apply Lemma~\ref{lm:Euclidean_covering_number} to upper bound the above two terms on the right side separately.

		For the first term, we have
		\begin{eqnarray}
		\label{ineq2:Euclidean_covering_number}
		\begin{aligned}
		\int_{0}^{\eps} \sqrt{\log \mN_G(\Gamma_G, \alpha)} d\alpha  
		=\,\,\, &  O(1)\cdot \int_{0}^{\eps} \sqrt{kd\log(\eps^{-1}\alpha^{-1})} d\alpha & (\text{Lemma~\ref{lm:Euclidean_covering_number}}) \\
		=\,\,\, &  O(\sqrt{k\eps^{-2} \log(k\eps^{-1})})\cdot \int_{0}^{\eps} \sqrt{\log(\eps^{-1}\alpha^{-1})} d\alpha & (\text{Asm.~\ref{assumption:Euclidean_dataset}})\\
		=\,\,\, &  O(\sqrt{k\eps^{-2} \log(k\eps^{-1})})\cdot O(\eps\sqrt{\log (\eps^{-1}})) & \\
		=\,\,\, &  O\left(\sqrt{k\log(k\eps^{-1}) \log (\eps^{-1})}\right). &
		\end{aligned}
		\end{eqnarray}
		For the second term, we have the following upper bound
		\begin{eqnarray}
		\label{ineq3:Euclidean_covering_number}
		\begin{aligned}
		&  \int_{\eps}^{(8z)^z \eps^{-z+1}} \sqrt{\log \mN_G(\Gamma_G, \alpha)} d\alpha & \\
		=\,\,\, &  O\left(\int_{\eps}^{(8z)^z \eps^{-z+1}} \sqrt{ k^{\frac{2z+2}{z+2}}\alpha^{-2}\log \|P\|_0\log \eps^{-1} \log(\eps^{-1}\alpha^{-1})} d\alpha\right) & (\text{Lemma~\ref{lm:Euclidean_covering_number}}) \\
		=\,\,\, &  O(\sqrt{k^{\frac{2z+2}{z+2}}\log (k\eps^{-1})\log\eps^{-1}})\cdot \int_{\eps}^{(8z)^z \eps^{-z+1}} \sqrt{\alpha^{-2} \log(\eps^{-1}\alpha^{-1})} d\alpha & (\text{letting $2^\beta = k^{\frac{1}{z+2}}$}) \\
		=\,\,\, &  O(\sqrt{k^{\frac{2z+2}{z+2}}\log (k\eps^{-1})\log \eps^{-1}})\cdot \int_{\eps}^{(8z)^z \eps^{-z+1}} \sqrt{\alpha^{-2} \log(\eps^{-2})} d\alpha & \\ 
		=\,\,\, &  O(\sqrt{k^{\frac{2z+2}{z+2}}\log (k\eps^{-1})\log\eps^{-1}})\cdot O(\sqrt{\log^3(\eps^{-1})}) & \\
		=\,\,\, &  O\left(\sqrt{k^{\frac{2z+2}{z+2}}\log (k\eps^{-1}) \log^4 (\eps^{-1})}\right), &
		\end{aligned}
		\end{eqnarray}
		
		\noindent
		Now we are ready to upper bound $\int_{0}^{(8z)^z \eps^{-z+1}} \sqrt{\log \mN_G(\Gamma_G, \alpha)} d\alpha$.
		Combining with Inequalities~\eqref{ineq1:Euclidean_covering_number} to~\eqref{ineq3:Euclidean_covering_number}, we conclude that 
		\begin{align*}
		&  \int_{0}^{(8z)^z \eps^{-z+1}} \sqrt{\log \mN_G(\Gamma_G, \alpha)} d\alpha & 
		(\text{Eq.~\eqref{ineq1:Euclidean_covering_number}}) \\
		\leq &  O\left(\sqrt{k\log(k\eps^{-1}) \log\eps^{-1}}\right)  + O\left(\sqrt{k^{\frac{2z+2}{z+2}}\log (k\eps^{-1}) \log^4 (\eps^{-1})}\right) & (\text{Ineq.~\eqref{ineq2:Euclidean_covering_number} and \eqref{ineq3:Euclidean_covering_number}}) \\
		\leq &  O\left(\sqrt{k^{\frac{2z+2}{z+2}}\log (k\eps^{-1}) \log^4 (\eps^{-1})}\right). & 
		\end{align*}
		Similarly, we can prove that for an outer group $G\in \calG^{(o)}$,
		\[
		\int_{0}^{(8z)^z \eps^{-z+1}} \sqrt{\log \oN_G(\Gamma_G, \alpha)} d\alpha = O\left(\sqrt{k\log (k\eps^{-1}) \log^4 (\eps^{-1})}\right).
		\]
		Consequently, we have that $\Gamma_G = O\left(k^{\frac{2z+2}{z+2}}\eps^{-2} \log (k\eps^{-1})\log^4(\eps^{-1})\right)$ satisfies Inequality~\eqref{eq:outer_sample}.
		This completes the proof.
		%
		%
		
	\end{proof}

	\subsection{Terminal Embedding with Additive Errors}
	\label{sec:notation_Euclidean}
	
	Before proving Lemma~\ref{lm:Euclidean_covering_number}, we introduce two types of terminal embeddings that are useful for dimension reduction.
	
	\paragraph{Terminal embedding.}
	Roughly speaking, a terminal embedding projects a point set $X\subseteq \R^d$ to a low-dimensional space while approximately preserving all pairwise distances between $X$ and $\R^d$.

	\begin{definition}[\bf{Terminal embedding}]
		\label{def:embedding}
		Let $\alpha\in (0,1)$ and $X\subseteq \R^d$ be a collection of $n$ points. 
		A mapping $f: \R^d\rightarrow \R^m$ is called an $\alpha$-terminal embedding of $X$ if for any $p\in X$ and $q\in \R^d$,
		\[
		d(p,q)\leq d(f(p),f(q))\leq (1+\alpha)\cdot d(p,q).
		\]
	\end{definition}
	
	\noindent
	We have the following recent result on terminal embedding.

	\begin{theorem}[\bf{Optimal terminal embedding~\cite{Narayanan2019OptimalTD,Cherapanamjeri2022TerminalEI}}]
		\label{thm:embedding}
		Let $\alpha\in (0,1)$ and $X\subseteq \R^d$ be a collection of $n$ points. 
		There exists an $\alpha$-terminal embedding $f$ with a target dimension $O(\alpha^{-2}\log n)$.
		Specifically, $f$ is constructed as an extension of a Johnson-Lindenstrauss (JL) transform
		with the following properties:
		\begin{enumerate}
			\item Let $g: X\rightarrow \R^{m-1}$ be a JL transform.
			\item For each $p\in X$, let $f(p) = (g(p),0)$.
			\item For each $q\in \R^d$, the mapping $f(q)\in \R^m$ satisfies that $d(p,q)\leq d(f(p),f(q))\leq (1+\alpha)\cdot d(p,q)$ for all $p\in X$.
		\end{enumerate}
	\end{theorem}
	
	\noindent \sloppy
	Accordingly, if $\|X\|_0=\poly(k\eps^{-1})$, there exists an $\alpha$-terminal embedding of target dimension $O\left(\alpha^{-2} \log(k\eps^{-1})\right)$.

	\paragraph{Additive terminal embedding.}
	Recall that we introduce a new notion of terminal embedding with additive error, in order to handle the covering number $\mN_G(\Gamma_G, \alpha)$ of a main group.
	We need to prove the following theorem.
	
	\begin{theorem}[\bf{Restatement of Theorem~\ref{thm:additive_terminal_main}}]
		\label{thm:additive_terminal}
		Let $\alpha\in (0,1)$, $r>0$, and $X\subset \R^d$ be a collection of $n$ points within a ball $B(0, r)$ with $0^d\in X$.
		There exists an $\alpha$-additive terminal embedding with a target dimension $O(\alpha^{-2}\log n)$.
	\end{theorem}
	
	\begin{proof}
		Let $m=O(\alpha^{-2}\log n)$.
		We first let $g:X\rightarrow \R^{m-1}$ be a JL transform and assume 
		that the distortion of $g$ is $1+\alpha$ over the points of $X$.
		For each $p\in X$, we let $f(p) = (g(p),0)$. 
		Then we show how to extend $f$ to an additive terminal embedding $f:\R^d\rightarrow \R^m$.
		For each $q\in B(0, 2r)$, by Theorem~\ref{thm:embedding}, we know that there exists a mapping $f(q)\in \R^m$ such that for any $p\in X$,
		\[
		d(p,q)\leq d(f(p),f(q)) \leq (1+\alpha)\cdot d(p,q) \leq d(p,q) + 8\alpha r,
		\]
		since $d(p,q)\leq 4r$.

		Thus, we only need to focus on points $q\in \R^d\setminus B(0,2r)$.
		By the property of JL transform, we have the following lemma.

		\begin{lemma}[\bf{Restatement of~\cite[Lemma 3.1]{Narayanan2019OptimalTD}}]
			\label{lm:JL}
			With probability at least $0.9$, for any $q\in \R^d$, there exists $q'\in \R^{m-1}$ such that $\|q'\|_2\leq \|q\|_2$ and 
			\[
			\forall p\in X, ~ |\langle g(p),q' \rangle - \langle p,q \rangle|\leq \alpha\cdot \|p\|_2\|q\|_2.
			\]
		\end{lemma}
		
		\noindent
		By Lemma~\ref{lm:JL}, there exists $q'\in \R^{m-1}$ such that $\|q'\|\leq \|q\|_2$ and for any $p\in X$,
		\begin{align}
		\label{ineq2:proof_JL}
		|\langle g(p), q' \rangle - \langle p,q \rangle|\leq \alpha\cdot \|p\|_2 \|q\|_2 \leq \alpha r \|q\|_2.
		\end{align}
		Construct $f(q) = (q', \sqrt{\|q\|_2 - \|q'\|_2})$,
		which is the same mapping as in~\cite{mahabadi2018nonlinear,Narayanan2019OptimalTD}.
		The construction implies that $d(f(0),f(q)) = \|f(q)\|_2 = \|q\|_2 = d(0,q)$.
		It remains to prove the correctness of $|d(p,q) - d(f(p),f(q))|\leq \alpha\cdot r$.
		By construction, we have
		\begin{align}
		\label{ineq3:proof_JL}
		d(p,q)^2 = \|p\|_2^2 + \|q\|^2 - 2\langle p, q \rangle,
		\end{align}
		and
		\begin{align}
		\label{ineq4:proof_JL}
		d(f(p),f(q))^2 = \|f(p)\|^2 + \|f(q)\|_2^2 - 2\langle g(p), q' \rangle.
		\end{align}
		Since $\|f(q)\|_2 = \|q\|_2$, we have
		\begin{eqnarray}
		\label{ineq5:proof_JL}
		\begin{aligned}
		| d(p,q)^2 - d(f(p),f(q))^2| 
		\leq &  \left|\|p\|_2^2 - \|f(p)\|_2^2\right| + 2|\langle g(p), q' \rangle - \langle p,q \rangle| & (\text{Ineqs.~\eqref{ineq3:proof_JL} and~\eqref{ineq4:proof_JL}}) \\
		\leq &  4\alpha r^2 + 2|\langle g(p), q' \rangle - \langle p,q \rangle| & (\text{$p\in P$ and JL}) \\
		\leq &  4\alpha r^2 + 2\alpha r\|q\|_2 & (\text{Ineq.~\eqref{ineq2:proof_JL}})
		\end{aligned}
		\end{eqnarray}
		Hence, we obtain that
		\begin{align*}
		| d(p,q) - d(f(p),f(q))| 
		= &  \frac{| d(p,q)^2 - d(f(p),f(q))^2|}{d(p,q) + d(f(p),f(q))} & \\
		\leq &  \frac{4\alpha r^2 + 2\alpha r \|q\|_2}{\|q\|_2 - \|p\|_2 + \|f(q)\|_2 - \|f(p)\|_2} & (\text{Ineq.~\eqref{ineq5:proof_JL} and triangle ineq.}) \\
		= &  \frac{4\alpha r^2 + 2\alpha r \|q\|_2}{2 \|q\|_2 - \|p\|_2 - \|f(p)\|_2} & (\|f(q)\|_2 = \|q\|_2) \\
		\leq &  \frac{4\alpha r^2 + 2\alpha r \|q\|_2}{2 \|q\|_2 - 3r} & (\text{JL}) \\
		\leq &  \frac{4\alpha r \|q\|_2}{ 0.5 \|q\|_2} & (\|q\|_2 > 2r) \\
		= &  8\alpha r. &
		\end{align*}
		Overall, we complete the proof.
	\end{proof}

	\subsection{Proof of Lemma~\ref{lm:Euclidean_covering_number}: Bounding the Covering Number in Euclidean Case}
	\label{sec:proof_Euclidean_covering_number}
	
	Recall that $\Gamma_G$ is the number of samples for each main/outer group $G\in \calG$ (Line 1 of Algorithm~\ref{alg:coreset}).
	By~\cite[Lemma 21]{cohenaddad2022towards}, we can obtain the following upper bound for $\oN_G(\Gamma_G, \alpha)$, which suffices
	for our purpose.
	\[
	\log \oN_G(\Gamma_G, \alpha) = O\left(z^2 k\cdot\min\left\{\alpha^{-2}\log \|P\|_0, d\right\}\cdot \log \alpha^{-1}\right).
	\]

	Now, we consider the main groups.
	Fix $\alpha\in (0,1)$ and let $G\in \calG^{(m)}(j)$ for some $j$ be a main group.
	In the following, we focus on proving the upper bound for the covering number $\mN_G(\Gamma_G, \alpha)$.
	We first have the following lemma that provides a common construction for coverings.

	\begin{lemma}[\bf{$\alpha$-Covering for Euclidean spaces}]
		\label{lm:covering_construction}
		Suppose $d\geq \log k$ and constant $z \geq 1$.
			Let $X\subset \R^d$ be a dataset that consists of $X_1\subset B(a^\star_1, r_1), X_1\subset B(a^\star_2, r_2), \ldots, X_t \subset B(a_t, r_t)$ for some $1\leq t\leq k$, $a^\star_1,\ldots, a^\star_t\in \R^d$ and $r_1,\ldots, r_t >0$.
			Let $\alpha\in (0,1)$ and $u\geq 1$.
			There exists an $\alpha$-covering $V\subset \R^{|X|}$ of $X$ with $\log |V| = O(z kd\log(u\alpha^{-1}))$, i.e., for any $k$-center set $C\in \calC$, there exists a cost vector $v\in V$  $i\in [t]$ with $d(a^\star_i, C)\leq u\cdot r_i$ and $p\in X_i$ such that
			\[
			|d^z(p,C) - d^z(a^\star_i,C) - v_p| \leq \alpha\cdot r_i^z.
			\]
	\end{lemma}
	
	\begin{proof}
		For each $i\in [t]$, take an $\frac{\alpha}{10 (u+2)^z}\cdot r_i$-net of the Euclidean ball $B(a^\star_i, 10 z u r_i)$.
			Since $d\geq \log k$ and $z = O(1)$, the union $\Lambda$ of these nets has size at most
			\[
			k\cdot \exp\left(O(zd\log(u\alpha^{-1}) )\right) = \exp\left(O(z d\cdot \log(z u\alpha^{-1}) )\right).
			\]
			For any $S\subseteq G$, we then define an $\alpha$-covering $V\subset \R^{|S|}$ of $S$ as follows: for any $k$-center set $C\subseteq \Lambda$, we construct a vector $v\in V$ in which the entry corresponding to point $p$ (say $p\in R_{ij}\cap S$) is 
			\[
			v_p = d^z(p,C) - d^z(a^\star_i,C).
			\]
			Obviously, we have $\log |V| = O(zkd\cdot \log (u \alpha^{-1}))$.
			It remains to verify that $V$ is indeed an $\alpha$-covering of $S$.
			For any $C=(c_1,\ldots,c_k)\in \calC$, let $C'=(c'_1,\ldots,c'_k)\in V$ such that $c'_i$ is the closest point of $c_i$ in $\Lambda$ ($i\in [k]$).
			Then for any $i\in [t]$ with $d^z(a^\star_i,C)\leq u\cdot r_i$ and $p\in X_i$, we have the following observation: for every $c_i\in C$,
			\begin{enumerate}
				\item If $d(p,c_i)\geq 4u\cdot r_i^z$, then $d(p,c'_i)\geq 2u\cdot r_i\geq d(p,C)$ since $d(p,C)\leq d^z(a^\star_i,C) + d(p,a^\star_i)$.
				\item If $d(p,c)\leq 4u\cdot r_i$, then $d(p,c')\in d(p,c)\pm \frac{\alpha}{10(u+2)^z}\cdot r_i$ by the construction of $\Lambda$.
			\end{enumerate}
			Consequently, we have
			\begin{align}
			\label{ineq1:proof_Euclidean_general}
			d(p,C') \in d(p,C)\pm \frac{\alpha}{10(u+2)^z}\cdot r_i,
			\end{align}
			which implies that
			\begin{align*}
			&  |d^z(p,C') - d^z(p,C)| & \\
			\leq &  |d(p,C') - d(p,C)|\cdot z (d^{z-1}(p,C') + d^{z-1}(p,C)) & \\
			\leq &  \frac{\alpha}{10(u+2)^z}\cdot r_i\cdot \left((d^z(a^\star_i,C') + r_i)^{z-1} + (d^z(a^\star_i,C) + r_i)^{z-1} \right) & (\text{triangle ineq. and Ineq.~\eqref{ineq1:proof_Euclidean_general}}) \\
			\leq &  \frac{\alpha}{10(u+2)^z}\cdot r_i\cdot 2(u+2)^{z-1} r_i^{z-1} & (d^z(a^\star_i,C) \leq u\cdot r_i) \\
			\leq &  \frac{\alpha}{2} \cdot r_i^z. \\
			\end{align*}
			The above two inequalities directly lead to the following conclusion:
			\[
			|d^z(p,C) - d^z(a^\star_i,C) - v_p| = |d^z(p,C) - d^z(a^\star_i,C) - (d^z(p,C')-d^z(a^\star_i,C'))|\leq \alpha\cdot r_i^z,
			\]
			which completes the proof.
	\end{proof}
	
	\noindent
	We go back to upper bound $\mN_G(\Gamma_G, \alpha)$.

	\paragraph{The first upper bound for $\mN_G(\Gamma_G, \alpha)$.}
	We first verify
	\begin{align}
	\label{ineq1:proof_Euclidean_covering_number}
	\log \mN_G(\Gamma_G, \alpha) = O(zkd \log (\eps^{-1}\alpha^{-1})).
	\end{align}
	This is actually a direct corollary of Lemma~\ref{lm:covering_construction} by letting $X = S$, $X_i = S\cap R_{ij}$, $r_i = 2^{j+1}\Delta_i$ for $i\in [k]$, and $u = 10z\eps^{-1}$.

	\paragraph{The second upper bound for $\mN_G(\Gamma_G, \alpha)$.}
	Next, we verify 
	\begin{align}
	\label{ineq:main_upper_bound_2}
	\log \mN_G(\Gamma_G, \alpha) \leq 2^{O(z)} k^{\frac{2z+2}{z+2}} \alpha^{-2} \log\|P\|_0 \log\eps^{-1} \log (\eps^{-1}\alpha^{-1})).
	\end{align}
	For preparation, we propose the following notion of partitions of $G\setminus H(G,C)$ w.r.t. center sets $C\in \calC$ according to the ratio $\frac{d^z(a^\star_i, C)}{\Delta_i}$.
	\begin{definition}[\bf{A partition of $G\setminus H(G,C)$}]
		\label{def:main_partition}
		Let $G\in \calG^{(m)}(j)$ for some integer $z\log (\eps/z)<j<2z\log(z\eps^{-1})$ be a main group.
		For a $k$-center set $C\in \calC$ and integer $\beta\geq 1$, we define
		\[	        
		\HGCB{G}:=\left\{p\in R_{ij}\cap G\setminus H(G,C): i\in [k], 2^{\beta z+j+1}\Delta_i\leq d^z(a^\star_i,C) < 2^{\beta z+j+2}\Delta_i \right\}.
		\]
		Also, define 
		\[
		\HGC{G}:=\left\{p\in R_{ij}\cap G\setminus H(G,C): i\in [k], d^z(a^\star_i,C) < 2^{j+2}\Delta_i \right\}.
		\]
		For any subset $S\subseteq G$ and $k$-center set $C\in \calC$, we define $\HGCB{S} = S\cap \HGCB{G}$ for integer $\beta\geq 0$.
	\end{definition}
	
	
	\begin{figure}
		\centering
		\includegraphics[width=0.93\textwidth]{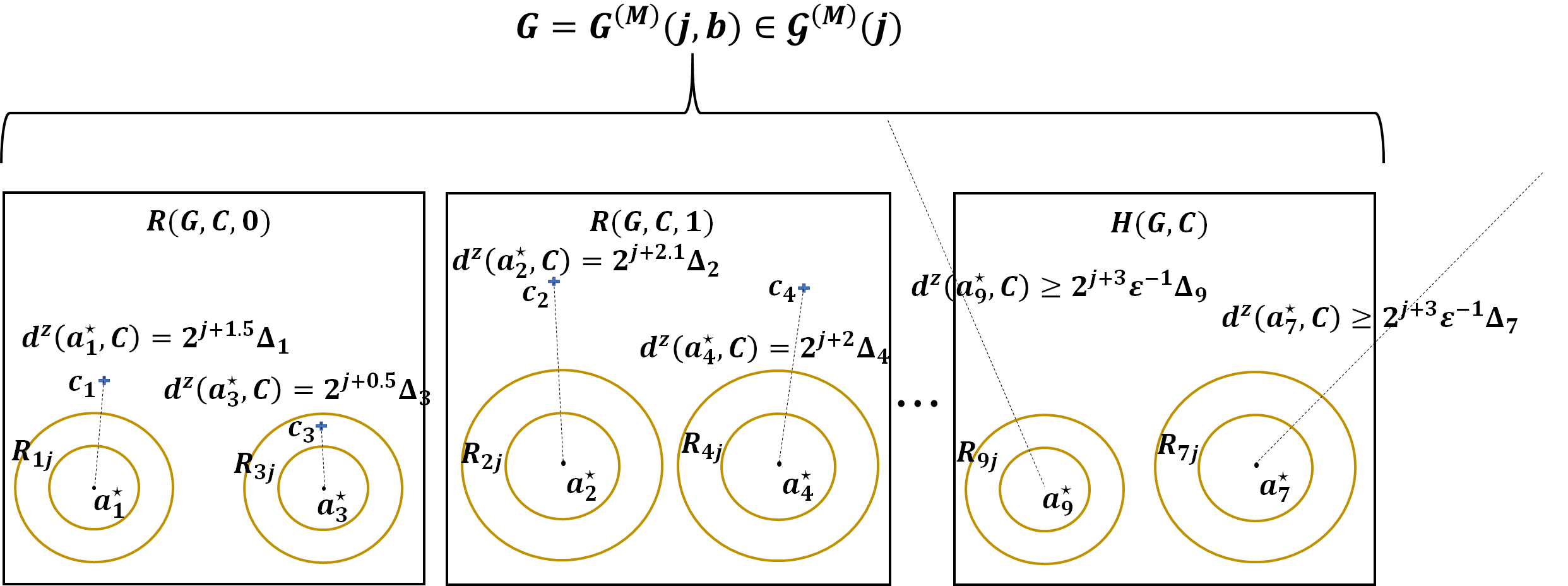}
		\caption{An example of Definition \ref{def:main_partition}}
		\label{fig:main_partition}
	\end{figure}
	\noindent
	By the above definition, for a main group $G\in \calG^{(m)}(j)$ and a given $k$-center set $C\in \calC$, $H(G,C)$ and all $\HGCB{G}$s $(0\leq \beta\leq 2+\log (z\eps^{-1}))$ are disjoint, and their union is exactly $G$; see Figure~\ref{fig:main_partition} for an example.
	Intuitively, for all points $p\in \HGCB{G}$, the fractions $\frac{d^z(p,C)}{d^z(p,A^\star)}$ are ``close'', which is an important property for our coreset construction.
	More concretely, combining with Definition~\ref{def:group} and the triangle inequality, we have the following observation.
	\begin{observation}[\bf{Relations between $d^z(p,C)$ and $d^z(p,A^\star)$ in partitions $\HGCB{G}$}]
		\label{ob:main_gap}
		For a $k$-center set $C\in \calC$, 
		\begin{enumerate}
			\item If $p\in \HGC{G}$, then $d^z(p,C)\leq 6^z d(p,A^\star)$;
			\item If $p\in \HGCB{G}$ for some integer $1\leq \beta \leq 2+\log(z\eps^{-1})$, then $(2^{\beta z}-1)\cdot d^z(p,A^\star)\leq d^z(p,C)\leq (2^{(\beta+2)z}+1)\cdot d^z(p,A^\star)$.
		\end{enumerate}
	\end{observation}

	\noindent
	For each $C\in \calC$, denote a certificate $\phi^{(C)}\in \left\{0,1,\ldots, 2+\lfloor\log \eps^{-1}\rfloor, +\infty\right\}^{|G|}$ as follows: for each $i\in [k]$ with ring $R_{ij}\in G$,
	\begin{itemize}
		\item If there exists some integer $0\leq \beta\leq 2+\log (z\eps^{-1})$ such that $R_{ij}\subseteq \HGCB{G}$, let $\phi^{(C)}_i = \beta$;
		\item Otherwise, let $\phi^{(C)}_i = +\infty$.
	\end{itemize}
	Intuitively, every entry $\phi^{(C)}_i$ reflects the distance between ring $R_{ij}$ and center set $C$.
	As $C$ is far away from points in $R_{ij}$, $\phi^{(C)}_i$ becomes larger.
	Note that there are at most $(3+\log (z\eps^{-1}))^k$ possible $\phi^{(C)}$s over all $C\in \calC$, since each $\phi^{(C)}_i$ has at most $3+\log (z\eps^{-1})$ different choices.
	Fix an arbitrary subset $S\subseteq G$ of size at most $\Gamma_G$ and a vector $\phi\in \left\{0,1,\ldots, 2+\lfloor\log \eps^{-1}\rfloor, +\infty\right\}^k$. 
	In the following, we investigate the covering number induced by all center sets $C\in \calC$ with $\phi^{(C)} = \phi$.
	For each integer $0\leq \beta\leq 2+\log(z\eps^{-1})$, let $S_{\phi,\beta}$ denote the collection of $R_{ij}\cap S$ with $\phi_i = \beta$.
	Note that $S_{\phi,\beta}$s are disjoint for all $\beta$, and the union of all $S_{\phi,\beta}$s and $S\cap H(G,C)$ is exactly $S$.
	Also note that for each $C\in \calC$ with $\phi^{(C)} = \phi$, we have $\HGCB{S} = S_{\phi,\beta}$ for any integer $0\leq \beta\leq 2+\log(z\eps^{-1})$.

    \begin{definition}[\bf{Covering of a main group for $S_{\phi,\beta}$}]
    \label{def:covering_beta}
    Let $G\in \calG^{(m)}(j)$ for some integer $z\log (\eps/z)<j<2z\log(z\eps^{-1})$ be a main group and let $S\subseteq G$ be a subset.
	Fix a vector $\phi\in \left\{0,1,\ldots, 2+\lfloor\log \eps^{-1}\rfloor, +\infty\right\}^k$ and an integer $0\leq \beta\leq 2+\log(z\eps^{-1})$.
	We say $V_{\phi,\beta}\subset \R^{|S_{\phi,\beta}|}$ is an $\alpha$-covering of $S_{\phi,\beta}$ if for each $C\in \calC$ with $\phi^{(C)} = \phi$, there exists a cost vector $v\in V_{\phi,\beta}$ such that for any $i\in [k]$ and $p\in R_{ij}\cap S_{\phi,\beta}$, Inequality~\eqref{eq:covering} holds, i.e.,
	\begin{align*}
		\left|d^z(p,C)-d^z(a^\star_i,C)-v_p\right|\leq \alpha\cdot (\sqrt{d^z(p,C)\cdot d^z(p,A^\star)}+d^z(p,A^\star))\cdot \sqrt{\frac{\ZGCA{G}}{\cost_z(G,A^\star)}}.
	\end{align*}
    \end{definition}

    \noindent
	We have the following lemma.

	\begin{lemma}[\bf An upper bound of the covering number for each $\beta$]
		\label{lm:beta_upper}
		There exists an $\alpha$-covering $V_{\phi,\beta}\subset \R^{|S_{\phi,\beta}|}$ of $S_{\phi,\beta}$ such that
		\[
			\log |V_{\phi,\beta}| \leq 2^{O(z)} k\cdot \min\left\{2^{\beta z}\alpha^{-2}\log (k\eps^{-1}), (1+k 2^{-2\beta})\alpha^{-2}\log (k\eps^{-1})\right\}\cdot \log (\eps^{-1}\alpha^{-1}).
		\]
	\end{lemma}
	
	\noindent
	By this lemma, we can directly construct an $\alpha$-covering $V$ of $S$ as follows:
	\begin{enumerate}
		\item For each $\phi\in \left\{0,1,\ldots, 2+\lfloor\log \eps^{-1}\rfloor, +\infty\right\}^k$, construct $V_\phi$ to be the collection of cost vectors $v\in \R^{\Gamma_G}$ satisfying that 1) for each integer $0\leq \beta\leq 2+\log(z\eps^{-1})$, $v|_{S_{\phi,\beta}}\in V_{\phi,\beta}$; 2) for $p\in S\cap H(G,C)$, $v_p=0$.
		\item Let $V$ be the union of all $V_\phi$s.
	\end{enumerate} 
	By construction, we can upper bound the size of $V$ by
	\begin{align*}
	& \log |V| &\\
	\leq &  \log (3+\log(z\eps^{-1}))^k + \sum_{0\leq \beta\leq 2+\log(z\eps^{-1})} \log |V_{\phi,\beta}| &\\
	\leq &  O(k\log \log\eps^{-1}) & \\
    & + \sum_{0\leq \beta\leq 2+\log(z\eps^{-1})} O(k\cdot \min\left\{2^{\beta z}\alpha^{-2}\log \|P\|_0, (1+k 2^{-2\beta})\alpha^{-2}\log \|P\|_0\right\} \log (\eps^{-1}\alpha^{-1})) &\\
	\leq &  O(k\log \log\eps^{-1}) + \sum_{0\leq \beta\leq 2+\log(z\eps^{-1})} O(k^{\frac{2z+2}{z+2}}\alpha^{-2}\log \|P\|_0 \log (\eps^{-1}\alpha^{-1})) \\
	\leq &  O\left(k^{\frac{2z+2}{z+2}}\alpha^{-2}\log \|P\|_0 \log\eps^{-1} \log (\eps^{-1}\alpha^{-1})\right),
	\end{align*}
	where the second inequality holds due to Lemma~\ref{lm:beta_upper} and the third inequality follows from letting $2^\beta = k^{\frac{1}{z+2}}$ and Assumption~\ref{assumption:Euclidean_dataset}.
	This completes the proof of Inequality~\eqref{ineq:main_upper_bound_2}.
	Hence, it remains to prove Lemma~\ref{lm:beta_upper}.

	\begin{proof}[of Lemma~\ref{lm:beta_upper}]
		
		\paragraph{Case $\beta = 0$.} 
		Let $f$ be an $2^{-5z} \alpha$-terminal embedding of $A^\star\cup S$ into $m=2^{10 z} \alpha^{-2} \log \|P\|_0$ dimensions given by Theorem~\ref{thm:embedding}.
			Given a $k$-center set $C\in \calC$ with $\phi^{(C)}=\phi$, we have that for any $p\in A\cup S_{\phi,0}$,
			\begin{eqnarray}
			\label{ineq2:proof_Euclidean_covering_number}
			\begin{aligned}
			d^z(f(p),f(C)) & =  \min_{q\in C} d^z(f(p),f(q)) & \\
			& \in  (1\pm \alpha 2^{-5z} )^z\cdot \min_{q\in C} d^z(p,q) & (\text{Theorem~\ref{thm:embedding}}) \\
			& \in  (1\pm \alpha 2^{-3z} )\cdot d^z(p,C) & \\
			& \in  d^z(p,C)\pm \alpha \cdot d^z(p,A^\star). & (\text{Observation~\ref{ob:main_gap}})
			\end{aligned}
			\end{eqnarray}
		and
		\[
		d^z(f(a^\star_i), f(C)) \in d^z(a^\star_i, C)\pm O(\alpha)\cdot d^z(p,A^\star).
		\]
		Moreover, by the definition of $S_{\phi,0}$, we know that $d^z(f(p), f(C))\leq O(1)\cdot 2^j \Delta_i$.

		By applying Lemma~\ref{lm:covering_construction} with $d=m$, $X = S_{\phi,0}$, $X_i = R_{ij}\cap S_{\phi,0}$, $r_i = (2^{j+4}\Delta_i)^{1/z}$ for $i\in [k]$, and $u = O(1)$, we can construct an $\alpha$-covering $V_{\phi,0}\subset \R^{|S_{\phi,0}|}$ of $f(S)$ satisfying that
		\begin{enumerate}
			\item $|V_{\phi,0}| = \exp\left(O(km \log (u\alpha^{-1}))\right) = \exp\left(O(k\alpha^{-2} \log \|P\|_0 \log (\alpha^{-1}))\right)$;
			\item For any $C\in \calC$ with $\phi^{(C)} = \phi$, there exists $v\in V_{\phi,0}$ such that for any $i\in [k]$ and $p\in R_{ij}\cap S_{\phi,0}$,
			\begin{align}
			\label{ineq:embedding}
			|d^z(f(p), f(C)) - d^z(f(a^\star_i), f(C)) - v_p|\leq \alpha\cdot r_i \leq O(\alpha)\cdot d^z(p,A^\star).
			\end{align}
		\end{enumerate}
		It remains to verify that $V_{\phi,0}$ is an $O(\alpha)$-covering of $S$.
		For each $C\in \calC$ with $\phi^{(C)} = \phi$, we find a cost vector $v\in V_{\phi,0}$ satisfying Inequality~\eqref{ineq:embedding}.
		Consequently, we have for any $i\in [k]$ and $p\in R_{ij}\cap S_{\phi,0}$,
		\begin{align*}
		|d^z(p,C) - d^z(a^\star_i,C) - v_p| & \leq  |d^z(f(p),f(C)) - d^z(f(a^\star_i), f(C)) - v_p| + O(\alpha)\cdot d^z(p,A^\star) & (\text{Ineq.~\eqref{ineq2:proof_Euclidean_covering_number}}) \\
		& \leq  O(\alpha)\cdot d(p,A^\star), & (\text{Ineq.~\eqref{ineq:embedding}})
		\end{align*}
		which completes the proof.

		\paragraph{Case $\beta\geq 1$.} 
		We first construct an $\alpha$-covering $V_{\phi,\beta}$ of $S_{\phi,\beta}$ with 
			\begin{align}
			\label{ineq3:proof_Euclidean_covering_number}
			\log |V_{\phi,\beta}| = 2^{O(z)}k2^{\beta z}\alpha^{-2}\log \|P\|_0\cdot \log (\eps^{-1}\alpha^{-1}).
			\end{align}
			Let $f$ be an $2^{-5z-\beta z/2} \alpha$-terminal embedding of $A^\star\cup S_{\phi,\beta}$ into $m=2^{10 z + \beta z} \alpha^{-2} \log \|P\|_0$ dimensions given by Theorem~\ref{thm:embedding}.
			Given a $k$-center set $C\in \calC$ with $\phi^{(C)}=\phi$, we have that for any $i\in [k]$ and $p\in R_{ij}\cap S_{\phi,\beta}$,
\begin{align}
   \label{eq:case_beta}
                \begin{aligned}
			d^z(f(p),f(C)) & =  \min_{q\in C} d^z(f(p),f(q)) & \\
			& \in  (1\pm \alpha 2^{-5z - \beta z/2} )^z\cdot \min_{q\in C} d^z(p,q) & (\text{Theorem~\ref{thm:embedding}}) \\
			& \in  (1\pm \alpha 2^{-3z - \beta z/2} )\cdot d^z(p,C) & \\
			& \in  d^z(p,C)\pm O(\alpha) \cdot \sqrt{d^z(p,C)\cdot d^z(p,A^\star)}, & (\text{Observation~\ref{ob:main_gap}})
   \end{aligned}
\end{align}
			and
			\[
			d^z(f(a^\star_i),f(C)) \in d^z(a^\star_i,C)\pm O(\alpha) \cdot \sqrt{d^z(p,C)\cdot d^z(p,A^\star)}.
			\]
		Moreover, by the definition of $S_{\phi,\beta}$, we know that $d^z(f(p), f(C))\leq O(1)\cdot 2^{j+\beta z} \Delta_i$.

		By applying Lemma~\ref{lm:covering_construction} with $d=m$, $X = S_{\phi,\beta}$, $X_i = R_{ij}\cap S_{\phi,\beta}$, $r_i = (2^{j+4}\Delta_i)^{1/z}$ for $i\in [k]$, and $u = O(2^{\beta })$, we can construct an $\alpha$-covering $V_{\phi,\beta}\subset \R^{|S_{\phi,\beta}|}$ of $f(S)$ with 
		\[
		|V_{\phi,\beta}| = \exp\left(O(km \log (2^{\beta} \alpha^{-1}))\right) = \exp\left(O(k2^{\beta z}\alpha^{-2}\log \|P\|_0\cdot \log (\eps^{-1}\alpha^{-1}))\right);
		\]
		and for any $C\in \calC$ with $\phi^{(C)} = \phi$, there exists a cost vector $v\in V_{\phi,\beta}$ such that for any $i\in [k]$ and $p\in R_{ij}\cap S_{\phi,\beta}$,
		\[
		|d^z(f(p),f(C)) - d^z(f(a^\star_i),f(C)) - v_p|\leq \alpha\cdot r_i \leq O(\alpha)\cdot  d^z(p,A^\star).
		\]
		Combining with Inequality~\eqref{eq:case_beta}, we conclude that $V_{\phi,\beta}$ is also an $O(\alpha)$-covering of $S_{\phi,\beta}$, which completes the proof of Inequality~\eqref{ineq3:proof_Euclidean_covering_number}.

		Finally, we construct an $\alpha$-covering $V_{\phi,\beta}$ of $S_{\phi,\beta}$ with 
			\begin{align}
			\label{ineq4:proof_Euclidean_covering_number}
			\log |V_{\phi,\beta}| = 2^{O(z)}\cdot k^2 2^{-2\beta}\alpha^{-2}\log \|P\|_0\cdot \log (\eps^{-1}\alpha^{-1}).
			\end{align}
			We again let $B\subseteq [k]$ be the collection of $i\in [k]$ with $R_{ij}\cap S\subseteq S_{\phi,\beta}$.
			Note that for any $C\in \calC$ with $\phi^{(C)}=\phi$, 
			\begin{align}
			\label{ineq6:proof_Euclidean_covering_number_general}
			\frac{\cost_z(G,C)}{\cost_z(G,A^\star)} \geq \frac{|B| 2^{(\beta z - 1)}}{k}.
			\end{align}
			For each $i\in B$, let $g_i$ be an $\sqrt{|B|k^{-1}} 2^{\beta - 5z}\alpha$-additive terminal embedding of $T_i=\left\{p-a^\star_i: p\in R_{ij}\cap S\right\}$ into $m=O(|B|^{-1} k 2^{10z-2\beta}\alpha^{-2}\log \|P\|_0)$ dimensions given by Theorem~\ref{thm:additive_terminal}.

			Let $r = 2^{j+1}\Delta_i$.
			We again define a function $f_i$ as follows: $f_i(p) = g_i(p-a^\star_i)$ for any $p\in \R^d$.
			Recall that $G\in \calG^{(m)}(j)$, and hence, $R_{ij}\cap S\in B(a^\star_i, r)$.
			Then by Theorem~\ref{thm:additive_terminal}, we have that for any $p\in R_{ij}\cap S$ or $p = a^\star_i$ and $q\in \R^d$, $|d(p,q) - d(f_i(p),f_i(q))|\leq \sqrt{|B|k^{-1}} 2^{\beta - 5z}\alpha r$.
			Given a $k$-center set $C\in \calC$ with $\phi^{(C)}=\phi$, we have that for any $p\in R_{ij}\cap S$, $d^z(a^\star_i,C)\geq 2^{\beta+j+2}\Delta_i\geq 4r$.
			Consequently, we know that $C\subset \R^d\setminus B(a^\star_i, 2r)$.
			For any $p\in R_{ij}\cap S_{\phi,\beta}$ and $c_p = \arg\min_{c\in C} d(p,C)\in \R^d$, we have
			\begin{align*}
			&  d^z(f_i(p),f_i(c_p)) & \\
			\in &  (d(p,c) \pm \sqrt{|B|k^{-1}} 2^{\beta - 5z}\alpha r)^z & \\
			\in &  (1\pm 2^{1-5z} \sqrt{|B|k^{-1}} \alpha)^z d^z(p,c_p) & (p\in H(G,C)\cap S) \\
			\in &  (1\pm 2^{-4z} \sqrt{|B|k^{-1}} \alpha) d^z(p,c_p) & \\
			\in &  d^z(p,c_p) \pm \alpha \cdot \sqrt{d^z(p,c_p)\cdot d^z(p,A^\star)}\cdot \sqrt{\frac{\cost_z(G,C)}{\cost_z(G,A^\star)}} & (\text{Observation~\ref{ob:main_gap} and Ineq.~\eqref{ineq6:proof_Euclidean_covering_number_general}}) \\
			\in &  d^z(p,C) \pm \alpha \cdot \sqrt{d^z(p,C)\cdot d^z(p,A^\star)}\cdot \sqrt{\frac{\cost_z(G,C)}{\cost_z(G,A^\star)}} & (\text{Defn. of $c_p$})
			\end{align*}
			Moreover, we have for any $p\in R_{ij}\cap S_{\phi,\beta}$ and $c_p = \arg\min_{c\in C} d(p,C)\in \R^d$,
			\[
			|d^z(f_i(a^\star_i),f_i(C)) - d^z(a^\star_i,C)|\leq O(\alpha)\cdot \sqrt{d^z(p,C)\cdot d^z(p,A^\star)}\cdot \sqrt{\frac{\cost_z(G,C)}{\cost_z(G,A^\star)}}.
			\]
			Combining the above two inequalities, we directly have
			\begin{align}
   \label{eq:case_beta2}
   \begin{aligned}
			&  \left|(d^z(f_i(p),f_i(C)) - d^z(f_i(a^\star_i),f_i(C))) - (d^z(p,C) - d^z(a^\star_i, C)) \right| \\
			\leq &  O(\alpha) \cdot \sqrt{d^z(p,C)\cdot d^z(p,A^\star)}\cdot \sqrt{\frac{\cost_z(G,C)}{\cost_z(G,A^\star)}}.
   \end{aligned}
			\end{align}
		Moreover, by the definition of $S_{\phi,\beta}$, we know that $d^z(f_i(p), f_i(C))\leq O(1)\cdot 2^{j+\beta} \Delta_i$.
		By Equation~\eqref{ineq1:proof_Euclidean_covering_number}, we can construct an $\alpha$-covering $V_i\subset \R^{|R_{ij}\cap S_{\phi,\beta}|}$ of $f_i(R_{ij}\cap S_{\phi,\beta})$ with 
		\[
		|V_i| = \exp\left(O(km \log (\eps^{-1}\alpha^{-1}))\right) = \exp\left(O(|B|^{-1} k^2 2^{-2\beta}\alpha^{-2}\log \|P\|_0\cdot \log (\eps^{-1}\alpha^{-1}))\right)
		\]
		such that for any $C\in \calC$ with $\phi^{(C)} = \phi$, there exists a cost vector $v\in V_{\phi,\beta}$ such that for any $i\in [k]$ and $p\in R_{ij}\cap S_{\phi,\beta}$,
		\[
		|d^z(f_i(p),f_i(C)) - d^z(f_i(a^\star_i),f_i(C)) - v_p|\leq O(\alpha)\cdot 2^j \Delta_i \leq O(\alpha)\cdot  d^z(p,A^\star).
		\]
		Combining with Inequality~\eqref{eq:case_beta2}, we conclude that $V_i$ is also an $O(\alpha)$-covering of $R_{ij}\cap S_{\phi,\beta}$.
		Now we can construct an $\alpha$-covering $V_{\phi,\beta}\subset \R^{|S_{\phi,\beta}|}$ of $S_{\phi,\beta}$  
		as the Cartesian Product of all $V_i$s
		\begin{align}
		\label{ineq:size_Vi}
		V_{\phi,\beta} = \prod_{i\in B} V_i = \bigl\{(v_1, \ldots, v_{|B|}): v_i \in V_i \text{ for all } i\in B\bigr\}.
		\end{align}
		By the construction of $V_i$s, we know that for each $C\in \calC$ with $\phi^{(C)} = \phi$, there exists a cost vector $v^{(i)}\in \R^{|R_{ij}\cap S_{\phi,\beta}|}$ such that for any $p\in R_{ij}\cap S_{\phi,\beta}$, Inequality~\eqref{eq:covering} holds.
		Since $\prod_{i\in B} v^{(i)}\in V_{\phi,\beta}$ by construction, we have that $V_{\phi,\beta}$ is indeed an $\alpha$-covering of $S_{\phi,\beta}$.
		By Inequality~\eqref{ineq:size_Vi} and the construction of $V_{\phi,\beta}$, it is obvious that $|V_{\phi,\beta}|$ satisfies Inequality~\eqref{ineq4:proof_Euclidean_covering_number}.
		Overall, we complete the proof.
	\end{proof}

\end{document}